%% file: paper.tex
\documentclass{article}

\usepackage{arxiv}
\usepackage{authblk}
\usepackage[T1]{fontenc}   
\usepackage{amsthm}
\usepackage{amsmath}
\usepackage{bigints}
\usepackage{chngcntr}
\usepackage[english]{babel}
\usepackage{mathtools}
\usepackage{enumitem}
\usepackage{amssymb}
\usepackage{microtype}
\usepackage[table,usenames,dvipsnames]{xcolor}
\usepackage{booktabs}
\usepackage{multirow}
\usepackage{subcaption}
\usepackage{placeins}
\usepackage{bm}
\usepackage[authoryear]{natbib}
\usepackage[
  colorlinks,
  linkcolor=blue,
  citecolor=blue,
  urlcolor=blue,
  filecolor=blue,
  backref=page
]{hyperref}
\usepackage{graphicx}

\graphicspath{%
  {./},
  {./figures}
}

\citefix{} 

\NewCommandCopy{\oldbackref}{\backref}
\renewcommand*{\backref}[1]{
  \hspace{0.1em}\oldbackref{#1}
}

\newcommand\firstcolor[1]{\textbf{#1}}
\newcommand\secondcolor[1]{{\color{RoyalBlue}\textit{#1}}}

\numberwithin{equation}{section}

\theoremstyle{plain}
\newtheorem{theorem}{Theorem}
\newtheorem{proposition}[theorem]{Proposition}
\newtheorem{lemma}[theorem]{Lemma}
\newtheorem{corollary}[theorem]{Corollary}

\theoremstyle{definition}
\newtheorem{condition}[theorem]{Condition}
\newtheorem{definition}[theorem]{Definition}

\renewcommand{\epsilon}{\varepsilon}

\newcommand{\N}{\mathbb{N}}
\newcommand{\R}{\mathbb{R}}
\newcommand{\E}{\mathbb{E}}
\newcommand{\Hcal}{\mathcal{H}}
\newcommand{\tTheta}{\tilde{\Theta}}

\DeclareMathOperator{\supp} {supp}
\DeclareMathOperator{\KL} {KL}

\newcommand\dotprod[2]{\left\langle#1,#2\right\rangle}
\newcommand{\textoversim}[1]{\overset{\text{#1}}{\sim}}

\title{A Bayesian approach to functional regression: theory and computation}

\date{\today}

\newbox{\orcid}\sbox{\orcid}{\includegraphics[scale=0.06]{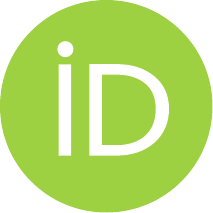}} 
\author[1,2]{%
	\href{https://orcid.org/0000-0003-0728-7748}{\usebox{\orcid}\hspace{1mm}José R.~Berrendero\thanks{\texttt{\href{mailto:joser.berrendero@uam.es}{joser.berrendero@uam.es}}}}%
}
\author[1]{%
	\href{https://orcid.org/0009-0004-7554-9193}{\usebox{\orcid}\hspace{1mm}Antonio Coín\thanks{\texttt{\href{mailto:antonio.coin@uam.es}{antonio.coin@uam.es}} (corresponding author)}}%
}
\author[1,2]{%
	\href{https://orcid.org/0000-0002-7993-0096}{\usebox{\orcid}\hspace{1mm}Antonio Cuevas\thanks{\texttt{\href{mailto:antonio.cuevas@uam.es}{antonio.cuevas@uam.es}}}}%
}
\affil[1]{Departamento de Matemáticas, Universidad Autónoma de Madrid (UAM), Madrid, Spain}
\affil[2]{Instituto de Ciencias Matemáticas ICMAT (CSIC-UAM-UC3M-UCM), Madrid, Spain}
\newcommand\shortauthor{J. R. Berrendero, A. Coín and A. Cuevas}

\hypersetup{
pdftitle={A Bayesian approach to functional regression: theory and computation},
pdfsubject={Preprint},
pdfauthor={José R.~Berrendero, Antonio Coín, Antonio Cuevas},
pdfkeywords={functional data analysis, functional regression, reproducing kernel Hilbert space, reversible jump MCMC, Bayesian inference, posterior consistency},
}

\begin{document}

\maketitle

\begin{abstract}
  We propose a novel Bayesian methodology for inference in functional linear and logistic regression models based on the theory of reproducing kernel Hilbert spaces (RKHS's). We introduce general models that build upon the RKHS generated by the covariance function of the underlying stochastic process, and whose formulation includes as particular cases all finite-dimensional models based on linear combinations of marginals of the process, which can collectively be seen as a dense subspace made of simple approximations. By imposing a suitable prior distribution on this dense functional space we can perform data-driven inference via standard Bayes methodology, estimating the posterior distribution through reversible jump Markov chain Monte Carlo methods. In this context, our contribution is two-fold. First, we derive theoretical results that guarantee strong posterior consistency and contraction at an optimal rate under mild conditions. Second, we show that several prediction strategies stemming from our Bayesian procedure are competitive against other usual alternatives in both simulations and real data sets, including a Bayesian-motivated variable selection method.
  \end{abstract}

\keywords{functional data analysis \and functional regression \and reproducing kernel Hilbert space \and Bayesian inference \and reversible jump MCMC \and posterior consistency}


\section{Introduction}\label{sec:intro}

The problem of predicting a scalar response from a functional covariate is one that has gained traction over the last few decades, as more and more real-world data is being generated with an ever-increasing level of granularity in the measurements. While in principle the functional data could simply be regarded as a discretized vector in a very high dimension, there are often many advantages in taking into account its functional nature, ranging from modeling the correlation among points that are close in the domain, to extracting information that may be hidden in the derivatives of the function in question. As a consequence, numerous proposals have arisen on how to suitably deal with functional data, all of them encompassed under the term Functional Data Analysis (FDA), which essentially explores statistical techniques to process, model and make inference on data varying over a continuum. A partial survey on such methods is \citet{cuevas2014partial} or \citet{goia2016introduction}, while a more detailed exposition of the theory and applications can be found for example in \citet{hsing2015theoretical} and \citet{horvath2012inference}.

FDA is undoubtedly an active area of research, which finds applications in a wide variety of fields, such as biomedicine, finance, meteorology or chemistry \citep{ullah2013applications}. Accordingly, there are many recent contributions on how to tackle functional data problems, both from a theoretical and practical standpoint. Chief among them is the approach of reducing the problem to a finite-dimensional one, for example using a truncated basis expansion or spline interpolation methods \citep{muller2005generalized, aguilera2013comparative}. At the same time, much effort has been put into building a sound theoretical basis for FDA, generalizing different frequentist concepts to the infinite-dimensional framework. Examples of this endeavor include the definition of centrality measures and depth-based notions for functional data \citep{lopez2009concept}, functional ANOVA tests \citep{cuevas2004anova}, a functional Mahalanobis distance \citep{galeano2015mahalanobis, berrendero2020mahalanobis}, or an extension of Fisher's discriminant analysis for random functions \citep{shin2008extension}, among many others. As the name suggests, FDA techniques are heavily inspired by functional analysis tools and methods: Hilbert spaces, linear operators, orthonormal bases, and so on. Incidentally, a notion that also intersects with the classical theory of machine learning and pattern recognition, and that has attained popularity in recent years, is that of reproducing kernel Hilbert spaces (RKHS's) and their applications in functional data problems \citep{kupresanin2010rkhs, yuan2010reproducing, berrendero2018use}. On the other hand, Bayesian inference methods are ubiquitous in the realm of statistics, and though they also make use of random functions, the approach is slightly different from the FDA case. While there are recent works that offer a Bayesian treatment of functional data \citep[e.g.][]{crainiceanu2010bayesian, shi2011gaussian}, there is still no systematic approach to Bayesian methodologies within FDA. It is precisely at this relatively unexplored intersection between FDA and Bayesian methods that our work is aimed.

In particular, our goal is to study functional regression problems, the infinite-dimensional equivalents of the typical statistical regression problems, from a Bayesian perspective. We follow the path started by \citet{ferguson1974prior} of setting a prior distribution on a functional space and using the corresponding posterior for inference. In our case, we use a particular RKHS as the ambient space, resulting in functional regression models that allow for a simple yet efficient Bayesian treatment of functional data. Moreover, we also study the basic theoretical question of posterior consistency in these RKHS models within the proposed Bayesian framework. Consistency and posterior concentration are a type of frequentist validation criteria that have arguably been an active point of research in the last few decades, particularly in infinite-dimensional settings \citep{amewou2003posterior, choi2008remarks}, and also in the functional regression case \citep{lian2016posterior,abraham2020posterior}. To put it simply, posterior consistency ensures that with enough data points, the Bayesian updating mechanism works as intended and the posterior distribution eventually concentrates around the true value of the parameters, supposing the model is well specified. We leverage the properties of RKHS's and certain extensions of classical results by \citet{doob1949application} and \citet{schwartz1965bayes} to show that posterior consistency and contraction at an optimal rate hold in our models with minimal conditions, thus providing a coherent background to our Bayesian approach.

Finally, this theoretical side is complemented by extensive experimentation that showcases the application of the functional RKHS models in various prediction tasks. Following recent trends in Bayesian computation techniques, the posterior distribution is approximated via Markov chain Monte Carlo (MCMC) methods, specifically the \textit{reversible jump} variant (RJMCMC) proposed by \citet{green1995reversible}. This computational work highlights the predictive performance of the proposed functional regression models, especially when compared with other usual frequentist methods.

\subsubsection*{\(L^2\)-models, shortcomings and alternatives}

In this work we are concerned with functional linear and logistic regression models, that is, situations where the goal is to predict a continuous or dichotomous variable from functional observations. Even though these problems can be formally stated with almost no differences from their finite-dimensional counterparts, there are some fundamental challenges as well as some subtle drawbacks that emerge as a result of working in infinite dimensions. To set a common framework, we will consider throughout a scalar response variable \(Y\) (either continuous or binary) which has some dependence on a stochastic \(L^2\)-process \(X=X(t)=X(t, \omega)\) with trajectories in \(L^2[0, 1]\), observed on a dense grid. We will further suppose for simplicity that \(X\) is centered, that is, its mean function \(m(t)=\E[X(t)]\) vanishes for all \(t\in[0,1]\). In addition, we will tacitly assume the existence of a labeled data set \(\mathcal D_n =\{(X_i, Y_i): i=1,\dots, n\}\) of independent observations from \((X, Y)\), and our aim will be to accurately predict the response corresponding to unlabeled samples from \(X\).

The most common scalar-on-function linear regression model is the classical \(L^2\)-model, widely popularized since the first edition (1997) of the pioneering monograph by~\citet{ramsay2005functional}. It can be seen as a generalization of the usual finite-dimensional model, replacing the scalar product in \(\R^d\) for that of the functional space \(L^2[0,1]\) (henceforth denoted by \(\dotprod{\cdot}{\cdot}\)):
\begin{equation}\label{eq:l2-linear-model}
  Y = \alpha + \dotprod{X}{\beta} + \epsilon = \alpha + \int_0^1 X(t)\beta(t)\, dt + \epsilon,
\end{equation}
where \(\alpha\in \R\), \(\epsilon\) is a random error term independent from \(X\) with \(\E [\epsilon]=0\), and the functional slope parameter \(\beta=\beta(\cdot)\) is a member of the infinite-dimensional space \(L^2[0, 1]\). In this case, the inference on \(\beta\) is hampered by the fact that \(L^2[0,1]\) is an extremely broad space that contains many non-smooth or ill-behaved functions, so any estimation procedure involving optimization on it will typically be hard. In spite of this, model~\eqref{eq:l2-linear-model} is not flexible enough to include ``simple'' impact points models based on linear combinations of the marginals of \(X\), such as \(Y=\alpha + \beta_1 X(t_1)+ \cdots + \beta_p X(t_p) + \epsilon\) for some constants \(\beta_j\in\R\) and instants \(t_j\in[0,1]\), which are especially appealing to practitioners confronted with functional data problems; see \citet{berrendero2024functional} for additional details on this. Moreover, the non-invertibility of the covariance operator associated with \(X\), which plays the role of the covariance matrix in the infinite case, invalidates the usual least squares theory \citep{cardot2011functional}. Thus, some regularization or dimensionality reduction technique is needed for parameter estimation; see \citet{reiss2017methods} for a summary of several widespread methods.

A similar \(L^2\)-based functional logistic equation can be derived for the binary classification problem via the logistic function:
\begin{equation}\label{eq:l2-logistic-model}
  \mathbb P(Y=1 \mid X) = \frac{1}{1 + \exp\{-\alpha - \dotprod{X}{\beta}\}},
\end{equation}
where \(\alpha \in \R\) and \(\beta \in L^2[0, 1]\). In the equivalent finite-dimensional problem the natural way of estimating the slope parameter is via its maximum likelihood estimator (MLE). However, apart from the issues outlined above for the linear model, functional settings pose an additional challenge in the logistic case: under fairly general conditions, the MLE does not exist with probability one \citep[see][]{berrendero2023functional}.

It turns out that in both scenarios a natural alternative to the \(L^2\)-model is the so-called reproducing kernel Hilbert space (RKHS) model, which instead assumes the unknown functional parameter \(\beta\) to be a member of the RKHS associated with the covariance function of the process \(X\). As we will show later on, not only is this model simpler and arguably easier to interpret, but it also constrains the parameter space to smoother and more manageable functions. In fact, it does include a model based on finite linear combinations of the marginals of \(X\) as a particular case, while also generalizing the aforementioned \(L^2\)-models under some conditions. These RKHS-based models and their idiosyncrasies have been explored in \citet{berrendero2019rkhs, berrendero2024functional} in the linear case, and in \citet{berrendero2023functional} in the logistic case.

A major aim of this work is to motivate these models within the functional framework, while also providing efficient techniques to apply them in practice. Our main contribution is the proposal of a Bayesian approach for inference in these RKHS models, in which a prior distribution is imposed on \(\beta\) to use the posterior probabilities for prediction. Although placing a prior distribution on a functional space is generally a hard task, the specific parametric formulation we propose greatly facilitates this. Similar Bayesian schemes have recently been explored in \citet{grollemund2019bayesian} and \citet{abraham2024informative}, albeit not in a RKHS setting. Another set of techniques extensively studied in this context are variable selection methods, which aim to select the marginals \(\{X(t_j)\}\) of the process that better summarize it according to some optimality criterion (see \citealp{ferraty2010most} or \citealp{berrendero2016variable} by way of illustration). As it happens, some RKHS-based variable selection methods have already been proposed \citep[e.g.][]{bueno2019variable}, but in general they have their own dedicated algorithms and procedures. As will shortly become apparent, our Bayesian methodology allows us to easily isolate the marginal posterior distribution corresponding to a finite set of points \(\{t_j\}\), providing a Bayesian variable selection process along with the other prediction methods that naturally arise from our approach.

\subsubsection*{Some essentials on RKHS's and notation}\label{sec:rkhs}

The methodology proposed in this work relies heavily on the use of RKHS's, so we will briefly outline the main characteristics of these spaces from a probabilistic point of view \citep[for a more detailed account, see][]{berlinet2004reproducing}. Let us denote by \(K(t, s)= \mathbb E[X(t)X(s)]\) the covariance function (the ``kernel'') of the centered process \(X\), and in what follows suppose that it is continuous. To construct the corresponding RKHS \(\Hcal(K)\), we start by defining the functional space \(\Hcal_0(K)\) of all finite linear combinations of evaluations of \(K\), that is,
\begin{equation}\label{eq:h0}
  \Hcal_0(K) = \left\{ f \in L^2[0,1]: \ f(\cdot) = \sum_{j=1}^p \beta_j K(t_j, \cdot), \ p \in \N, \ \beta_j \in \R, \ t_j \in [0, 1] \right\}.
\end{equation}
This space is endowed with the inner product \(\dotprod{f}{g}_K = \sum_{j, k} \beta_j \gamma_k K(t_j, s_k)\), given that \(f(\cdot)=\sum_j \beta_j K(t_j, \cdot) \) and \(g(\cdot)=\sum_k \gamma_k K(s_k, \cdot)\). Then, \(\Hcal(K)\) is defined to be the completion of \(\Hcal_0(K)\) under the norm induced by the scalar product \(\dotprod{\cdot}{\cdot}_K\). As it turns out, functions in this space satisfy the so-called \textit{reproducing property}: \(\dotprod{K(t, \cdot)}{f}_K = f(t)\) for all \(f \in \Hcal(K)\) and \(t \in [0, 1]\). An important consequence of this identity is that \(\Hcal(K)\) is a space of genuine functions and not of equivalence classes, since the values of the functions at specific points are in fact relevant, unlike in \(L^2\)-spaces.

Now, a particularly useful approach in statistics is to regard \(\Hcal(K)\) as an isometric copy of a well-known space related to \(X\). Specifically, via \textit{Loève's isometry} \citep{loeve1948fonctions} one can establish a congruence \(\Psi_X\) between \(\Hcal(K)\) and the linear span of the process, \(\mathcal L(X)\), in the space of all random variables with finite second moment, \(L^2(\Omega)\) \citep[see Lemma 1.1 in][]{lukic2001stochastic}. This isometry is essentially the completion of the correspondence
\begin{equation*}
  \sum_{j=1}^p \beta_j X(t_j) \longleftrightarrow \sum_{j=1}^p \beta_j K(t_j, \cdot),
\end{equation*}
and can be formally defined, in terms of its inverse, as \(\Psi^{-1}_X(U)(t) = \E[U X(t)]\) for \(U \in \mathcal L(X)\).
Despite the close connection between the process \(X\) and the space \(\Hcal(K)\), special care must be taken when dealing with concrete realizations, since in general the trajectories of \(X\) do not belong to the corresponding RKHS with probability one \citep[][Corollary~7.1]{lukic2001stochastic}. As a consequence, the expression \(\dotprod{x}{f}_K\) is ill-defined and lacks meaning when \(x\) is a realization of \(X\). However, following Parzen's approach in his seminal work \citep[][Theorem~4E]{parzen1961approach}, we can leverage Loève's isometry and identify \(\dotprod{X}{f}_K\) with the random variable in \(\mathcal L(X)\) associated with \(f\in \Hcal(K)\), so that \(\dotprod{x}{f}_K \) with \(x=X(\omega)\) is defined as the image \(\Psi_x(f) := \Psi_X(f)(\omega)\). This notation, viewed as a formal extension of the inner product, often proves to be useful and convenient.

\subsubsection*{Organization of the article}

The rest of the paper is organized as follows. In Section~\ref{sec:methodology} we explain the Bayesian methodology and the functional regression models we propose, including an overview of the reversible jump MCMC scheme. In Section~\ref{sec:consistency} we derive theoretical posterior consistency and posterior concentration results. The empirical findings of the experimentation are contained in Section~\ref{sec:results}, along with a short discussion of computational details. Lastly, the conclusions drawn from this work are presented in Section~\ref{sec:conclusion}. Additional details, proofs and results are included in Appendices~\ref{app:model-choice},~\ref{app:proofs},~\ref{app:experiments} and~\ref{app:source-code}.

\section{A Bayesian methodology for RKHS-based functional regression models}\label{sec:methodology}

In principle, the functional RKHS models contemplated in this work are those obtained by considering a functional parameter \(\beta \in \Hcal(K)\) and replacing the scalar product for \(\dotprod{X}{\beta}_K\) in the \(L^2\)-models~\eqref{eq:l2-linear-model} and~\eqref{eq:l2-logistic-model}, which has tangible benefits both in theory and practice on account of the RKHS properties. However, to further simplify things we will follow a parametric approach and suppose that \(\beta\) is in fact a member of the dense subspace \(\Hcal_0(K)\) defined in~\eqref{eq:h0}, i.e.:
\begin{equation}\label{eq:beta}
  \beta(\cdot) = \sum_{j=1}^p \beta_j K(t_j, \cdot).
\end{equation}
As we said before, with a slight abuse of notation we will understand the expression \(\dotprod{x}{\beta}_K\) as \(\Psi_x(\beta)\), where \(x=X(\omega)\) is a realization of \(X\) and \(\Psi_x\) is Loève's isometry. Hence, taking into account that \(\beta(\cdot)=\sum_j \beta_j K(t_j,\cdot)\) and that \(\Psi_X(K(t, \cdot)) = X(t)\) by definition, we can identify \(\dotprod{x}{\beta}_K\) with \(\sum_j \beta_j x(t_j)\). In addition, we will assume that \(X\) has a version with continuous sample paths, so that point evaluation is well-defined.

Although natural in this context, the crucial assumption that \(\beta \in \Hcal_0(K)\) may seem too restrictive at first glance, since we are essentially truncating the dimensionality of the model. Nevertheless, in this way we get a simpler, finite-dimensional approximation of the functional RKHS model, which we argue reduces the overall complexity while still capturing most of the relevant information. Plus, the simplified model remains ``truly functional'' in the sense that the number of components \(p\) and the time instants \(t_j\) are not fixed beforehand. In short, we are exploiting the RKHS perspective to give a functional nature to the collection of finite-dimensional models based on random linear combinations of the marginals, offering a unified treatment of all the so-called \textit{impact points} models (\citealp[see][]{lindquist2009logistic, kneip2016functional,poss2020superconsistent}).

In view of~\eqref{eq:beta}, to place a prior distribution on the unknown function \(\beta\) (that is, a prior distribution on the functional space \(\Hcal_{0}(K)\)) it suffices to consider a discrete distribution on the number of components \(p\), and then select \(p\)-dimensional continuous priors for the coefficients \(\beta_j\) and the times \(t_j\) given \(p\). Thanks to this parametric approach, the challenging task of setting a prior distribution on a space of functions is considerably simplified, while simultaneously not constraining the model to any specific distribution (in contrast to, say, Gaussian process regression methods). Moreover, note that our simplifying assumption on \(\beta\) is not particularly limiting from a Bayesian point of view, since any distribution \(\mathbb{P}_0\) on \(\Hcal_0(K)\) can be directly extended to a distribution \(\mathbb{P}\) on \(\Hcal(K)\) by defining \(\mathbb{P}(B) = \mathbb{P}_0(B\cap \Hcal_0(K))\) for all Borel sets \(B\) on \(\Hcal(K)\).

Lastly, since the value of \(p\) can vary, we need a way to introduce dimension information in our MCMC posterior approximation scheme. There are several alternatives such as product space formulations \citep{carlin1995bayesian} or Bayesian averaging/model selection methods \citep{hoeting1999bayesian}, but this is precisely the problem that reversible jump samplers were designed to solve. The use of these samplers fits naturally within our framework, as they allow the complexity of the model to be directly chosen by the data, jointly inferring about the dimensionality and the value of the parameters.

\subsection{Functional linear regression}\label{sec:rkhs-linear-model}

In the case of functional linear regression, the simplified RKHS model considered is
\[
  Y = \alpha + \dotprod{X}{\beta}_K + \epsilon = \alpha + \sum_{j=1}^p \beta_j X(t_j) + \epsilon,
\]
where \(\beta(\cdot)=\sum_{j=1}^p\beta_j K(t_j, \cdot) \in \Hcal_0(K)\), \(\alpha\in\R\), and \(\epsilon \sim \mathcal N(0,\sigma^2)\) is an error term independent from \(X\). This model is essentially a finite-dimensional approximation from a functional perspective to the more general RKHS model that assumes \(\beta \in \Hcal(K)\) \citep{berrendero2024functional}. Since the number of components \(p\) is unknown, the full parameter space is \(\Theta = \bigcup_{p\in\N}\Theta_p\), where \(\Theta_p = \R^p \times [0, 1]^p \times \R \times \R^+_0\). In the sequel, a generic element of \(\Theta_p\) will be denoted by \(\theta_p = (\beta_1,\dots, \beta_p, t_1,\dots, t_p, \alpha, \sigma^2) \equiv (b_p, \tau_p, \alpha, \sigma^2)\), though we will occasionally omit the subscript when the value of \(p\) is understood. For \(\theta \in \Theta\), the reinterpreted model in terms of the available sample information is
\begin{equation}\label{eq:rkhs-model-linear-2}
  Y_i \mid X_i, \theta \ \stackrel{\text{ind.}}{\sim} \mathcal N\left(\alpha + \sum_{j=1}^p \beta_j X_i(t_j),\, \sigma^2\right), \quad i =1,\dots, n.
\end{equation}

It is worth mentioning that the model remains linear in the sense that it fundamentally involves a random variable \(\dotprod{X}{\beta}_K = \Psi_X(\beta)\) in the linear span of the process \(X\). Moreover, this RKHS model is particularly suited as a basis for variable selection methods \citep[see][]{berrendero2019rkhs}.

\subsubsection*{Prior distributions}

A simple and intuitive prior distribution \(\Pi\) for the parameter vector \(\theta \in \Theta\), suggested by the structure of the parameter space and usually employed in the literature, is given by
\begin{equation}\label{eq:prior-linear}
  \begin{array}{r@{\ }c@{\ }l@{\quad}l}
    p                    & \sim               & \pi,                   &                \\[0.3em]
    t_j                  & \textoversim{ind.} & \mathcal{U}[0,1],      & j = 1,\dots,p, \\[0.3em]
    \beta_j              & \textoversim{ind.} & \mathcal{N}(0,\eta^2), & j = 1,\dots,p, \\[0.3em]
    \Pi(\alpha,\sigma^2) & \propto            & \dfrac{1}{\sigma^2},   &
  \end{array}
\end{equation}
where  \(\pi\) is any discrete distribution on \(\N\) (e.g.~uniform on a given finite subset) and \(\eta^2 \in \R^+\) is a hyperparameter of the model that depends strongly on the expected scale of the data. On the one hand, note the use of a joint prior distribution on \(\alpha\) and \(\sigma^2\), which is a widely used non-informative prior known as Jeffrey's prior \citep{jeffreys1946invariant}. One should be wary of the Jeffreys-Lindley paradox when assigning improper priors \citep[see][]{robert2014jeffreys}, but since the two parameters involved are common to all possible values of \(p\), this formulation should not present difficulties. On the other hand, the prior on the coefficients \(\beta_j\) (which are interchangeable when coupled with their respective \(t_j\)) is deliberately kept simple and independent of \(p\), as this will greatly speed up and simplify the computations of the RJMCMC sampler later on. Nevertheless, the methods could be adapted to more complicated and dependent priors without too much trouble.

\subsection{Functional logistic regression}\label{sec:rkhs-logistic-model}

In the case of functional logistic regression, we regard the binary response variable \(Y\in\{0, 1\}\) as a Bernoulli random variable given the functional regressor \(X\). Then, following the approach suggested by \citet{berrendero2023functional}, a simplified logistic RKHS model is given by the equation
\begin{equation}\label{eq:rkhs-model-logistic}
  \mathbb P(Y=1 \mid X) = \frac{1}{1 + \exp\{-\alpha - \dotprod{X}{\beta}_K\}}, \quad \alpha \in \R, \ \beta \in \Hcal_{0}(K).
\end{equation}

Indeed, note that this can be seen as a finite-dimensional approximation (with a functional interpretation) to the general RKHS functional logistic model proposed by these authors, which can be obtained by replacing \(\Hcal_{0}(K)\) with \(\Hcal(K)\). After incorporating the sample information, we can rewrite~\eqref{eq:rkhs-model-logistic} in parametric form for \(\theta\in\Theta\) as
\begin{equation}\label{eq:rkhs-model-logistic-2}
  Y_i \mid X_i,\theta \ \stackrel{\text{ind.}}{\sim} \operatorname{Bernoulli}(H_\theta(X_i)), \quad i=1,\dots, n,
\end{equation}
where
\[
  H_\theta(X_i) = \mathbb P(Y_i=1 \mid X_i,\theta) = \frac{1}{\displaystyle 1 + \exp\left\{-\alpha - \sum_{j=1}^p \beta_j X_i(t_j)\right\}}, \quad i=1,\dots, n.
\]

In much the same way as the linear regression model described above, this RKHS-based logistic regression model offers some advantages over the classical \(L^2\)-model. First, it has a more straightforward interpretation and allows for a workable Bayesian approach. Secondly, it can be shown that under mild conditions the general RKHS logistic functional model holds whenever the conditional distributions \(X | Y=i\) (\(i=0,1\)) are homoscedastic Gaussian processes, and in some cases it also entails the \(L^2\)-model~\citep[see][]{berrendero2023functional}; this arguably provides a solid theoretical motivation for the reduced model. Incidentally, these models also shed light on the near-perfect classification phenomenon for functional data, described by \citet{delaigle2012achieving} and further examined for example in the works of \citet{berrendero2018use} or \citet{torrecilla2020optimal}.

Furthermore, a maximum likelihood approach for parameter estimation (although not considered here) is possible as well, since the finite-dimensional approximation  mitigates the problem of non-existence of the MLE in the functional case. However, let us recall that even in finite-dimensional settings there are cases of quasi-complete separation in which the MLE does not exist \citep{albert1984existence}. Additionally, the non-existence issue of the MLE becomes more pronounced as the dimension increases, as exemplified in the theory recently developed by \citet{candes2020phase}. In any event, we argue that the simplified RKHS model presented here is a compelling and feasible approach to functional logistic regression, since it bypasses the main difficulties of the usual maximum likelihood techniques.

\subsubsection*{Prior distributions}

As far as prior distributions go, we proceed as we did in the linear model. However, following the advice in \citet{gelman2008weakly} and \citet{ghosh2018use}, we do change the prior of the coefficients \(\beta_j\) and the constant term \(\alpha\), since their interpretation is different now:
\begin{equation}\label{eq:prior-logistic}
  \begin{array}{r@{\ }c@{\ }l@{\quad}l}
    p       & \sim               & \pi,                   &                \\[0.3em]
    t_j     & \textoversim{ind.} & \mathcal{U}[0,1],      & j = 1,\dots,p, \\[0.3em]
    \beta_j & \textoversim{ind.} & t_5(0,2.5),            & j = 1,\dots,p, \\[0.3em]
    \alpha  & \sim               & \mathrm{Cauchy}(0,10), &
  \end{array}
  \end{equation}
The scaled Student's \(t\) and Cauchy distributions represent robust and weakly informative priors that provide shrinkage and can handle the case of separation in logistic regression. The specific values of the parameters have been chosen via experimentation following the initial recommendations in \citet{ghosh2018use}, which also establishes the need to scale the regressors to have mean 0 and standard deviation 0.5. Note that in this case the nuisance parameter \(\sigma^2\) does not appear in the model.

\subsection{Reversible jump samplers for prediction}\label{sec:rjmcmc}

For the inference step in our Bayesian procedure, the usual approach would be to consider a fixed number of components, chosen separately via some model selection criterion. Instead, we aim to construct a fully Bayesian methodology that allows us to model the number of components and the component parameters jointly. Since we do not impose conjugate priors and the posterior distribution does not have a recognizable shape, a simple way to achieve the desired outcome is to use reversible jump MCMC (RJMCMC) techniques to approximate the posterior. Originally envisioned for approximate inference in Bayesian mixtures, they can be used to sample models with an unknown number of components, providing a certain level of flexibility and theoretically allowing the exploration of the whole parameter space \citep[see][]{richardson1997bayesian}.

The basis of the RJMCMC mechanism is a clever reformulation of the standard MCMC technique: on each iteration, apart from updating the current parameters, it tries to increase or decrease the dimension, creating new components or eliminating some already present. The acceptance fraction for these new moves is selected so that detailed balance as a whole is maintained, but this includes the possibly challenging computation of a certain Jacobian that is the key to dimension matching \citep{green1995reversible}. However, in nested models such as ours, we can simplify this expression by only allowing on each iteration either the birth of a new component or the death of an existing one (i.e.~changing the dimension by one unit at a time). If we also make the proposal distribution for the newly birthed components independent of previous values, the Jacobian term disappears, and the acceptance fraction reduces to \citep{brooks2003efficient}
\[
  \alpha_\text{nested} = \min\left(1, \frac{\mathcal{L}(Y | X, \theta_{p+1})}{\mathcal{L}(Y| X, \theta_p)}\frac{\Pi(\theta_{p+1})}{\Pi(\theta_p)}\frac{b_{p, p+1}}{d_{p+1, p}}\frac{1}{q(\theta_{+1})}\right).
\]
Here \(p\) is the current number of components, \(b_{p, p+1}\) is the probability of increasing the dimension from \(p\) to \(p+1\) (birth), \(d_{p+1, p}\) is the probability of the reverse move (death), \(q(\theta_{+1})\) is the proposal distribution for the added source, \(\Pi\) is the prior probability and \(\mathcal L\) represents the likelihood. This is already a rather manageable expression that can be implemented efficiently, but we introduce another simplification: the proposal distribution is chosen to match the prior distribution, so that the corresponding terms cancel out. Nonetheless, in real-world scenarios one could consider more sophisticated proposal distributions that might better explore the parameter space; see for example \citet{davies2023transport} or \citet{korsakova2024neural} for some interesting ideas.

From a higher level perspective, the RJMCMC algorithm is an iterative procedure that produces a chain of \(M\) approximate samples \((p_m, \theta^*_{p_m})\) of the posterior distribution \(\Pi_n(p, \theta_p| \mathcal D_n)\), each possibly of a different dimension; see Appendix~\ref{app:validation} for a visual representation. Given a previously unseen regressor \(X_{\text{test}}\), with each of these samples we can generate an individual approximate response following our RKHS models, denoted by \(F(X_{\text{test}}, \theta^*_{p_m})\). In the linear case we sample from \(Y | X_{\text{test}}, \theta^*_{p_m}\) as in~\eqref{eq:rkhs-model-linear-2}, and in the logistic case we directly consider the probabilities \(\mathbb{P}(Y=1|X_{\text{test}},\theta^*_{p_m})\) in~\eqref{eq:rkhs-model-logistic-2}. We propose to combine these predictions in four different ways.

\subsubsection*{One-stage methods}

In these methods we essentially make use of the so-called posterior predictive distribution (PP), and they are in turn divided in two approaches.

\paragraph*{(I) Weighted sum (W-PP).} The first idea involves utilizing all available information by summarizing and aggregating every individual RKHS response. We start by considering a point summary statistic \(g\) (such as the mean, median or mode) and consolidating the predictions from all sub-models, each having distinct dimension. Then, we bring together these predictions through a weighted sum, in which the weights correspond to the relative frequencies of the corresponding values of \(p\) (their approximate posterior):
\[
  \hat Y = \sum_{p} \tilde\Pi_n(p\mid\mathcal D_n) g(\{F(X_{\text{test}}, \theta^*_{p_m})\}_{m:p_m=p}),
  \]
where \(\tilde \Pi_n(p|\mathcal D_n) = M^{-1}\sum_m \mathbb I(p_m=p)\) and \(\mathbb I\) is the indicator function. In the logistic case, we employ the usual thresholding procedure to convert the final probability to a binary class label in \(\{0,1\}\). Note that this approach is reminiscent of the factorization \(\Pi_n(p, \theta_p|\mathcal D_n) = \sum_{p}\Pi_n(p|\mathcal D_n)\Pi_n(\theta_p|p, \mathcal D_n)\) of the posterior distribution. We shall see in Section~\ref{sec:results} that this method produces the best results in practice.

\paragraph*{(II) Maximum a posteriori (MAP-PP).} We also consider a MAP strategy, where only the most probable sub-model is used and the rest of the samples are discarded. In this case, if \(\tilde p = \arg\max_p \tilde  \Pi_n(p|\mathcal D_n)\), predictions are computed as \(\hat Y = g(\{F(X_{\text{test}}, \theta^*_{p_m})\}_{m:p_m=\tilde p})\). Although this method may ignore many samples in the posterior approximation, this omission can help reduce noise and prevent outliers from affecting the final prediction.

\subsubsection*{Two-stage methods}

In these methods we focus only on the marginal posterior distribution of \(\tau_p = (t_1, \dots, t_p)\), disregarding the rest of the parameters and effectively constructing a variable selection procedure. After choosing a set of time instants \(\hat \tau_p=(\hat t_1,\dots, \hat t_p)\) using the approximate posterior samples \(\{\tau^*_{p_m}\}\), we can reduce the original functional regressors to just the marginals \(\{X_i(\hat \tau_p)\}=\{(X_i(\hat t_1),\dots, X_i(\hat t_p))\}\) and apply any of the well-known finite-dimensional prediction algorithms suited for this situation.

\paragraph*{(III) Weighted sum (W-VS).} We can mirror the weighted sum approach of the one-stage methods, with prediction computed as
\[
  \hat Y = \sum_{p} \tilde\Pi_n(p\mid \mathcal D_n) G(X_{\text{test}}, \hat \tau_{p}),
\]
where \(\hat \tau_{p}=(\hat t_1, \dots ,\hat t_p)=g(\{\tau^*_{p_m}\}_{m:p_m=p})\) and \(g\) is a component-wise summary statistic. The function \(G\) represents the prediction for \(X_{\text{test}}\) of a regular linear/logistic regression algorithm, fitted with the transformed data set \(\{(X_i(\hat \tau_p), Y_i): i=1,\dots,n\}\).

\paragraph*{(IV) Maximum a posteriori (MAP-VS).} We also consider a MAP approach to variable selection with only the information of the most probable sub-model, i.e., \(\hat Y = G(X_{\text{test}}, \hat \tau_{\tilde p})\).

\section{Posterior consistency}\label{sec:consistency}

This section explores the theoretical foundations of the proposed Bayesian models in the context of predictive inference. In particular, we establish results grounded in the theory of posterior consistency, which provides rigorous justification for the reliability of these models in asymptotic settings. For an in-depth treatment of posterior consistency and related concepts, we refer the reader to \citet{ghosal2017fundamentals}.

Firstly, let us recall what we understand by posterior consistency and posterior contraction rates. Consider the general setting of an i.i.d.\ sample \(X_1,\dots, X_n\) from a random variable \(X\) taking values in a certain space \(\mathcal X\). Let us fix a prior distribution \(\Pi\) for random variables \(\theta\) on the parameter space \(\Theta\), that is, \(\theta \sim \Pi\), and let \(P_\theta\) represent a sampling model (a distribution on \(\mathcal X\) indexed by \(\theta \in \Theta\)) such that \(X | \theta \sim P_{\theta}\). Furthermore, assume that the model is well-specified, i.e., there is a true value \(\theta_{0}\in\Theta\) such that \(X \sim P_{\theta_0}\), and denote by \(P_0^\infty\) and \(P_0^{(n)}\) the joint probability measure of \((X_1, X_2, \dots)\) and \((X_1, X_2, \dots, X_n)\), respectively, when \(\theta_0\) is the true value of the parameter.

\begin{definition}
We say that the posterior distribution is (strongly) consistent at \(\theta_0\) if for every neighborhood \(B\) of \(\theta_0\) (which for a metric space can just be the open balls around \(\theta_0\)) we have
  \[
    \lim_{n\to\infty} \Pi_n(\theta \in B \mid X_1, \dots, X_n) = 1 \quad P_0^\infty-\text{a.s.}
  \]
\end{definition}

Note that the conditional probabilities are computed under the assumed joint distribution of \((\theta, (X_1, X_2,\dots))\). Essentially, we are saying that the posterior concentrates around \(\theta_0\) for almost all sequences of data, and thus the effect of the prior gets diluted as more and more data is available for the inference. Furthermore, when the parameter space carries a metric, say \(d\), we are also interested in the speed at which the posterior approaches the true parameter \(\theta_0\). This is what we call a \textit{posterior contraction rate}, which is a significant refinement of the comparatively weaker concept of consistency.

\begin{definition}
  A sequence \(\epsilon_n\to 0\) is a posterior contraction rate at \(\theta_0\) with respect to \(d\) if for every \(M_n\to\infty\) the posterior satisfies \(\Pi_n(\theta: d(\theta_0, \theta) \geq M_n \epsilon_n | X_1,\dots, X_n) \to 0\) in \(P_0^{(n)}\)-probability.
\end{definition}

Naturally, once we have established a contraction rate, every slower sequence is also a valid contraction rate. Even though we are interested in the fastest possible rate, this is often difficult to compute or does not exist at all. Hence, we are usually content with finding a good enough rate, which we call ``the'' rate of contraction for the model.

In the rest of the section we analyze how these frequentist concepts apply to our Bayesian functional models; we focus on the linear case, but the results we obtain generally hold \textit{mutatis mutandis} in the logistic case. All technical details and proofs are deferred to Appendix~\ref{app:proofs}. For the functional linear model, our data is an i.i.d.\ sample \((X_1, Y_1), \ldots, (X_n, Y_n)\) of a random vector \((X,Y)\), where \(X=X(t)\) is a second-order stochastic process taking values in \(\mathcal X=\mathcal C[0,1]\), the space of continuous functions, and \(Y\) is a random variable taking values in \(\mathcal Y=\R\). The state space is \(\mathcal X \times \mathcal Y = \mathcal C[0,1]\times \R\), which is a complete separable metric space. Meanwhile, the infinite-dimensional parameter space can be characterized via the Euclidean spaces \(\Theta_p = \{(b, \tau, \alpha, \sigma^2): b=(\beta_1,\dots,\beta_p) \in \R^p, \tau=(t_1,\dots,t_p) \in [0,1]^p, \alpha\in \R, \sigma^2 \in \R^+_0\}\) as the infinite union
\begin{equation}\label{eq:parameter-space-union}
  \Theta = \bigcup_{p=1}^\infty \Theta_p,
\end{equation}
and note that given \(\theta \in \Theta\) there is a unique \(p=p(\theta)\) such that \(\theta \in \Theta_p\). As usual, we equip both \(\mathcal X \times \mathcal Y\) and \(\Theta\) with their respective Borel sigma-algebras.

In terms of \(\theta\in\Theta\), the data distribution can be written as \(P_\theta(X,Y)=P_{b, \tau, \alpha, \sigma^2}(X,Y)\), and formally the joint distribution  factorizes as \(P_{\theta}(X,Y)=Q_X P_\theta(Y|X)\). Here \(Q_X\) is the distribution of the underlying process \(X\), and in our RKHS setting, \(P_\theta(Y|X)\) is the normal distribution \(\mathcal N(\alpha + \sum_{j=1}^{p(\theta)} \beta_j X(t_j),\, \sigma^2)\). Moreover, for convenience, we will denote the sequences \((X,Y)_{1:\infty} := (X_1, Y_1), (X_2, Y_2), \dots\) and \((X,Y)_{1:n} := (X_1,Y_1), (X_2, Y_2), \dots, (X_n, Y_n)\). The full hierarchical model under consideration is
\begin{equation}\label{eq:model-linear}
  \begin{aligned}
    \text{(no.\ of components)}\quad & \mathcal P \sim \pi,                                                                                    \\
    \text{(component values)}\quad   & b \mid \mathcal P=p \sim \phi_p,                                                                        \\
    \text{(component times)}\quad    & \tau \mid \mathcal P=p \sim \psi_p,                                                                     \\
    \text{(intercept)}\quad          & \alpha \sim \Gamma,                                                                                     \\
    \text{(error variance)}\quad     & \sigma^2 \sim \Delta,                                                                                   \\
    \text{(observed data)}\quad      & (X,Y)_{1:n} \mid b, \tau, \alpha, \sigma^2 \sim P_{b, \tau, \alpha, \sigma^2}(X,Y) \quad \text{i.i.d.},
  \end{aligned}
\end{equation}
where \(\pi\), \(\phi_p\), \(\psi_p\), \(\Gamma\) and \(\Delta\) are probability measures on \(\N\), \(\R^p\), \([0,1]^p\), \(\R\) and \(\R^+_0\), respectively. Lastly, define the random variable \(\theta = (b, \tau, \alpha, \sigma^2)\), which takes values in \(\Theta\), and denote by \(\Pi\) the prior distribution on \(\theta\) implied by the model in~\eqref{eq:model-linear}.

\subsection{Consistency via Doob's theorem}

It turns out that, under very general conditions, the posterior distribution is always consistent at almost every value of \(\theta_0\) with respect to the measure induced by the prior \citep{doob1949application}. We state this classical result in the general setting introduced earlier.

\begin{theorem}[Doob's consistency theorem]\label{th:doob}
Let the state space \(\mathcal X\) and the parameter space \(\Theta\) be complete separable metric spaces, endowed with their respective Borel sigma-algebras. If \(\theta \mapsto P_\theta\) is one-to-one and \(\theta \mapsto P_\theta(A)\) is measurable for all measurable sets \(A\subseteq \mathcal X\), then the posterior distribution is consistent at \(\Pi\)-almost all values of \(\Theta\). That is, there exists \(\Theta_*\subseteq \Theta\) such that \(\Pi(\Theta_*)=1\) and for all \(\theta_0\in\Theta_*\), if \(X_1,X_2,\ldots \sim P_{\theta_0}\) i.i.d., then for any neighborhood \(B\) of \(\theta_0\) we have
  \[
    \lim_{n\to\infty} \Pi_n(\theta \in B \mid X_1,\dots, X_n) = 1 \quad P_{0}^\infty-\text{a.s.}
  \]
\end{theorem}

At this point we can follow an approach very similar to the one in \citet{miller2023consistency}, where posterior consistency is established through Doob's theorem in a mixture model with an unknown number of components, which has an infinite-dimensional parameter space that factorizes in the same way as~\eqref{eq:parameter-space-union}. Considering that we will only be interested in small balls around the true value of the parameter, we can define a bounded metric for \(\theta, \theta' \in \Theta\) by
\begin{equation}\label{eq:metric-doob}
  d_\Theta(\theta, \theta')= \begin{cases}
    \min \left\{\|\theta - \theta'\|, 1\right\}, \quad & \text{if } p(\theta)=p(\theta'), \\
    1, \quad                                           & \text{otherwise}.
  \end{cases}
\end{equation}
Since each \(\Theta_p\) is itself a complete separable metric space with the inherited Euclidean norm, Proposition A.1 in \citet{miller2023consistency} ensures that \((\Theta, d_\Theta)\) is a complete separable metric space. Note that \(P_\theta(X,Y)\) is invariant under permutations of the component labels \(b\) and \(\tau\), so in an identifiable model we can only show consistency in the parameter space up to one such permutation. To that effect, let \(S_p\) denote the set of permutations of \(\{1,\dots,p\}\), and for \(\theta \in \Theta_p\) and any \(\nu\in S_p\), denote by \(\theta[\nu]\) the result of applying the permutation \(\nu\) to the component labels of \(\theta\). That is, if \(\theta=(\beta_1,\dots,\beta_p, t_1,\dots, t_p,\alpha,\sigma^2)\), then \(\theta[\nu]=(\beta_{\nu_1},\dots,\beta_{\nu_p}, t_{\nu_1},\dots, t_{\nu_p},\alpha,\sigma^2)\). Next, for \(\theta_0\in\Theta_p\) and \(\epsilon>0\) define the neighborhood \(\tilde{B}(\theta_0, \epsilon)  =\bigcup_{\nu\in S_p} \{\theta \in \Theta: d_\Theta(\theta_0[\nu], \theta) < \epsilon \}\), which consists of all parameters that are within \(\epsilon\) of some permutation of (the component labels of) \(\theta_0\). Now, for identifiability to hold, we need to place some restrictions on the prior.

\begin{condition}[Identifiability constraints] Under the model in~\eqref{eq:model-linear}, for all \(p\in\N\):\label{cond:condition-ident}
  \begin{enumerate}[label=(\roman*)]
    \item \(\Pi(t_i=t_j|\mathcal P=p)=0\) for all \(1\leq i < j \leq p\).\label{cond:condition-ident-1}
    \item There exists \(\delta>0\) such that \(\Pi(|\beta_j|<\delta|\mathcal P=p)=0\) for all \(1\leq j \leq p\).\label{cond:condition-ident-2}
  \end{enumerate}
\end{condition}
Both assumptions can be interpreted as a way of pursuing parsimony in the model, aiming for as few components as possible. In practical and computational terms, we can think of \(\delta\) as the \textit{machine precision number}, so that virtually all continuous prior distributions satisfy the associated condition when implemented numerically in a computer. With this setup in mind, we are ready to state our own consistency result.

\begin{theorem}\label{th:consistency-doob-linear}
  Suppose that Condition~\ref{cond:condition-ident} holds and the covariance function \(K\) of the underlying process \(X\) is strictly positive definite. Then, there exists \(\Theta_*\subseteq \Theta\) such that \(\Pi(\theta \in \Theta_*)=1\) and for all \(\theta_0\in\Theta_*\), if \((X,Y)_{1:\infty} \sim P_{\theta_0}(X,Y)\) i.i.d., then for all \(\epsilon > 0\)
  \[
    \lim_{n\to\infty} \Pi_n(\theta \in \tilde{B}(\theta_0,\epsilon) \mid (X,Y)_{1:n}) = 1 \quad P_{0}^\infty(X,Y)-\text{a.s.}
  \]
  and
  \[
    \lim_{n\to\infty} \Pi_n(\mathcal P=p(\theta_0) \mid (X,Y)_{1:n}) = 1 \quad P_{0}^\infty(X,Y)-\text{a.s.}
  \]
\end{theorem}

The hypothesis of positive definiteness of \(K\) is needed to ensure identifiability. On the other hand, the second conclusion is of certain relevance in itself, because the estimation of the number of components in mixture-like models is a hard problem in general \citep[see][and references therein]{miller2018mixture}. Moreover, it is worth pointing out that the proof of Theorem~\ref{th:consistency-doob-linear} can be easily adjusted to guarantee consistency when the number of components is fixed beforehand and the parameter space is finite-dimensional.

\subsubsection*{A remark on Lebesgue consistency}

All in all, Theorem~\ref{th:consistency-doob-linear} guarantees consistency for \(\Pi\)-almost every parameter in the support of the prior distribution. However, even though we can choose the prior so that \(\supp( \Pi) = \Theta\), in principle there is no assurance that the \(\Pi\)-null set in which consistency may fail will not be a large set with respect to other measures. In fact, when the parameter space is infinite-dimensional there are examples of big inconsistency sets, even for reasonably chosen prior distributions~\citep{diaconis1986consistency}. Nonetheless, this problem can be alleviated when the parameter space is a countable union of disjoint finite-dimensional sets. First, note that there is a natural extension of the Lebesgue measure to our parameter space \(\Theta\): just consider the genuine Lebesgue measure \(\lambda_p\) on \(\Theta_p\), and for all \(B\subseteq \Theta\) measurable define \(\lambda_\infty(B) = \sum_{p=1}^\infty \lambda_p(\Theta_p \cap B)\). Then, if we choose a prior distribution with respect to which this measure is absolutely continuous, the inconsistency set in Theorem~\ref{th:consistency-doob-linear} will satisfy \(\lambda_\infty(\Theta \setminus \Theta_*)=0\) and thus be ``small'' with respect to a Lebesgue-type measure. In our case, the requirement of absolute continuity can be relaxed so that sets with nonzero Lebesgue measure have nonzero prior probability for some permutation of the component labels.

\begin{proposition}\label{prop:consistency-lebesgue-linear}
  Suppose that Condition~\ref{cond:condition-ident} holds. Furthermore, assume that for all \(p\in\N\) we have
  \begin{enumerate}[label=(\roman*)]
    \item \(\Pi(\mathcal P = p) > 0\).\label{cond:condition-lebesgue-1}
    \item \(\sum_{\nu\in S_p} \Pi(\theta[\nu] \in B|\mathcal P = p) = 0\) implies \(\lambda_p(B)=0\), for all \(B\subseteq \Theta_p\) measurable.\label{cond:condition-lebesgue-2}
  \end{enumerate}
  Then the conclusion of Theorem~\ref{th:consistency-doob-linear} remains valid with \(\lambda_\infty(\Theta \setminus \Theta_*)=0\).
\end{proposition}

The first condition is a somewhat technical requirement. The second condition is met, for example, if \(\theta| p\) has a density with respect to the Lebesgue measure that is invariant to permutations of the component labels and positive on all of \(\Theta_p\). A similar approach is considered in \citet{nobile1994bayesian} and \citet{miller2023consistency} to establish ``Lebesgue''-almost sure consistency in finite mixture models with a prior on the number of components.

\subsection{Consistency and contraction rates via Schwartz's theorem}\label{sec:consistency-schwartz}

The Doob-type results of the previous section, although interesting, are still assertions on consistency from the point of view of the prior distribution. We now seek additional results that offer more definitive statements and, more importantly, establish posterior contraction rates. In a nonparametric setting where the object of interest is a probability density, there is a stronger consistency result by \citet{schwartz1965bayes} which omits the \(\Pi\)-almost sure qualification under some more restrictive conditions, though it requires a dominated model and a density to estimate. To this end, we consider the conditional model \(Y|X=x, \theta \sim f_\theta(y|x)\), where \(\theta\in\Theta\) is a parameter vector and \(f_\theta(y|x)\) is the density of \(Y\) given \(X=x\) with respect to the Lebesgue measure \(\lambda\) on \(\mathcal Y\) (which under our assumptions is the normal density given by~\eqref{eq:rkhs-model-linear-2}). We could work with this model and analyze the fixed design case with non-i.i.d.\ data \citep[see][]{choi2008remarks}, but we choose to focus on the arguably more significant random design case.

In general, in infinite-dimensional models there is no reference measure with respect to which to define a density. However, in our specific functional regression setting we can find such a density through the following observation: it follows from the disintegration theorem and a straightforward application of Dynkin's \(\pi\)-\(\lambda\) theorem that the function \((x,y)\mapsto f_\theta(y|x)\) is a joint density of \((X,Y)\) with respect to the product measure \(\rho=Q_X\times \lambda\), where \(Q_X\) is the law of the process \(X\). In this way, denoting \(f_\theta := f_\theta(y|x)\), we can express our full model (with respect to the product measure \(\rho\)) as
\[
  (X, Y)_{1:n} \mid f_\theta \stackrel{\text{i.i.d.}}{\sim} f_\theta \quad \text{and} \quad f_\theta \sim \Pi_{\mathcal F},
\]
where \(\Pi_{\mathcal F}\) is a prior distribution on the parameter class \(\mathcal F=\{f_\theta: \theta \in \Theta\}\), which is the set of all densities following our model relative to the dominating measure \(\rho\) on \(\mathcal X \times \mathcal Y\). In practice, we work with the prior distribution \(\Pi\) on \(\Theta\) specified in~\eqref{eq:model-linear}, and interpret \(\Pi_{\mathcal F}\) as the pushforward measure of \(\Pi\) by the measurable mapping \(\Phi: \Theta \to \mathcal F\) given by \(\Phi(\theta)= f_\theta\). As always, we assume the existence of a ``true'' density \(f_0 :=f_{\theta_0}(y|x)\in \mathcal F\) that generates the observations, and denote \(\Theta_0\equiv \Theta_{p(\theta_0)}\).

Now we introduce a few concepts that lie at the core of most extensions of Schwartz's consistency theorem. Recall that the Kullback-Leibler divergence from \(f_0\) to \(f_\theta\) is defined as \(D_{\mathrm{KL}}(f_0 \,\|\, f_\theta) = \int f_0 \log (f_0/f_\theta)\, d\rho\). We say that \(f_0\) belongs to the \textit{Kullback-Leibler support of \(\Pi_{\mathcal F}\)}, and write it as \(f_0 \in \operatorname{KL}(\Pi_{\mathcal F})\), if \(\Pi_{\mathcal F}\left(f_\theta: D_{\mathrm{KL}}(f_0 \,\|\, f_\theta) < \epsilon\right) > 0\) for all \(\epsilon > 0\). This can be interpreted as saying that the prior places sufficient mass ``near'' \(f_0\). On the other hand, let \(N(\epsilon, \mathcal F, d)\) denote the \(\epsilon\)-covering number of the set \(\mathcal F\) with respect to the distance \(d\), i.e., the minimal number of \(d\)-balls of radius \(\epsilon\) needed to cover \(\mathcal F\). Finally, define the Hellinger distance \(d_H\) between two densities as \(d^2_H(f_\theta, f_{\theta'}) = 1 - \int \sqrt{f_\theta f_{\theta'}}\, d\rho\). The following result will give us consistency in the sense of Schwartz, by means of a sieve that controls the complexity of the model as the sample size \(n\) increases, measured in terms of the \textit{metric entropy} (the logarithm of the \(\epsilon\)-covering number) of the model.

\begin{theorem}[Theorem 6.23 in \citet{ghosal2017fundamentals}]\label{th:consistency-hellinger}
  Suppose that for every \(\epsilon > 0\) there exist measurable sets \(\mathcal F_n \subseteq \mathcal F\) and a constant \(C>0\) such that, for sufficiently large \(n\),
  \begin{enumerate}[label=(\roman*)]
    \item \(\log N(\epsilon, \mathcal F_n, d_H) \leq n\epsilon^2\).
    \item \(\Pi_{\mathcal F}(\mathcal F \setminus \mathcal F_n) \leq \exp\{-Cn\}\).
  \end{enumerate}
Then the posterior distribution is consistent relative to \(d_H\) at every \(f_0\in \KL(\Pi_{\mathcal F})\).
\end{theorem}

Looking at the form of the parameter space \(\Theta\) in~\eqref{eq:parameter-space-union}, a natural choice for the sieve is \(\Theta_n = \bigcup_{p=1}^{p_n} \Theta_p\), for \(p_n=O(h(n))\) and \(h(n)\to \infty\) a function to be determined later, with the corresponding sieve of densities defined as \(\mathcal F_n=\{f_\theta\in\mathcal F: \theta \in \Theta_n \}\). Moreover, to keep the conditions as simple as possible, it is more convenient to work with compact parameter spaces \(\Theta_p\) and adapt the prior distributions in~\eqref{eq:model-linear} accordingly; we assume this is so for the remainder of the section. Plus, we need to impose some smoothness restrictions on the underlying process \(X\).

\begin{condition}[Functional constraints]\label{cond:condition-schwartz-X}
  Suppose that each trajectory \(x\) of \(X\) is \(Q_X\)-almost surely Lipschitz with Lipschitz constant \(L(x)>0\), and assume that \(\mathbb E_{Q_X}[L^2]<\infty\).
\end{condition}

See Appendix~\ref{app:proofs} for a brief discussion and examples of processes that satisfy this condition. We can now state a new Hellinger consistency result for our model, which transforms the requirement that \(f_0\in \KL(\Pi_\mathcal F)\) into a condition on the prior \(\Pi\) on \(\Theta_0\) and places an additional mild constraint on the tails of the marginal prior distribution of \(p\). These conditions guarantee that we can choose a suitable growth order for \(p_n\) in the sieve to apply Theorem~\ref{th:consistency-hellinger}.

\begin{theorem}\label{th:consistency-schwartz-linear}
  Assume Condition~\ref{cond:condition-schwartz-X} holds, and suppose that \(\Pi(p)=O(e^{-\delta p (\log p)^k})\) as \(p\to\infty\), with \(\delta > 0\) and \(k>1\). If \(\Pi\) assigns positive mass to every open neighborhood of \(\theta_0\) in \(\Theta_0\), then the posterior distribution in the model \(\theta \sim \Pi\) and \((X,Y)_{1:n}|\theta \sim f_\theta\) i.i.d.\ is consistent relative to \(d_H\) at \(f_0\), i.e., for every \(\epsilon>0\) it holds that
  \[
    \lim_{n\to\infty} \Pi_n(\theta: d_H(f_0, f_\theta) < \epsilon \mid (X,Y)_{1:n})=1\quad P_0^\infty(X,Y)-\text{a.s.}
  \]
\end{theorem}

An immediate consequence is that if \(\Pi(p)\) is supported on a finite set (as we do in our experiments), then the prior tail condition is automatically verified. Moreover, since it only needs to hold for large values of \(p\), we can always consider a hybrid prior distribution that changes at some large cutoff point to have rapidly decaying tails.

Lastly, to find posterior contraction rates for our model we consider an extension of Theorem~\ref{th:consistency-hellinger} that employs the so-called \textit{second Kullback-Leibler variation}, defined by \(V_2(f_0, f_\theta)=\E_{f_0}[(\log(f_0/f_\theta) - D_{\mathrm{KL}}(f_0 \,\|\, f_\theta))^2]\), to quantify the Kullback-Leibler property of \(f_0\) in terms of convergence speed. In particular, for \(\epsilon>0\) we consider the neighborhood
\[
B_2(f_0, \epsilon) = \{f\in \mathcal F: D_{\mathrm{KL}}(f_0 \,\|\, f_\theta) < \epsilon^2,\  V_2(f_0, f_\theta) < \epsilon^2\}.
\]

\begin{theorem}[Theorem~8.9 in \citet{ghosal2017fundamentals}]\label{th:contraction-rate}
  Suppose that there exist measurable sets \(\mathcal{F}_n \subseteq \mathcal F\) and a constant \( C > 0 \) such that, for a sequence \(\epsilon_n \to 0\) with \( n\epsilon_n^2 \to \infty\), the following hold for sufficiently large \(n\):
  \begin{enumerate}[label=(\roman*)]
    \item \( \log N(\epsilon_n/2, \mathcal{F}_n, d_H) \leq n \epsilon_n^2\).
    \item \( \Pi_{\mathcal F}(\mathcal{F} \setminus \mathcal{F}_n) \leq \exp\{-(C+4)n \epsilon_n^2\}\).
    \item \( \Pi_{\mathcal F}(B_2(f_0, \epsilon_n)) \geq \exp\{-C n\epsilon_n^2\}\).
  \end{enumerate}
  Then the posterior rate of contraction at \(f_0 \) with respect to \(d_H\) is \(\epsilon_n \).
\end{theorem}

Note that the condition \(n\epsilon_n^2\to\infty\) excludes the parametric rate \(n^{-1/2}\), since in nonparametric situations the rates are usually slower. Condition (i) bounds the complexity of the model, condition (ii) merely expresses that subsets of negligible prior mass do not play a role in the rate of contraction, and condition (iii) ensures that sufficient prior mass is put ``near'' the true density \(f_0\). Following a similar proof strategy as in Theorem~\ref{th:consistency-schwartz-linear}, we can find a convenient sieve to apply Theorem~\ref{th:contraction-rate} and prove the following result in our model.

\begin{theorem}\label{th:contraction-rate-schwartz-linear}
Assume Condition~\ref{cond:condition-schwartz-X} holds. Let \(0<\gamma<1/2\) and suppose that the prior distribution satisfies \(\Pi(p)=O(e^{-\delta p \log p})\) as \(p\to\infty\), with \(\delta>16(3-4\gamma)/(1-2\gamma)\). If \(\Pi\) has a density on \(\Theta_0\) with respect to Lebesgue measure that is bounded away from zero in a neighborhood of \(\theta_0\), then the posterior distribution in the model \(\theta \sim \Pi\) and \((X,Y)_{1:n}|\theta \sim f_\theta\) i.i.d.\ contracts at a rate \(\epsilon_n= n^{-\gamma}\) relative to \(d_H\) at \(f_0\), i.e., for every \(M_n\to\infty\) it holds that
\[
    \Pi_n(\theta: d_H(f_0, f_\theta) \geq M_n \epsilon_n \mid (X,Y)_{1:n}) \to 0 \quad \text{in } P_0^{(n)}(X,Y)\text{-probability}.
    \]
\end{theorem}

The requirements of this result are akin to those of Theorem~\ref{th:consistency-schwartz-linear}, but the conclusion is considerably stronger. To verify the condition on the density of \(\Pi\), provided that it exists, it suffices to show that it is continuous and positive at \(\theta_0\) (which our proposed prior distributions satisfy). The contraction rate \(\epsilon_n=n^{-\gamma}\) with \(0<\gamma<1/2\) is optimal in the sense that it is as close as desired to the parametric rate of \(n^{-1/2}\), and it is \textit{minimax optimal} for estimators of \(f_\theta\) given the model \(\mathcal F_n\) \citep[see][p.~198 and references therein]{ghosal2017fundamentals}.

To conclude, it is worth pointing out that a contraction rate relative to \(d_H\) is also a contraction rate relative to any distance bounded above by a multiple of the Hellinger distance, so we immediately have consistency and a contraction rate in more descriptive distances. Let \([h]_M\) be the real-valued function \(h\) truncated to the interval \([-M, M]\).

\begin{corollary}\label{th:corollary-contraction-rate}
  Under the same conditions as in Theorem~\ref{th:contraction-rate-schwartz-linear}, \(\epsilon_n = n^{-\gamma}\) is a posterior contraction rate relative to the total variation distance \(d_1(f_0, f_\theta)=\int |f_0 - f_\theta|\, d\rho\) and the mean-variance discrepancy metric \(d_{2,Q_X}(\theta_0, \theta) = \|[\mu_{\theta_0}]_M - [\mu_\theta]_M\|_{2,Q_X} + |\sigma_0^2 - \sigma^2|\), for any \(M>0\), where \(\|\cdot\|_{2,Q_X}\) is the \(L^2(Q_X)\)-norm and \(\mu_\theta(x)= \alpha + \sum_{j=1}^{p(\theta)} \beta_j x(t_j)\).
\end{corollary}

\section{Experimental results}\label{sec:results}

In this section we present the results of the experiments carried out to test the performance of our Bayesian methods in different scenarios, together with an overview of their computational implementation. Further details such as simulation parameters, execution times or algorithmic decisions, as well as additional experiments, figures and tables are provided in Appendices~\ref{app:model-choice},~\ref{app:experiments} and~\ref{app:source-code}, while the code itself is publicly available on GitHub at \url{https://github.com/antcc/rk-bfr-jump}.

For simulated data we consider \(n=200\) training samples and \(n'=100\) testing samples, with functional regressors measured on an equispaced grid of \(100\) points on \([0, 1]\), and for the real data sets we perform a 66\%/33\% train/test split. We fit our models on the training data, using RJMCMC to sample from the approximate posterior, and then compute out-of-sample predictions on the testing set. We independently repeat the whole process 10 times (each with a different train/test configuration) to account for the stochasticity in the sampling step, and average the results across these executions. For the purposes of prediction we consider the point statistics \textit{trimmed mean} (10\%), \textit{median} and \textit{mode} to aggregate together predictions and summarize parameters (see Section~\ref{sec:rjmcmc}). Moreover, the values of \(\{t_j\}\) that fall outside the specified grid are truncated to the nearest neighbor in the grid, with the additional restriction that time instants in different components of our models cannot be equal.

The Python library used to perform the RJMCMC approximation is \textit{Eryn} \citep{karnesis2023eryn}. It is a toolbox for Bayesian inference that allows trans-dimensional posterior approximation, running multiple chains in an ensemble configuration with different starting points, and incorporating a parallel tempering mechanism \citep{hukushima1996exchange} to increase both convergence speed and acceptance rate. Because of execution time constraints, the hyperparameters of the Eryn sampler are selected manually based on preliminary experiments and the authors' recommendations. This reduction in computational cost renders the sampling step practically viable; information on the execution times is provided in Appendix~\ref{app:execution-times}. Moreover, some amount of post-processing is needed to mitigate the well-known \textit{label switching} phenomenon that occurs in MCMC approximations of mixture-like models \citep{stephens2000dealing}; see Appendix~\ref{app:label-switching}.

\subsubsection*{Data sets}

We consider a set of functional regressors common to linear and logistic regression problems. They are four zero-mean Gaussian processes (GPs), each with a different covariance function. In particular, we consider a Brownian motion, a fractional Brownian motion, an Ornstein-Uhlenbeck process, and a GP with a squared exponential kernel.

\paragraph*{Linear regression data sets.} We employ two different types of simulated data sets, all with a common value of \(\alpha=5\) for the intercept and \(\sigma^2=0.5\) for the error variance:
\begin{itemize}
  \item A finite-dimensional RKHS response with three components for each of the four GP regressors mentioned above: \(Y=5 -5X(0.1) + 5X(0.6) + 10X(0.8) + \epsilon\).
  \item A ``component-less'' response generated by an \(L^2\)-model with a smooth underlying coefficient function, namely \(\beta(t)=\log(1+4t)\), again for the same four GPs.
\end{itemize}
As for the real data sets, we use the Tecator data set \citep{borggaard1992optimal} to predict fat content based on near-infrared absorbance curves of 193 meat samples, as well as what we call the Moisture \citep{kalivas1997two} and Sugar \citep{bro1999exploratory} data sets. The first consists of near-infrared spectra of 100 wheat samples and the objective is to predict the samples' moisture content, whereas the second contains 268 samples of sugar fluorescence data in order to predict ash content. The regressors of the three data sets are measured on a grid of 100, 101 and 115 equispaced points on \([0, 1]\), respectively.

\paragraph*{Logistic regression data sets.} Again we consider two different types of simulated data sets, with a common value of \(\alpha=-0.5\) for the intercept:
\begin{itemize}
  \item Four logistic finite-dimensional RKHS responses with the same functional parameter as in the linear regression case (one for each GP). Specifically,
        \[
          \mathbb P(Y=1\mid X) = \frac{1}{1 + \exp\left\{0.5 +5X(0.1) - 5X(0.6) - 10X(0.8)\right\}}.
        \]
  \item Four logistic responses following an \(L^2\)-model with the same coefficient function as in the linear regression case, i.e., \(\beta(t)=\log(1+4t)\).
\end{itemize}
Additionally, we use three real data sets well known in the literature. The first one is a subset of the Medflies data set \citep{carey1998relationship}, consisting on samples of the number of eggs laid daily by 534 flies over 30 days, to predict whether their longevity is high or low. The second one is the Berkeley Growth Study data set \citep{tuddenham1954physical}, which records the height of 54 girls and 39 boys over 31 different points in their lives. Finally, we selected a subset of the Phoneme data set \citep{hastie1995penalized} based on 200 digitized speech frames over 128 equispaced points to predict the phonemes ``aa'' and ``ao''.

\subsubsection*{Comparison algorithms}

We have included a fairly comprehensive suite of comparison algorithms, chosen among the most common frequentist methods used in machine learning and FDA, and following a standard choice of implementation and hyperparameters. There are purely functional methods (such as the usual \(L^2\) regression based on model~\eqref{eq:l2-linear-model}), finite-dimensional models that work on the discretized data (e.g.\ penalized finite-dimensional regression), and variable selection/dimension reduction procedures (like Principal Component Analysis or Partial Least Squares). The main parameters of all these algorithms are selected by cross-validation, and when a number of components needs to be specified, we use the same range as in our own models so that comparisons are fair. A more detailed account of these algorithms is available in Appendix~\ref{app:data-sets}.

\subsubsection*{Results display}

We have adopted a visual approach to presenting the experimentation results, using colored graphs to help visualize them. In each case, the mean and standard deviation of the score obtained across the 10 independent runs is shown, depicting our methods in blue and the comparison algorithms in orange. We also show the global mean of all the comparison algorithms with a dashed vertical line, excluding extreme negative results to avoid distortion. Moreover, we separate one-stage and two-stage methods, the latter being the ones that perform variable selection or dimension reduction prior to a multiple linear/logistic regression method (represented with ``+r''/``+log'' in the figures). We name our methods according to the acronyms described in Section~\ref{sec:rjmcmc}.

\subsection{Functional linear regression}\label{sec:experiments-linear}

The initial experiments carried out indicate that low values of \(p\) provide sufficient flexibility in most scenarios, so we allow the number of components to vary in the set \(\{1,2,\dots,10\}\). Following \citet{nobile2007bayesian} we select a truncated Poisson prior with a low rate parameter \(\lambda=3\) for \(p\), so that simpler models are favored. However, in the experiments with real data we set \(p\sim \mathcal U\{1,2,\dots,10\}\) to allow a less informative exploration of the parameter space. Moreover, for simplicity we choose to scale the regressors and response to have standard deviation unity in the inference step, but then convert the results back to the original scale for prediction. This allows us to set a reasonable value of \(\eta^2=25\) in the weakly informative prior of \(\beta_j\) (see~\eqref{eq:prior-linear}). Lastly, the metric used to evaluate the performance is the Root Mean Square Error (RMSE).

\subsubsection*{Simulated data sets}

In Figure~\ref{fig:reg_rkhs} we see the results for the RKHS response. This is our baseline and the most favorable case for us, as the underlying model coincides with our assumed model. Indeed, we can see that in most instances our algorithms are the ones with lower RMSE, as expected, and there is not much variance among our different approaches to prediction, at least in the one-stage methods. We even manage to consistently beat the \textit{lasso} finite-dimensional regression mechanism, which can select all 100 components for its model, versus our 10 components at most. This is an indicator that our models make a more efficient selection of variables, encapsulating more information with fewer components and justifying our partiality to parsimonious models.

\begin{figure}[ht!]
  \centering
  \includegraphics[width=.95\textwidth]{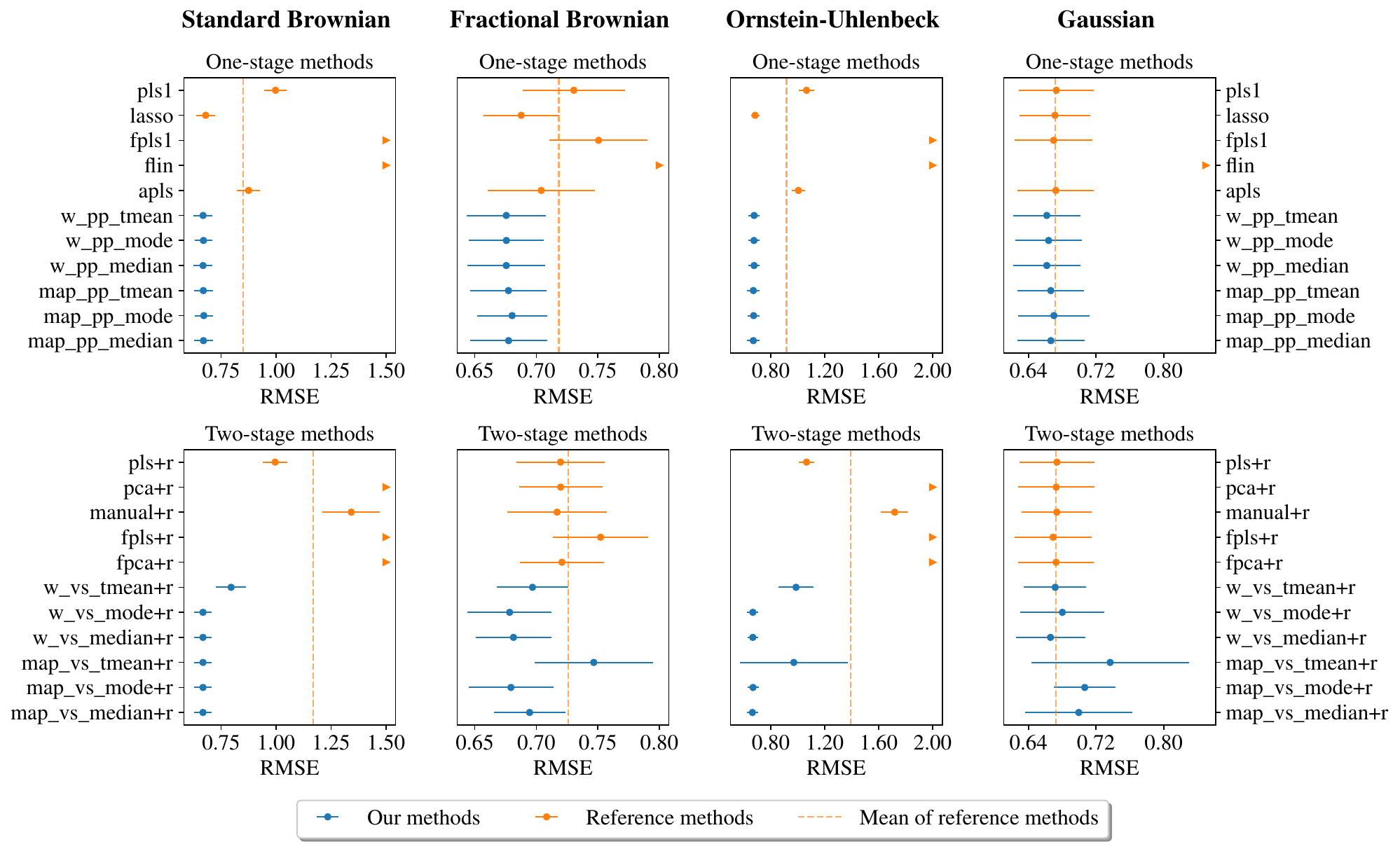}
  \caption{Mean and standard error of RMSE of predictors (lower is better) for 10 runs with GP regressors, one on each column, that obey an underlying linear RKHS model.}\label{fig:reg_rkhs}
\end{figure}

Figure~\ref{fig:reg_l2} shows the results for an underlying \(L^2\)-model, which would be our direct competitor and a more representative test for our Bayesian model. In this case the outcome is satisfactory, as for the most part our models are on a par with the rest, even surpassing other methods that were designed with the \(L^2\)-model in mind. Especially interesting is the comparison with the standard \(L^2\) functional regression (\textit{flin} in the graphics), which we outperform most of the time.

\begin{figure}[ht!]
  \vspace{2em}
  \centering
  \includegraphics[width=.95\textwidth]{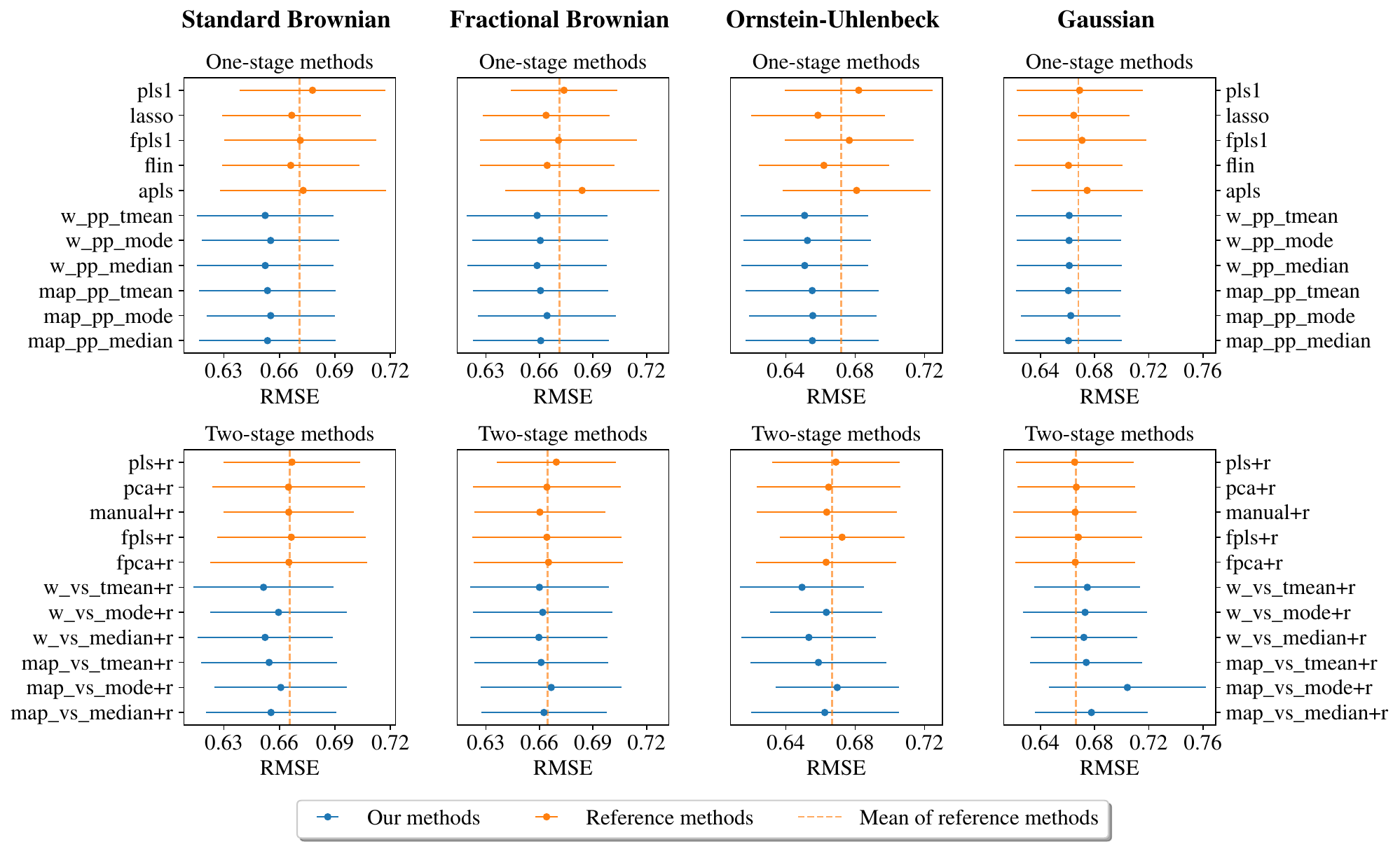}
  \caption{Mean and standard error of RMSE of predictors (lower is better) for 10 runs with GP regressors, one on each column, that obey an underlying linear \(L^2\)-model.}\label{fig:reg_l2}
\end{figure}

\subsubsection*{Real data}

Figure~\ref{fig:reg_real} depicts the results for the real data sets, where we can see that the performance of our one-stage methods is about the same as that of the comparison methods. However, our variable selection methods are somewhat worse that the reference methods, although not by a wide margin. We have to bear in mind that real data is more complex and noisy than simulated data, and it is possible that after a suitable pre-preprocessing we would have obtained better results. Nonetheless, our goal was to perform a general comparison with a uniform methodology, without focusing too much on the specifics of any particular data set. On another note, we see that some of our Bayesian models have a higher standard deviation, partly because there is an intrinsic randomness in the sampling mechanism, and it can be the cause of the occasional worse performance. In relation to this, we observe that the methods that use the trimmed mean as a summary statistic tend to have a worse score, as this statistic is more sensitive to outliers than the median or the mode, even with the 10\% trimming.

\begin{figure}[ht!]
  \centering
  \includegraphics[width=.8\textwidth]{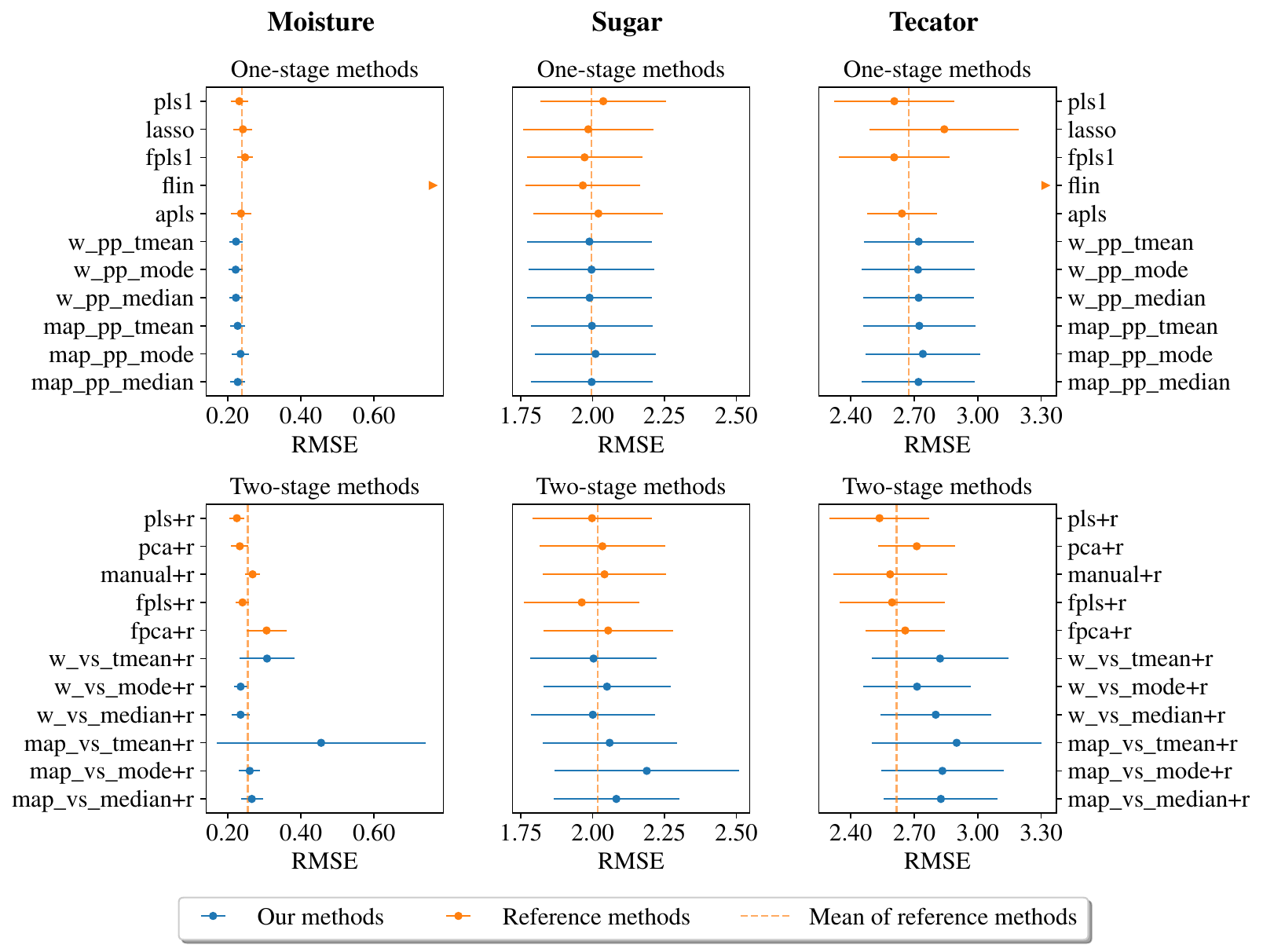}
  \caption{Mean and standard error of RMSE of predictors (lower is better) for 10 runs with real data sets, one on each column.}\label{fig:reg_real}
\end{figure}

\subsection{Functional logistic regression}\label{sec:experiments-logistic}

In this case we set the same prior distribution for \(p\) as in the linear case, but now in the fitting phase the regressors are scaled to have standard deviation 0.5 to accommodate the priors in~\eqref{eq:prior-logistic}. We use the standard threshold of 0.5 to convert probabilities to class labels, and the performance metric is the accuracy, i.e., the rate of correctly predicted samples.

\subsubsection*{Simulated data sets}

First, in Figure~\ref{fig:clf_rkhs} we see the results for the GP regressors with a logistic RKHS response. Our models perform fairly well in this advantageous situation, and even more so in the one-stage case. We again improve the results of other models that can select far more components than our self-imposed limit of 10. Moreover, whenever our two-stage methods score lower, the differences account for only a couple of misclassified samples at worst.

Subsequently, Figure~\ref{fig:clf_l2} shows that in the \(L^2\) scenario the results are again promising, since our models score consistently on or above the mean of the reference models, and in many instances exceed most of them. In this case we also beat the alternative functional logistic regression method (\textit{flog}). In addition, in this situation the overall accuracy of all methods is poor (around 60\%), so this is indeed a difficult problem in which even small increases in accuracy are relevant.

\begin{figure}[ht!]
  \centering
  \includegraphics[width=.925\textwidth]{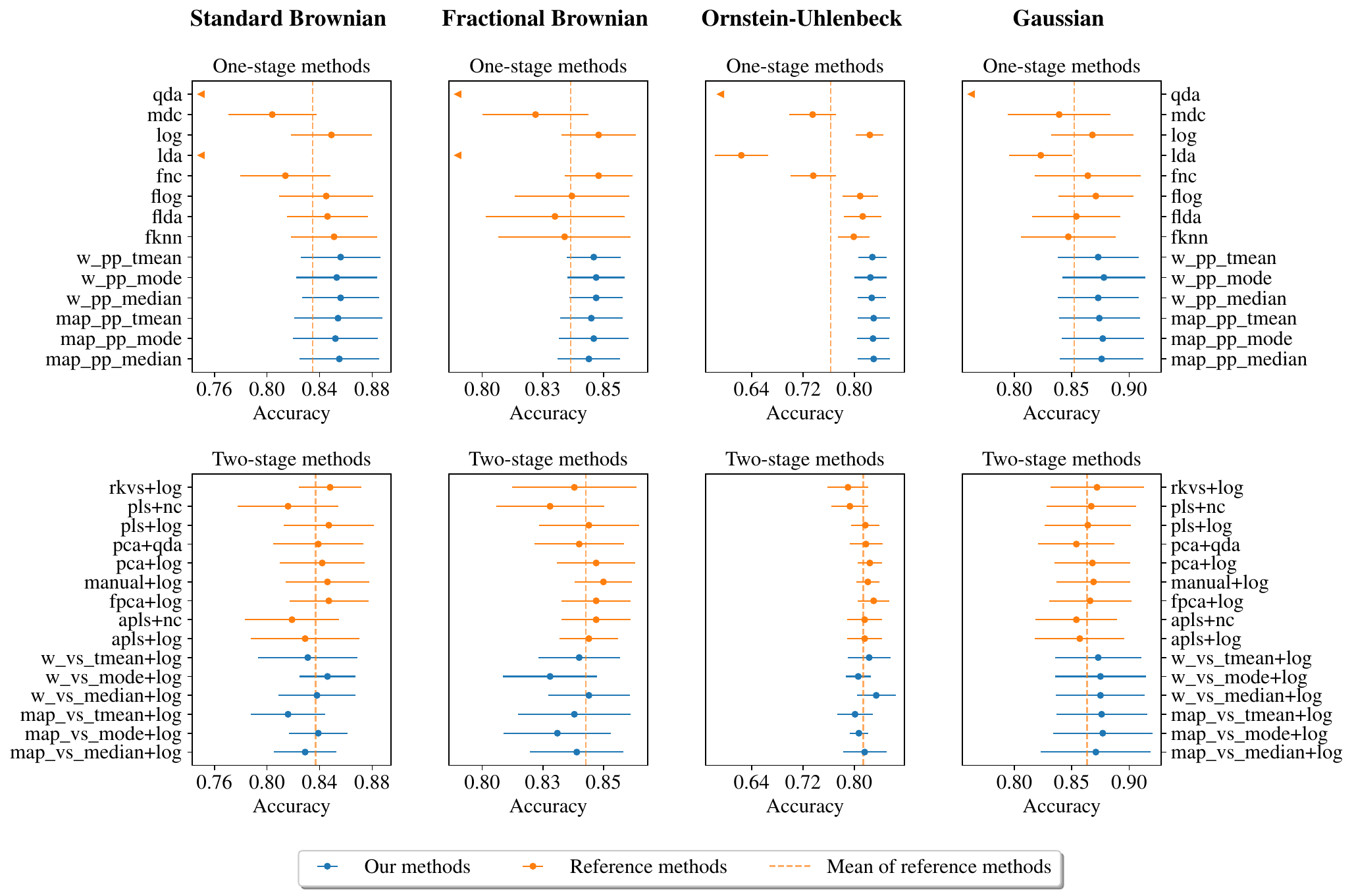}
  \caption{Mean and standard error of accuracy of classifiers (higher is better) for 10 runs with GP regressors, one on each column, that obey an underlying logistic RKHS model.}\label{fig:clf_rkhs}
\end{figure}

\begin{figure}[ht!]
  \vspace{2em}
  \centering
  \includegraphics[width=.925\textwidth]{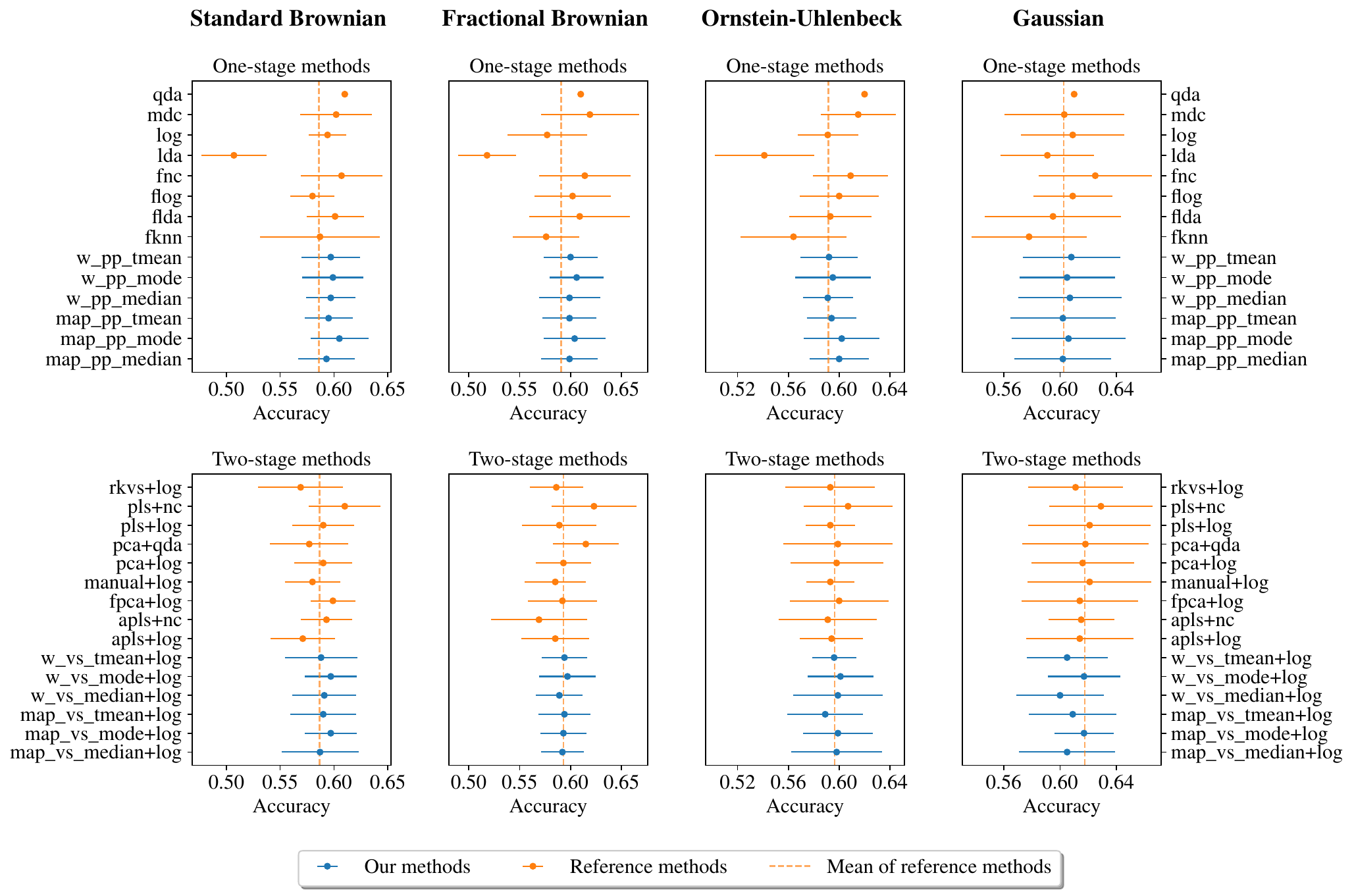}
  \caption{Mean and standard error of accuracy of classifiers (higher is better) for 10 runs with GP regressors, one on each column, that obey an underlying logistic \(L^2\)-model.}\label{fig:clf_l2}
\end{figure}

\subsubsection*{Real data}

As for the real data sets, in Figure~\ref{fig:clf_real} we see positive results in general, obtaining in some cases accuracies well above the mean of the reference models. In particular, the weighted sum methods tend to have slightly better results than the MAP methods, which is a trend that was also present in the simulated data sets and in the linear regression experiments. This was somewhat expected, since the weighted predictions use the full range of information available for prediction.

\begin{figure}[ht!]
  \centering
  \includegraphics[width=.71\textwidth]{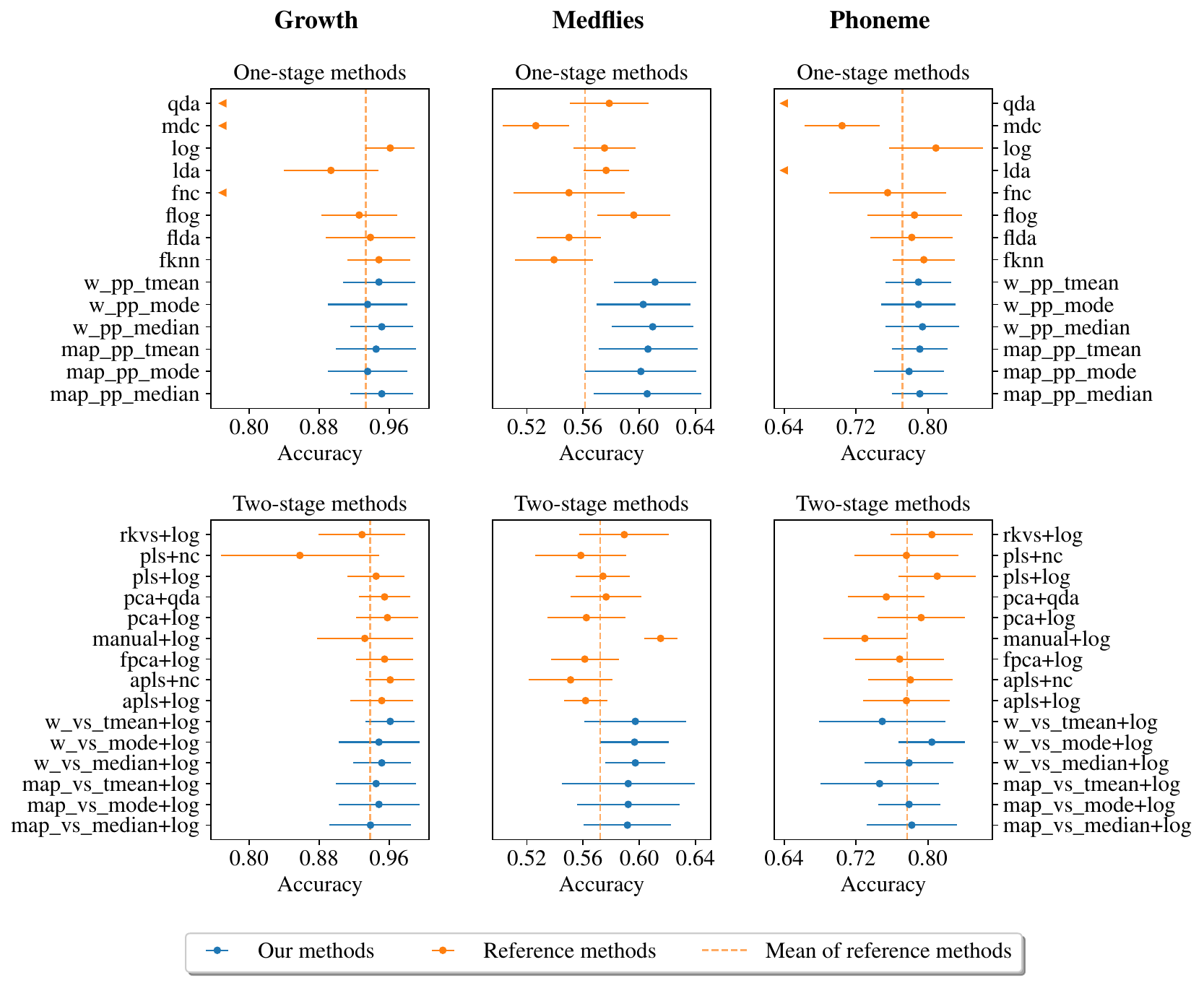}
  \caption{Mean and standard error of accuracy of classifiers (higher is better) for 10 runs with real data sets, one on each column.}\label{fig:clf_real}
\end{figure}

\section{Conclusion}\label{sec:conclusion}

We have introduced a natural and computationally feasible way of integrating Bayesian inference into functional regression models, by means of a RKHS approach that simplifies the inherently hard task of setting a prior distribution on an infinite-dimensional space. The proposed RKHS formulation gives a common framework to all finite-dimensional models based on linear combinations of the marginals of the underlying process, establishing a solid theoretical foundation for these popular points-of-impact models while retaining a functional viewpoint. Our approximation relies on simpler functional parameters, thus enhancing interpretability and ease of implementation, and also enables meaningful model comparisons across different dimensions and facilitates the analysis of the effects of local characteristics of the process.

We have also proved posterior consistency results that ensure the coherence and correctness of the Bayesian methods we developed. These kinds of results have in other contexts more intricate and restrictive conditions to arrive at essentially the same conclusions as we did, but again the introduction of RKHS's is the key point to greatly simplifying them. On the one hand, we have successfully adapted the methodology recently introduced in \citet{miller2023consistency}---originally intended for mixture models---to the fundamentally different scenario of functional regression, obtaining ``Lebesgue''-almost sure consistency for both the model parameters and the unknown number of components. On the other hand, we have derived alternative consistency results and optimal contraction rates through a more sophisticated approach based on Schwartz's theorem, which is the standard tool in nonparametric consistency problems. In addition, it is worth mentioning that the theorems we have proved are applicable to a wide range of prior distributions, including many that are rather simple and hence easier to work with.

Lastly, we have presented numerical evidence that supports the proposed Bayesian methodology and its predictive performance with simulated and real data sets. Thanks to our RKHS formulation, we can effectively leverage the capabilities of RJMCMC samplers and integrate the unknown number of components into the Bayesian procedure, which to our knowledge is a novel application in FDA. Moreover, a key finding in our empirical studies is that a relatively low number of components is sufficient to obtain good results. This practical side of our work showcases that the Bayesian prediction methods we constructed are competitive against several non-cherry-picked frequentist techniques, especially those based on the usual \(L^2\)-models, while still remaining viable implementation-wise. Of course, our proposed model is not without limitations, and we are not suggesting that the \(L^2\) perspective should be abandoned, but merely offering a simple, theoretically-backed alternative that can perform better in many situations.


  {\footnotesize
  \subsubsection*{Acknowledgments}
    The authors would like to acknowledge the valuable computational resources provided by the Centro de Computación Científica-Universidad Autónoma de Madrid (CCC-UAM). This research was partially supported by grants PID2023-148081NB-I00 and PRE2020-095147 of the Spanish Ministry of Science and Innovation (MCIN), co-financed by the European Social Fund (ESF). J. R. Berrendero and A. Cuevas also acknowledge financial support from grant CEX2023-001347-S funded by MICIU/AEI/10.13039/501100011033.
    }


\bibliography{bibliography}


\newpage
\appendix
\counterwithin{theorem}{section} 
\input{paper-supplement.tex}

\end{document}

%% file: paper-supplement.tex
%
%


\section{Model choice and implementation details}\label{app:model-choice}

\subsection{Posterior distributions}\label{app:posterior}

For the posterior distributions in our Bayesian formulation, we only compute a function proportional to their log-density, since that is enough for a MCMC algorithm to work. Consider the parameter vector \((p, \theta_p)\), where \(\theta_p=(b_p, \tau_p, \alpha, \sigma^2)\) with \(b_p=(\beta_1, \dots, \beta_p)\) and \(\tau_p = (t_1,\dots, t_p)\). Recall that we have a labeled data set of independent observations from \((X, Y)\), denoted by \(\mathcal D_n =\{(X_i, Y_i): i=1,\dots, n\}\). A standard algebraic manipulation in the posterior expression using Bayes' formula yields the following results.

\begin{proposition} Under the linear RKHS model and the prior distribution in Section~\ref{sec:rkhs-linear-model}, the log-posterior distribution up to an additive constant is
  \begin{align*}
  \log\Pi_n(p, \theta_p\mid \mathcal D_n) \propto {} & -\frac{1}{2}\left(\frac{\|b_p\|^2}{\eta^2} - \frac{\|\bm{Y}_n - \alpha\bm{1}_n + \mathcal X_{\tau_p} b_p\|^2}{\sigma^2} \right) - (n+2)\log \sigma\\
  &- p\log \eta - \frac{p}{2}\log(2\pi) + \log \Pi(p),
  \end{align*}
  where \(\bm Y_n=(Y_1,\dots,Y_n)^T\), \(\bm{1}_n\) is an \(n\)-dimensional vector of ones, and \(\mathcal X_{\tau_p}\) is the data matrix \((X_i(t_j))_{i,j}\), for \(i=1,\dots,n\) and \(j=1,\dots,p\).
\end{proposition}

\begin{proposition}
  Under the logistic RKHS model and the prior distribution in Section~\ref{sec:rkhs-logistic-model}, the log-posterior distribution up to an additive constant is
  \begin{align*}
    \log \Pi_n(p, \theta_p \mid \mathcal D_n) \propto {} & \sum_{i=1}^n \left[Y_i \left(\alpha + \dotprod{X_i}{\beta}_K\right) - \log\left(1 + \exp\left\{\alpha +\dotprod{X_i}{\beta}_K\right\}\right) \right]\\
     & - 3\sum_{j=1}^p \log\left(4\beta_j^2 + 125\right) + p \left[\frac{15}{2}\log 5 + 4 \log 2 - \log (3\pi)\right]\\
     & - \log\left(100 + \alpha^2\right) + \log \Pi(p),
  \end{align*}
  where \(\beta(\cdot)=\sum_{j=1}^p \beta_j K(t_j, \cdot)\) and \(\dotprod{X_i}{\beta}_K = \sum_{j=1}^p \beta_j X_i(t_j)\).
\end{proposition}

The discrete distribution \(\Pi(p)\) used in the experiments is either uniform in \(\{1,\dots, 10\}\), with density \(\Pi_{\text{uniform}}(p)=1/10\), or a Poisson distribution with rate parameter \(\lambda=3\) truncated to \(\{1,\dots, 10\}\), with density
\[
\Pi_{\text{Poisson}}(p) = \frac{3^p}{Cp!}, \quad p\in\{1,\dots,10\},
\]
where the normalization constant is \(C = \sum_{k=1}^{10} \frac{3^k}{k!}\).

\subsection{Label switching}\label{app:label-switching}

A well-known issue that arises when using MCMC methods in mixture-like models such as the one proposed in this work is \textit{label switching}, which in short refers to the non-identifiability of the components of the model caused by their interchangeability. In our case, this happens because the likelihood and the prior are symmetric with respect to the ordering of the parameters \(b\) and \(\tau\), i.e., \(\mathcal L(Y|X,\theta)=\mathcal L(Y|X, \theta[\nu])\) for any permutation \(\nu\) that rearranges the indices \(j=1,\dots, p\). Thus, since the components are arbitrarily ordered, they may be inadvertently exchanged from one iteration to the next in a MCMC algorithm. This can cause nonsensical answers when summarizing the marginal posterior distributions of the parameters to perform inference, as different labelings might be mixed on each component \citep{stephens2000dealing}. This is primarily the reason why we do not directly use summaries of the posterior distribution of the individual parameters in our prediction methods.

Moreover, the use of trans-dimensional samplers exacerbates this problem, since the change in dimension can further disrupt the internal ordering of the components \citep[see][]{roodaki2014relabeling}. However, this phenomenon is perhaps surprisingly a condition for the convergence of the MCMC method: as pointed out by many authors \citep[e.g.][]{celeux2000computational}, a lack of switching would indicate that not all modes of the posterior distribution were being explored by the sampler. For this reason, many ad-hoc solutions revolve around post-processing and relabeling the samples to eliminate the switching effect, but they generally do not prevent it from happening in the first place.

The most straightforward solutions consist on imposing an artificial identifiability constraint on the parameters to break the symmetry of their posterior distributions; see \citet{jasra2005markov} and references therein. A common approach that seems to work well is to simply enforce an ordering in the parameters in question, which in our case would mean requiring for example that \(\beta_i < \beta_j\) for \(i < j\), or the analogous with the times in \(\tau\). We have implemented a variation of this method described in \citet{simola2021approximate}, which works by post-processing the samples and relabeling the components to satisfy the order constraint mentioned above, choosing either \(b\) or \(\tau\) depending on which set of ordered parameters would produce the largest separation between any two of them (suitably averaged across all iterations of the chains). This is an area of ongoing research, and thus there are other, more complex relabeling strategies, both deterministic and probabilistic. A summary of several such methods can be found for example in \citet{sperrin2010probabilistic} or \citet{rodriguez2014label}.

\subsection{Affine-invariant ensemble samplers}\label{app:ensemble-sampler}

An interesting and often desirable property of regular MCMC sampling algorithms is that they be \textit{affine-invariant}, which informally means that they regard two distributions that differ in an affine transformation, say \(\pi(x)\) and \(\pi_{A, b}(Ax + b)\), as equally difficult to sample from. This is useful when one is working with very asymmetrical or skewed distributions, for an affine transformation can turn them into ones with simpler shapes. Generally speaking, a MCMC algorithm can be described through a function \(R\) as \(\Lambda(t+1)=R(\Lambda(t), \xi(t), \pi)\), where \(\Lambda(t)\) is the state of the chain at instant \(t\), \(\pi\) is the objective distribution, and \(\xi(t)\) is a sequence of i.i.d.\ random variables that represent the random behavior of the chain. With this notation, the affine-invariance property can be characterized as \(R(A\lambda+b, \xi(t), \pi_{A,b}) = AR(\lambda, \xi(t), \pi) + b\), for all \(A,b\) and \(\lambda\), and almost all \(\xi(t)\). This means that if we fix a random generator and run the algorithm twice, one time using \(\pi\) and starting in \(\Lambda(0)\) and a second time using \(\pi_{A,b}\) with initial point \(\Gamma(0)=A\Lambda(0)+b\), then \(\Gamma(t)=A\Lambda(t)+b\) for all \(t\). In \citet{goodman2010ensemble} the authors consider an ensemble of samplers with the affine invariance property. Specifically, they work with a set \(\Lambda=(\Lambda_1, \dots, \Lambda_L)\) of \textit{walkers}, where \(\Lambda_l(t)\) represents an individual chain at time \(t\). At each iteration, an affine-invariant transformation is used to find the next point, which is constructed using the current values of the rest of the walkers (similar to Gibb's algorithm), namely the \textit{complementary ensemble}
\[
  \Lambda_{-l}(t) = \{\Lambda_1(t+1), \dots, \Lambda_{l-1}(t+1), \Lambda_{l+1}(t), \dots, \Lambda_L(t)\}, \quad l=1,\dots, L.
\]

To maintain the affine invariance and the joint distribution of the ensemble, the walkers are advanced one by one following a Metropolis-Hastings acceptance scheme. The authors consider mainly two types of moves:

\begin{description}
  \item[Stretch move.] For each walker \(1\leq l \leq L\) another walker \(\Lambda_j \in \Lambda_{-l}(t)\) is chosen at random, and the proposal is constructed as
  \[
    \Lambda_l(t) \to \Gamma = \Lambda_j + Z(\Lambda_l(t) - \Lambda_j),
  \]
  where \(Z \stackrel{i.i.d.}{\sim} g(z)\) satisfying the symmetry condition \(g(z^{-1})=zg(z)\). In particular, the suggested density is
  \[
    g_a(z) \propto \begin{cases}
      \frac{1}{\sqrt{z}}, & \text{if } z \in [a^{-1}, a], \\
      0,                  & \text{otherwise.}
    \end{cases}, \quad a > 1.
  \]
  Supposing \(\R^p\) is the sample space, the corresponding acceptance probability (chosen so that the detailed balance equations are satisfied) is:
  \[
    \alpha = \min\left\{1, \ Z^{p-1}\frac{\pi(\Gamma)}{\pi(\Lambda_l(t))}\right\}.
  \]

  \item[Walk move.] For each walker \(1\leq l \leq L\) a random subset \(S_l \subseteq \Lambda_{-l}(t)\) with \(|S_l| \geq 2\) is selected, and the proposed move is
  \[
    \Lambda_l(t) \to \Gamma = \Lambda_l(t) + W,
  \]
  where \(W\) is a normal distribution with mean \(0\) and the same covariance as the sample covariance of all walkers in \(S_l\). The acceptance probability in this case is just the Metropolis ratio, namely \(\alpha=\min\{1, \pi(\Gamma)/\pi(\Lambda_l(t))\}\).
\end{description}

From a computational perspective, the Python library \textit{emcee} \citep{foreman2013emcee}  provides a parallel implementation of this algorithm. The idea is to divide the ensemble \(\Lambda\) into two equally-sized subsets \(\Lambda^{(0)}\) and \(\Lambda^{(1)}\), and then proceed on each iteration in the following alternate fashion:
\begin{enumerate}
  \item Update \textit{all} walkers in \(\Lambda^{(0)}\) through one of the available moves explained above, using \(\Lambda^{(1)}\) as the complementary ensemble.
  \item Use the new values in \(\Lambda^{(0)}\) to update \(\Lambda^{(1)}\).
\end{enumerate}
In this way the detailed balance equations are still satisfied, and each of the steps can benefit from the computing power of an arbitrary number of processors (up to \(L/2\)).  

\subsection{RJMCMC implementation}

For the computational implementation of our Bayesian prediction pipeline we chose the Python library \textit{Eryn} \citep{karnesis2023eryn}. This package is a general-purpose suite of MCMC methods which is reliable, easy to use and performs well in a wide class of problems. It is an enhancement of the \textit{emcee} library mentioned in Appendix~\ref{app:ensemble-sampler} that implements affine-invariant ensemble samplers, with a few key improvements:

\begin{description}
  \item[Reversible jump sampling.] The main reason for choosing this library is that it implements a reversible jump MCMC sampling scheme, which is an essential element in our proposed methodology. It allows trans-dimensional sampling for posterior approximation, letting the user select the likelihood, prior and proposal distributions, and providing a great level of control over the details.

  \item[Parallel tempering.] This is a mechanism to increase the efficiency with which the sampler explores the parameter space. The basic idea is to consider a set of Markov chains in parallel, each one sampling from a transformed posterior distribution given by \(\pi_T(p, \theta_p|X,Y) = \pi(Y|X, p, \theta_p)^{1/T}\pi(p, \theta_p)\), where \(T\geq 1\) is the temperature. In the words of \citet{karnesis2023eryn}, \textit{``intermediate temperatures `smooth out' the posterior by reducing the contrast between areas of high and low likelihood''}. In practice, these chains periodically exchange information, with the swaps controlled by an acceptance probability that maintains detailed balance, and ultimately we are only interested on the cold chain (\(T=1\)).

  \item[Multiple try.] Since the trans-dimensional moves are harder to manage and generally give a low acceptance rate, this library allows the proposal of several candidates for a given move, using a weight function to assign them a relative importance, and then choosing from them with probability given by the normalized weights. This naturally increases the computational cost, but it often produces better results.
\end{description}

Another advantage of this ensemble approach, apart from the property of affine-invariance, is that it only requires the specification of a few hyperparameters irrespective of the underlying dimension. This contrasts to, say, the \(O(N^2)\) degrees of freedom corresponding to the covariance matrix of an \(N\)-dimensional jump distribution in Metropolis-Hastings. We already covered the prior, likelihood and posterior distributions used in Section~\ref{sec:methodology} and in Appendix~\ref{app:posterior}. We give below an overview of other integral parts of our RJMCMC method with some implementation details.

\subsubsection*{Initial values}

We need to specify the initial values for the parameters of all the chains in our sampler. In general, we set these values by sampling from the prior distribution of the corresponding parameters. However, in the linear case the prior on \((\alpha, \sigma^2)\) is improper, so we use ad-hoc weakly informative distributions for the initial values instead. For \(\alpha\) we consider a normal distribution with mean \(0\) and standard deviation \(10|\bar{Y}_{\text{scaled}}|\), where \(Y_{\text{scaled}}\) is the version of the original data scaled to have standard deviation unity. For \(\sigma^2\) we use an inverse-gamma distribution with shape parameter \(a=2\) and scale parameter \(b=\hat\sigma^2_Y/\operatorname{Var}(Y)\), where \(\hat \sigma^2_Y\) is a rough estimate of an acceptable error in the scale of \(Y\). This estimation is done is practice as two order of magnitudes less than \(|\bar{Y}|\).

\subsubsection*{Moves}

As with any MCMC method, the Markov chains are advanced iteratively through a set of moves. For the \textit{in-model} moves (that is, moves that do not change the dimension), we divide our parameter vector in two parts:

\begin{description}
  \item[Stretch move.] For the parameters \(\alpha\) and \(\sigma^2\), which are common to all sub-models, we use the stretch move explained in Appendix~\ref{app:ensemble-sampler}.
  \item[Group stretch move.] For \(b\) and \(\tau\), we use a variant of the stretch move known as the group stretch move, which was designed as an extension of the original move that can handle reversible jump setups. The main difference is that the random walker \(\Lambda_j\) used to advance the chain is selected from a stationary group that does not change for several iterations (see Section~3.1 in \citealp{karnesis2023eryn} for more information about this move).
\end{description}
In both moves above, the key scaling parameter \(a\) starts with a value of \(2\), and is changed dynamically to guide the sampler towards an acceptance rate of about \(20\%-30\%\), which is within the range usually recommended in the literature \citep[e.g.][]{rosenthal2011optimal}.

For the \textit{between-model} moves (those which change the dimension), we set an equal probability of births and deaths, except on the end points of the range of \(p\). Specifically, \(b_{p,p+1}=d_{p, p-1}=0.5\) if \(1<p<10\), and \(b_{1,2}=d_{10,9}=1\). Note that these are only the probabilities to propose the corresponding move, which will be accepted or rejected according to the acceptance formula in Section~\ref{sec:rjmcmc}. When the birth of a new component is proposed, we use the prior distribution of \(b\) and \(\tau\) to generate the new values, and when a death is proposed, the method selects an existing component at random as a candidate for deletion.

\subsubsection*{Hyperparameters}

Other relevant hyperparameters include the burn-in period for the chains, which is the number of initial samples discarded, the number of actual steps, the number of chains, and the number of temperatures for parallel tempering. In the experiments we use 64 chains and 10 different temperatures, and run them for 5000 iterations in total, discarding the first 4000 as burn-in. Moreover, when the prior on \(p\) is uniform (i.e.~in the experiments with real data) we activate the multiple try scheme and set the number of tries to \(2\). Lastly, a computational decision we made is working with \(\log \sigma\) instead of \(\sigma^2\) so that the domain of this parameter is an unconstrained space, which is a widespread recommendation that helps increase the sampling efficiency.

\newpage
\section{Proofs of posterior consistency}\label{app:proofs}

In this section we complete the theoretical exposition in Section~\ref{sec:consistency} by providing detailed proofs of the stated results.

\subsection{Consistency via Doob's theorem}

For this part, the general strategy will be to apply Doob's theorem (Theorem~\ref{th:doob}) to a subset of the parameter space \(\Theta\), where permutations are suitably taken into account and full identifiability holds, and then extend the conclusions to the whole parameter space. As we will see shortly, save for an eventual permutation, identifiability of the map \(\theta \mapsto P_\theta(X,Y)\) is obtained in the RKHS case when the covariance function \(K\) of the underlying stochastic process is non-degenerate.

\subsubsection*{Proof of Theorem~\ref{th:consistency-doob-linear}}

\textit{Measurability issues}. We first show the measurability of the mapping \(\theta \mapsto P_\theta(X,Y)(A)\) for every measurable set \(A \subseteq \mathcal X \times \mathcal Y\). Define the function \(h(\theta, x, y)=f_\theta(y|x)\bm{1}_A(x,y)\), where \(f_\theta(\cdot|x)\) is the density of the normal distribution \( \mathcal N(\alpha + \sum_j \beta_j x(t_j),\, \sigma^2)\), and \(\bm{1}_{A}\) is the indicator function of the set \(A\). Let \(E(x,t)=x(t)\) be the evaluation map on \(\mathcal C[0,1]\times [0,1]\), and note that \(x\mapsto x(t)\) is continuous on \(\mathcal C[0,1]\) while \(t\mapsto x(t)\) is continuous on \([0,1]\). It follows that \(E\) is a Carathéodory function and hence jointly measurable \citep[][Lemma~4.51]{aliprantis2006infinite}. Thus, \(h\) is seen to be measurable as a composition of Borel measurable functions, and by Tonelli's theorem \citep[e.g.][Theorem~2.37]{folland1999real} the function
\[
  \theta \mapsto \int_{\mathcal X \times \mathcal Y} h(\theta, x, y)\, d\rho = P_\theta(X, Y)(A)
\]
is measurable. Another measurability concern is the existence of regular conditional distributions such as \(\theta|(X, Y)_{1:n}\). For this, one should see Theorem 10.2.1 and Theorem 10.2.2 in \citet{dudley2002real}, which guarantee they are well-defined provided that the underlying spaces are sufficiently regular.

\textit{Reduced parameter space}. Consider the space \(\R_\delta = (-\infty, -\delta]\cup[\delta, +\infty)\), with \(\delta\) the fixed tolerance given in Condition~\ref{cond:condition-ident}-\ref{cond:condition-ident-2}, and define
\[\tTheta_p = \R^p_\delta \times [0,1]^p_{\text{ord}} \times \R \times \R^+_0,
\]
where \(
[0,1]^p_{\text{ord}} =\{(t_1,\dots, t_p) \in [0,1]^p: t_1 < \cdots < t_p\}
\). Now consider \(\tTheta = \bigcup_{p\in \N}\tTheta_p\), and note that the sets \(\tTheta_1,\tTheta_2,\ldots\) are still disjoint. Then, again by Proposition A.1 in \citet{miller2023consistency}, we can conclude that \(\tTheta\) is a complete separable metric space under the metric \(d_\Theta\) in~\eqref{eq:metric-doob}. We will henceforth say that a parameter in \(\tTheta\) is ``ordered''.

\textit{Transformation into ordered form}. For \(\theta\in\Theta_p\), define \(T(\theta)=\theta[\nu]\) if there is a \(\nu\in S_p\) such that \(\theta[\nu]\in \tTheta_p\); otherwise set \(T(\theta)=\theta\). Note that we can only transform \(\theta\) to be in \(\tTheta_p\) if the times \(t_j\) are all distinct and the coefficients \(\beta_j\) satisfy \(|\beta_j|\geq\delta\). But since the prior distribution assigns probability one to both events by Condition~\ref{cond:condition-ident}, we have \(\Pi(T(\theta)\in\tTheta)=1\). It will be useful later to observe that for any set \(B\subseteq \tTheta_p\), if we denote \(B[\nu] = \{\theta[\nu]: \theta \in B\}\), we have
\begin{equation}\label{eq:inverse-image-T}
  \bigcup_{\nu\in S_p} B[\nu] = T^{-1}(B).
\end{equation}

\textit{Collapsed model}. Let \(\tilde \Pi_T\) denote the distribution of \(T(\theta)\) restricted to \(\tTheta\), and note that we have \(P_{T(\theta)}(X,Y)=P_{\theta}(X,Y)\). Then the following model holds on the reduced space \(\tTheta\):
\begin{equation}\label{eq:model-linear-collapsed}
  \begin{aligned}
     & T(\theta) \sim \tilde{\Pi}_T,                                         \\
     & (X,Y)_{1:n}\mid T(\theta) \sim P_{T(\theta)}(X,Y) \quad \text{i.i.d.}
  \end{aligned}
\end{equation}

\textit{Verifying conditions}. We will now show that the conditions of Doob's theorem hold on \(\tTheta\). First, since \(\theta \mapsto P_\theta(X, Y)(A)\) is measurable on \(\Theta\), it is also measurable on \(\tTheta\) for all sets \(A\subseteq \mathcal X \times \mathcal Y\) measurable. For the identifiability part, suppose by contradiction that there are two parameters \(\theta,\theta' \in \tTheta\) such that \(\theta\neq \theta'\) and \(P_{\theta}(X, Y) = P_{\theta'}(X, Y)\). Then necessarily \(P_{\theta}(Y|X) = P_{\theta'}(Y|X)\), which in turn implies that the means and variances of these distributions are equal, i.e., \(\sigma^2=(\sigma^2)'\) and
\[
  \alpha + \sum_{j=1}^{p(\theta)} \beta_j X(t_j) = \alpha' + \sum_{j=1}^{p(\theta')} \beta'_j X(t'_j),
\]
where \(t_j\neq t_k\) and \(t_j'\neq t_k'\) for \(j\neq k\). Reordering the terms and combining those where the impact points coincide, we have a linear combination of marginals of the process \(X\) that equals a constant. By taking variances, we can see that all the coefficients in this linear combination must vanish, since the covariance function of the process is strictly positive definite. But \(\beta_j\) and \(\beta_j'\) cannot be \(0\) for any \(j\) (by definition of \(\tTheta\)), so it must be the case that \(\theta'=\theta[\nu]\) for some \(\nu \in S_p\), where \(p=p(\theta)=p(\theta')\). However, \(t_1< \cdots < t_p\) and \(t'_1 < \cdots < t'_p\), so  \(\nu\) must be the identity permutation, that is, \(\theta'=\theta\), contradicting the initial assumption.

\textit{Applying Doob's theorem}. Next we analyze the conclusions of Doob's theorem applied to the collapsed model~\eqref{eq:model-linear-collapsed}: there exists \(\tTheta_* \subseteq \tTheta\) with \(\Pi(T(\theta)\in\tTheta_*)=1\) such that, if \(T(\theta_0)\in \tTheta_*\) and we have \((X,Y)_{1:\infty} \sim P_{T(\theta_0)}(X, Y)\) i.i.d., then for any neighborhood \(B\subseteq \tTheta\) of \(T(\theta_0)\) it holds that
\begin{equation}\label{eq:consistency-collapsed}
  \Pi_n(T(\theta) \in B \mid (X,Y)_{1:n}) \xrightarrow[]{n\to\infty} 1 \quad P_{T(\theta_0)}^\infty-\text{a.s.}
\end{equation}
Now define \(\Theta_*\) to be the set of all parameters in \(\Theta\) that can be obtained by permuting a parameter in \(\tTheta_*\), i.e., \(\Theta_* = \bigcup_{p=1}^\infty \bigcup_{\nu\in S_p}(\tTheta_* \cap \tTheta_p)[\nu]\). Then, by~\eqref{eq:inverse-image-T} we have
\[
  \Pi(\theta\in \Theta_*) = \Pi\left(T(\theta) \in \bigcup_{p=1}^\infty (\tTheta_*\cap \tTheta_p)\right) = \Pi(T(\theta)\in \tTheta_*) = 1.
\]

\textit{Extending the result to \(\Theta\)}. Let \(\theta_0\in\Theta_*\), suppose that \((X,Y)_{1:\infty}\sim P_{\theta_0}(X, Y)\) i.i.d., and define \(p_0=p(\theta_0)\) and \(S_0 = S_{p_0}\). Fix \(\epsilon\in(0,1)\) and consider the set \(B\) of all ordered parameters that are within \(\epsilon\) of the ordered version of \(\theta_0\), i.e.,
\begin{equation}\label{eq:set-B}
  B = \left\{ \theta \in \tTheta: d_\Theta(T(\theta_0), \theta) < \epsilon \right\}.
\end{equation}
Observe that, since \(\epsilon < 1\), we have \(B\subseteq \tTheta_{p_0}\) by definition of \(d_\Theta\). Moreover, \(\bigcup_{\nu\in S_{0}}B[\nu]\subseteq \tilde{B}(\theta_0, \epsilon)\). Then, again by~\eqref{eq:inverse-image-T}, we can write
\begin{equation}\label{eq:consistency-intermediate}
  \begin{aligned}
    \Pi_n(\theta \in \tilde{B}(\theta_0, \epsilon) \mid (X,Y)_{1:n}) & \geq \Pi_n(\theta \in \bigcup_{\nu\in S_{0}} B[\nu] \mid (X,Y)_{1:n}) \\
                                                                     & = \Pi_n(T(\theta) \in B \mid (X,Y)_{1:n}).
  \end{aligned}
\end{equation}
Now, \(T(\theta_0)\in{\tTheta_*}\) because \(\theta_0\in\Theta_*\), and in that case we know that the collapsed model is consistent at \(T(\theta_0)\). Note that we also have \((X,Y)_{1:\infty} \sim P_{T(\theta_0)}(X, Y)\) i.i.d.\ (since \(P_{\theta_0}=P_{T(\theta_0)}\)), and the set \(B\) in~\eqref{eq:set-B} is a neighborhood of \(T(\theta_0)\) in \(\tTheta\). Then, by~\eqref{eq:consistency-collapsed} we have \(\Pi_n(T(\theta) \in B | (X,Y)_{1:n}) \xrightarrow[]{n\to\infty} 1, \ P_0^\infty (X,Y)-\text{a.s.}\), and this fact together with~\eqref{eq:consistency-intermediate} proves consistency for \(\theta_0\) in the original model~\eqref{eq:model-linear}. Lastly, since \(\epsilon < 1\) implies \(\tilde{B}(\theta_0,\epsilon) \subseteq \Theta_{p_0}\), we have also proved the second assertion of our theorem:
\begin{align*}
  \Pi_n(\mathcal P=p_0\mid (X,Y)_{1:n}) & = \Pi_n(\theta \in \Theta_{p_0}\mid (X, Y)_{1:n})                                 \\
                                        & \geq \Pi_n(\theta \in\tilde{B}(\theta_0,\epsilon) \mid (X,Y)_{1:n})               \\
                                        & \xrightarrow[n\to\infty]{} 1 \quad P_0^\infty (X,Y)-\text{a.s.}\tag*{\(\square\)} 
\end{align*}

\subsubsection*{Proof of Proposition~\ref{prop:consistency-lebesgue-linear}}

Define \(\Theta_*\) as in the proof of Theorem~\ref{th:consistency-doob-linear} above. Recall that \(\Pi(\Theta_*)=1\), and observe that
\[
  0=\Pi(\Theta \setminus \Theta_*) = \sum_{p=1}^\infty \Pi(\Theta_p \setminus \Theta_*|\mathcal P = p)\Pi(\mathcal P = p).
\]
Since \(\Pi(\mathcal P = p)> 0\) for all \(p\), we have \(\Pi(\Theta_p\setminus \Theta_*|\mathcal P = p) = 0\) for all \(p\). Now, for \(\nu \in S_p\), let \(\mu^\nu_p\) be the distribution of \(\theta[\nu]|\mathcal P=p\) under the model, and note that \((\Theta_p\setminus \Theta_*)[\nu] = \Theta_p \setminus \Theta_*\) by definition of \(\Theta_*\). Thus, if \(id\) is the identity permutation, for all \(\nu\in S_p\) it holds that
\begin{equation}\label{eq:measure-zero}
  \mu^\nu_p (\Theta_p \setminus \Theta_*) = \mu_p^{id} (\Theta_p \setminus \Theta_*) = \Pi (\Theta_p \setminus \Theta_*\mid\mathcal P = p) =0.
\end{equation}
Lastly, \(\lambda_p \ll \sum_{\nu\in S_p}\mu^\nu_p\) by assumption, where \(\ll\) denotes absolute continuity, and this together with~\eqref{eq:measure-zero} implies that \(\lambda_p(\Theta_p \setminus \Theta_*)=0\). But this is valid for all \(p\in\N\), so we can conclude that the inconsistency set satisfies \(\lambda_\infty(\Theta\setminus \Theta_*) = \sum_{p=1}^\infty \lambda_p(\Theta_p\setminus \Theta_*)=0\), as claimed.\qed{}

\subsubsection*{The case of fixed dimension}

As we said earlier, the proofs above can be modified in a straightforward way to establish consistency when the number of components \(p\) is fixed.

\begin{corollary}
  Assume model~\eqref{eq:model-linear} with a fixed value of \(p\), where \(\Theta_p\) is the (finite-dimensional) parameter space. If Condition~\ref{cond:condition-ident}-\ref{cond:condition-ident-1} holds, then the posterior is consistent at \(\theta_0\in\Theta_p\) with \(\Pi\)-probability one. Moreover, if condition (ii) in Proposition~\ref{prop:consistency-lebesgue-linear} holds as well, then the inconsistency set has Lebesgue measure zero.
\end{corollary}

In this case Doob's theorem applies directly under the sole condition that the times be distinct with prior probability one. The coefficients \(\beta_j\) do not cause a problem for identifiability now, since the dimension of every parameter is the same. Note that by allowing \(\beta_j\) to be zero we can circumvent the fact that the true value of the parameter might not have exactly \(p\) components, as long as \(p\) is larger than the true value \(p(\theta_0)\). Indeed, if \(\theta_0\in\Theta\) with \(p(\theta_0) < p\) and \((X,Y)_{1:\infty} \sim P_{\theta_0}(X,Y)\) i.i.d., then we can find \(\theta_1\in\Theta_p\), which is just \(\theta_0\) completed with zeros, such that \(P_{\theta_0}\)(X, Y) = \(P_{\theta_1}(X,Y)\) and the result holds almost surely.

\subsection{Consistency and contraction rates via Schwartz's theorem}

In what follows, we introduce the notation \(a\lesssim b\) to mean inequality up to a multiplicative positive constant. We write \(g(n)=O(h(n))\) for positive functions \(g\) and \(h\) if there exist \(C,n_0>0\) such that \(g(n)\leq Ch(n)\) for all \(n\geq n_0\). Intuitively, this means that \(g(n)\) is of order smaller than or equal to \(h(n)\), i.e., \(\lim\sup_{n\to\infty} g(n)/h(n) < \infty\). Moreover, we say that \(g(n)=o(h(n))\) if for every \(c>0\) there exists \(n_0\in\mathbb{N}\) such that \(g(n)\leq ch(n)\) for all \(n\geq n_0\). Equivalently, this means that \(g(n)\) grows much slower than \(h(n)\), i.e., \(\lim_{n\to\infty} g(n)/h(n) = 0\).

Additionally, when the process \(X\) is understood as a mapping \(X:\Omega \to \mathcal C[0,1]\), we will make frequent use of the correspondence \(\mathbb E[g(X)]=\mathbb E_{Q_X}[g]\) for a measurable function \(g\) on \(\mathcal C[0,1]\), where the latter quantity is an integral in the Bochner sense. In a similar vein, if \(e_t:\mathcal C[0,1]\to \mathbb R\) is the evaluation functional \(x\mapsto x(t)\), given \(h:\mathbb R \to \mathbb R\) we will denote \(\mathbb E_{Q_X}[h\circ e_t]\) by \(\mathbb E_{Q_X}[h(x(t))]\). In particular, for \(h=id\) we have
\[
  \mathbb E_{Q_X}[x(t)]:=\mathbb E_{Q_X}[e_t]=\mathbb E[e_t(X)]=\mathbb E[X(t)].
\]

Next, we present a collection of standard results on covering numbers that will be extensively used in the subsequent proofs \citep[see e.g.][]{van1996weak}.

\begin{lemma}\label{lemma:covering-number}
  Suppose \(d\) is a distance on the set of densities \(\mathcal F=\{f_\theta: \theta\in\Theta\}\).
  \begin{enumerate}[label=(\roman*)]
    \item If \(\epsilon\leq \epsilon'\) then \(N(\epsilon', \mathcal F, d) \leq N(\epsilon, \mathcal F, d)\).
    \item If \(d'\) is another distance on \(\mathcal F\) such that \(d\leq \tilde C d'\), then \(N(\epsilon, \mathcal F, d) \leq N(\epsilon/\tilde C, \mathcal F, d')\) for all \(\epsilon > 0\).
    \item If there is a constant \(\tilde C>0\) such that \(d(f_{\theta}, f_{\theta'})\leq \tilde C \|\theta - \theta'\|\) for every \(\theta, \theta' \in  \Theta\), then the corresponding covering numbers satisfy \(N(\epsilon, \mathcal F, d)\leq N(\epsilon/\tilde C, \Theta, \|\cdot\|)\) for all \(\epsilon > 0\). 
    \item If \(\Theta_p \subset \R^p\) has finite diameter \(R_p\), then \(N(\epsilon, \Theta_p, \|\cdot\|) \leq (3R_p/\epsilon)^p\) for all \(0<\epsilon<R_p\). 
  \end{enumerate}
\end{lemma}

Recall that our full parameter space is the infinite union \(\Theta=\bigcup_{p\in\N}\Theta_p\). As we said in Section~\ref{sec:consistency-schwartz}, for convenience we take the finite-dimensional sets \(\Theta_p\) to be compact spaces. More specifically, we work with
\[
  \Theta_p = \{(b, \tau, \alpha, \sigma^2): b\in [-M_B, M_B]^p, \tau \in [0,1]^p, \alpha\in [-M_A, M_A], \sigma^2 \in [\sigma^2_{\text{min}}, \sigma^2_{\text{max}}]\},
\]
where \(M_A,M_B>0\) and \(0<\sigma^2_{\text{min}}<\sigma^2_{\text{max}}\). For a given parameter \(\theta\in\Theta\) we consider the function \(\mu_\theta:\mathcal C[0,1]\to\mathbb R\) given by \(\mu_\theta(x)=\alpha + \sum_{j=1}^{p(\theta)} \beta_j x(t_j)\), which allows us to express the density of \(Y|X=x, \theta\) in our linear model as the normal density
\[
  f_\theta(y\mid x) = \frac{1}{\sqrt{2\pi\sigma^2}}\exp\left\{-\frac{(y-\mu_\theta(x))^2}{2\sigma^2}\right\}.
\]

\subsubsection*{Proof of Theorem~\ref{th:consistency-schwartz-linear}}

We need to verify the conditions of Theorem~\ref{th:consistency-hellinger}. As was previously announced, we will be using the sieve of parameters \(\Theta_n = \bigcup_{p=1}^{p_n} \Theta_p\), and the corresponding sieve of densities \(\mathcal F_n=\{f_\theta\in \mathcal F: \theta \in \Theta_n\}\). The idea is to select a suitable growth order for the upper limit \(p_n\) in the sieve so that all conditions hold.

\noindent\textit{Condition (i): bounding the metric entropy}

\noindent First, we claim that there exists \(C^{(p)}>0\) so that \(d_H(f_\theta, f_{\theta'})\leq C^{(p)}\|\theta - \theta'\|\) for all \(\theta, \theta'\) in a fixed (finite-dimensional) \(\Theta_p\). Starting with the function \(\mu_\theta\) and a given trajectory \(x\) of \(X\), rewrite the absolute difference in means as
\[
  |\mu_\theta(x) - \mu_{\theta'}(x)| = \left| (\alpha - \alpha') + \sum_{j=1}^p \left(\beta_j (x(t_j) - x(t'_{j}))+ (\beta_j - \beta'_{j})x(t'_{j})\right) \right|.
\]
Now from the triangle inequality, the Cauchy-Schwarz inequality and the fact that \((a+b)^2 \leq 2(a^2 + b^2)\), the following holds \(Q_X\)-a.s.:
\[
  \frac{1}{2}(\mu_\theta(x) - \mu_{\theta'}(x))^2 \leq (\alpha - \alpha')^2 + 2p\left(M_B^2\sum_{j=1}^p |x(t_j) - x(t_j')|^2 + \sum_{j=1}^p x(t'_j)^2|\beta_j - \beta'_j|^2\right).
\]
Note that \(\|\theta-\theta'\|^2\) is always bigger than the sum of just some of the squared components of \(\theta - \theta'\). Using the Lipschitz property in Condition~\ref{cond:condition-schwartz-X} and integrating on both sides with respect to \(Q_X\), we get
\[
  \frac{1}{2}\mathbb E_{Q_X}\left[(\mu_\theta - \mu_{\theta'})^2\right] \leq \|\theta - \theta'\|^2 + 2p\left(M_B^2 \|\theta - \theta'\|^2 \mathbb{E}_{Q_X}[L^2] + \|\theta - \theta'\|^2 \sup_{t\in[0,1]} \mathbb{E}_{Q_X}[x(t)^2] \right).
\]
Recall that the covariance function \(K\) of \(X\) is assumed to be continuous, so the quantity \(\sup_{t\in[0,1]} \mathbb{E}_{Q_X}[x(t)^2]=\sup_{t\in[0,1]}K(t,t)\) is finite, and so is the integral \(\mathbb{E}_{Q_X}[L^2]\) owing to Condition~\ref{cond:condition-schwartz-X}. Then, we can bring all constants together (which are positive and depend on \(M_B\), \(p\) and \(X\)) into \(C_1>0\) to write
\begin{equation}\label{eq:th9-bound1}
  \mathbb{E}_{Q_X}\left[(\mu_\theta - \mu_{\theta'})^2\right] \leq C_1 \|\theta - \theta'\|^2.
\end{equation}

Now we compute \(d^2_H(f_\theta, f_{\theta'})\). Denote by \(\sigma^2_1\) and \(\sigma^2_2\) the variance parameters corresponding to \(\theta\) and \(\theta'\), respectively. Using the inequality \(e^{-a}\geq 1-a\) for \(a\geq 0\) and the well-known expression for the Hellinger distance between two univariate normal distributions \citep[e.g.][p.~51]{pardo2018statistical}, we have
\begin{align}\label{eq:th9-bound2}
  d^2_H(f_\theta, f_{\theta'}) & = 1 - \int \left(\int\sqrt{f_\theta(y|x) f_{\theta'}(y|x)}\, d\lambda(y)\right) dQ_X(x) \notag                                                                                                                                \\
                               & = 1 - \sqrt{\frac{2\sigma_1\sigma_2}{\sigma_1^2 + \sigma_2^2}}\bigintssss \exp\left\{ -\frac{(\mu_\theta(x) - \mu_{\theta'}(x))^2}{4(\sigma_1^2 + \sigma_2^2)} \right\} dQ_X(x) \notag                                        \\
                               & \leq 1 - \sqrt{\frac{2\sigma_1\sigma_2}{\sigma_1^2+\sigma_2^2}} + \sqrt{\frac{2\sigma_1\sigma_2}{\sigma_1^2 + \sigma_2^2}} \bigintssss \frac{(\mu_\theta(x) - \mu_{\theta'}(x))^2}{4(\sigma_1^2+\sigma_2^2)}\, dQ_X(x) \notag \\
                               & \leq C_2\|\theta - \theta'\|^2 + C_3 \mathbb E_{Q_X}\left[(\mu_\theta - \mu_{\theta'})^2\right],
\end{align}
where \(C_2\) and \(C_3\) are positive constants that depend on \(\sigma^2_{\text{min}}\) and \(\sigma^2_{\text{max}}\). Obtaining \(C_3\) is straightforward; to get \(C_2\), let \(u=\sigma_1^2\), \(v=\sigma_2^2\) and \(S=\sqrt{\frac{2\sqrt{uv}}{u+v}}\), and note that
\[
  1 - \sqrt{\frac{2\sigma_1\sigma_2}{\sigma_1^2 + \sigma_2^2}} = 1-S = \frac{1-S^2}{1+S} \leq 1 - S^2.
\]
Now, a bit of algebraic manipulation yields
\begin{equation}\label{eq:th9-computation-S}
  1 - S^2 = 1 - \frac{2\sqrt{uv}}{u+v} = \frac{(\sqrt{u}-\sqrt{v})^2}{u+v} = \frac{(u-v)^2}{(u+v)(\sqrt{u} + \sqrt{v})^2},
\end{equation}
and after reverting the change of variables the denominator can be easily bounded in terms of \(\sigma^2_{\text{min}}\) and \(\sigma^2_{\text{max}}\). To conclude, combine~\eqref{eq:th9-bound2} with the previous bound~\eqref{eq:th9-bound1} to get the desired inequality
\begin{equation}\label{eq:th9-bound3}
  d^2_H(f_\theta, f_{\theta'}) \leq C^{(p)} \|\theta - \theta'\|^2,
\end{equation}
where \(C^{(p)}>0\) increases linearly with \(p\).

Next, consider the distance \(d_\Theta\) on \(\Theta\) given in~\eqref{eq:metric-doob}. From~\eqref{eq:th9-bound3} one can readily verify that the Hellinger distance satisfies \(d_H(f_\theta, f_{\theta'}) \leq C(n) d_\Theta(\theta, \theta')\) for \(\theta, \theta' \in \Theta_n\), where \(C(n)=O(p_n)\). Moreover, the relationship \(d_\Theta(\theta,\theta')\leq \|\theta - \theta'\|\) always holds on \(\Theta_p\). Denote by \(R_p\) the diameter of \((\Theta_p, \|\cdot\|)\), and note that each \(\Theta_p\) is contained in a hypercube of dimension \(2p+2\) and length independent of \(p\), say \(\ell\). Then, \(R_p\leq \ell\sqrt{2p+2}\) for all \(p\). Putting it all together and invoking Lemma~\ref{lemma:covering-number}, for sufficiently small \(\epsilon>0\) we have
\begin{align*}
  N(\epsilon, \mathcal F_n, d_H) & \leq N(\epsilon/C(n), \Theta_n, d_\Theta)                               \\
                                 & \leq \sum_{p=1}^{p_n} N(\epsilon/C(n), \Theta_p, \|\cdot\|)             \\
                                 & \leq p_n \left(\frac{3C(n)\ell\sqrt{2p_n+2}}{\epsilon}\right)^{2p_n+2}.
\end{align*}
Taking logarithms, we may write
\begin{equation}\label{eq:th9-bound-covering}
  \log N(\epsilon, \mathcal F_n, d_H) \leq \log p_n + (2p_n+2) \log\left(\frac{3C(n) \ell\sqrt{2p_n+2}}{\epsilon}\right).
\end{equation}
Since \(C(n)=O(p_n)\), it suffices to choose \(p_n\) so that \(p_n\log p_n = o(n)\), or equivalently, \(p_n=o\left(n/\log n\right)\). In this way the right-hand side is eventually less than \(n\epsilon^2\), satisfying the first condition of Theorem~\ref{th:consistency-hellinger}.\\

\noindent\textit{Condition (ii): bounding the prior tail mass}

\noindent We need to show that the tail mass of the prior decays at least exponentially with \(n\). An initial crude bound on this tail mass in the parameter space is
\[
  \Pi(\theta \in \Theta \setminus \Theta_n) = \Pi\left(\theta \in \bigcup_{p=p_n + 1}^\infty \Theta_p\right)  = \sum_{p=p_n+1}^\infty \Pi(\theta \in \Theta_p) \leq \sum_{p=p_n+1}^{\infty} \Pi(p).
\]
Since \(\Pi(p)\lesssim \exp\{-\delta p(\log p)^k\}\) for sufficiently large \(p\) and multiplicative constants do not affect the bound, we can work without loss of generality with \(\Pi(p)=\exp\{-\delta p(\log p)^k\}\). Then, we have
\begin{align*}
  \Pi(\Theta \setminus \Theta_n) & \leq \sum_{p=p_n+1}^\infty \exp\{-\delta p(\log p)^k\}                                \\
                                 & = \sum_{m=1}^\infty \exp\{-\delta (p_n+m)(\log (p_n+m))^k\}                           \\
                                 & \leq \sum_{m=1}^\infty \exp\{-\delta p_n (\log p_n)^k\}\exp\{-\delta m (\log p_n)^k\} \\
                                 & = \exp\{-\delta p_n (\log p_n)^k\} \sum_{m=1}^\infty r_n^m,
\end{align*}
where \(r_n=\exp\{-\delta (\log p_n)^k\}<1\). Summing the geometric series, we get
\begin{equation}\label{eq:proof-infinite-prior-tail-bound}
  \Pi(\Theta \setminus \Theta_n) \leq \exp\{-\delta p_n (\log p_n)^k\} \cdot \frac{\exp\{-\delta (\log p_n)^k\}}{1-\exp\{-\delta (\log p_n)^k\}}.
\end{equation}
Now choose \(p_n \geq \frac{c_1n}{(\log n)^k}\) with \(c_1>0\), and observe that for sufficiently large \(n\) we have
\[
  \log p_n \geq \log(c_1) + \log n - k\log \log n \geq \frac{1}{2}\log n.
\]
In particular, \(p_n\to\infty\). Hence, the fraction in~\eqref{eq:proof-infinite-prior-tail-bound} goes to zero as \(n\) grows, and there exists a constant \(c_2>0\) such that, for large \(n\),
\[
  \frac{\exp\{-\delta (\log p_n)^k\}}{1-\exp\{-\delta (\log p_n)^k\}}\leq c_2.
\]
Then, using the lower bounds on \(p_n\) and \(\log p_n\) it follows that eventually
\begin{align*}
  \Pi(\Theta \setminus \Theta_n) & \leq c_2 \exp\{-\delta p_n (\log p_n)^k\}                                                    \\
                                 & \leq c_2 \exp\left\{-\delta \frac{c_1 n}{(\log n)^k}\left(\frac{1}{2}\log n\right)^k\right\} \\
                                 & = c_2 \exp\left\{-\tilde{C}n\right\}.
\end{align*}
As we said before, multiplicative constants do not significantly affect the bound, so we can absorb \(c_2\) into the exponential. Indeed, observe that \(c_2\leq e^{\xi n}\) for any \(\xi>0\) and \(n\) sufficiently large. Then, we have \(c_2e^{-\tilde{C}n} \leq e^{-(\tilde{C}-\xi)n}\), so it suffices to choose \(\xi\) small enough so that \(C=\tilde{C}-\xi>0\). We still need to translate the bound to the prior on \(\mathcal F \setminus \mathcal F_n\), but this is immediate, since
\[
  \Pi_{\mathcal F}(\mathcal F \setminus \mathcal F_n) = \Pi(\theta\in\Theta: f_\theta \in \mathcal F \setminus \mathcal F_n) \leq \Pi(\Theta \setminus \Theta_n).
\]

Note that the lower bound \(p_n \gtrsim \frac{n}{(\log n)^k}\) does not contradict the upper bound \(p_n=o(n/\log n)\) needed for condition (i) above, so we can choose an appropriate growth rate for \(p_n\) in our sieve to make both conditions hold simultaneously.\\

\noindent\textit{Condition on \(f_0\): verifying the KL support}

\noindent First, we proceed by directly computing \(D_{\mathrm{KL}}(f_0 \,\|\, f_\theta)\) for \(\theta\in\Theta_0\), which again has a familiar formula in the normal case \citep[e.g.][p.~47]{pardo2018statistical}, and bounding the result by a similar bound as in condition (i) above:
\begin{align}\label{eq:th9-bound4}
  D_{\mathrm{KL}}(f_0 \,\|\, f_\theta) & =  \int \left(f_0(y|x)\log\left(\frac{f_0(y|x)}{f_\theta(y|x)}\right)  d\lambda(y)\right)dQ_X(x)  \notag                                                              \\
                                       & = \frac{1}{2}\left(\frac{\sigma_0^2}{\sigma^2} - 1 - \log \frac{\sigma_0^2}{\sigma^2} \right) + \frac{1}{2\sigma^2}\int (\mu_0(x) - \mu_\theta(x))^2\, dQ_X(x) \notag \\
                                       & \leq C_1\|\theta_0 - \theta\|^2 + C_2\|\theta_0-\theta\|^2 \notag                                                                                                           \\
                                       & = \tilde C\|\theta_0 - \theta\|^2.
\end{align}
To obtain \(C_1\), let \(u=\sigma_0^2/\sigma^2\) and \(I=[\sigma^2_{\text{min}}/\sigma^2_\text{max},\, \sigma^2_{\text{max}}/\sigma^2_\text{min}]\), and consider the function \(g(u)=u - 1 - \log u\) for \(u\in I\). Define the auxiliary function \(h(u) = g(u)/(u-1)^2\) for \(u\neq 1\) and \(h(1)=1/2\). It is immediate to show that \(h\) is continuous on \(I\), and thus it attains its maximum value
\[H :=\max_{u \in I} h(u)\geq 1/2 > 0.
\]
It follows that \(g(u) \leq H (u-1)^2\) for \(u\in I\), which gives us the desired bound after undoing the change of variables. Then, from~\eqref{eq:th9-bound4} and the behavior of \(\Pi\) on \(\Theta_0\) we can conclude that \(f_0\in\operatorname{KL}(\Pi_{\mathcal F})\), since for every \(\epsilon>0\) we have
\begin{align*}
  \Pi_{\mathcal F}(f_\theta\in \mathcal F: D_{\mathrm{KL}}(f_0 \,\|\, f_\theta) < \epsilon) & =\Pi(\theta \in \Theta: D_{\mathrm{KL}}(f_0 \,\|\, f_\theta) < \epsilon)      \\
                                                                                            & \geq \Pi(\theta \in \Theta_0: D_{\mathrm{KL}}(f_0 \,\|\, f_\theta) <\epsilon) \\
                                                                                            & \geq \Pi(\theta\in\Theta_0: \|\theta_0 - \theta\|^2 < \epsilon/\tilde C) > 0.
\end{align*}

\subsubsection*{Proof of Theorem~\ref{th:contraction-rate-schwartz-linear}}

We will show that the sequence \(\epsilon_n=Dn^{-\gamma}\) with \(0<\gamma<1/2\) and a certain \(D>0\) satisfies the conditions of Theorem~\ref{th:contraction-rate}, and thus is a posterior contraction rate for our model. Note that the specific value of \(D\) is of no relevance, since it is absorbed into the diverging sequence \(M_n\) in the definition of a contraction rate.

This time we will consider the same type of sieve as before, namely \(\Theta_n = \bigcup_{p=1}^{p_n}\Theta_p\), but with an upper limit given by
\[
  p_n=\left\lfloor\frac{n^{1-2\gamma}}{\log n}\right\rfloor,
\]
where \(\lfloor \cdot \rfloor\) is the floor function. As we will see, the conditions imposed in the statement of Theorem~\ref{th:contraction-rate-schwartz-linear} are precisely the ones needed to ensure that the corresponding sieve of densities \(\mathcal F_n=\{f_\theta\in\mathcal F: \theta \in \Theta_n\}\) satisfies the requirements of Theorem~\ref{th:contraction-rate}.

\noindent\textit{Condition (i): bounding the metric entropy}

\noindent From the reasoning in the proof of Theorem~\ref{th:consistency-schwartz-linear} above (see~\eqref{eq:th9-bound-covering}), we have the bound
\[
  \log N(\epsilon_n/2, \mathcal F_n, d_H)  \leq \log p_n + (2p_n+2)\log\left(\frac{6C(n) \ell\sqrt{2p_n+2}}{\epsilon_n}\right).
\]
Since \(\epsilon_n=Dn^{-\gamma}\), \(C(n)\lesssim  p_n\) eventually, and \(\sqrt{2p_n+2}\leq \sqrt{4p_n}\) for \(p_n\geq 1\), we can absorb all constants inside the second logarithm into \(c_1>0\). Indeed, for sufficiently large \(n\) we have
\begin{equation}\label{eq:th11-bound1}
  \log N(\epsilon_n/2, \mathcal F_n, d_H)  \leq \log p_n + (2p_n+2)\log\left(c_1 p_n^{3/2}n^\gamma\right).
\end{equation}
Then, given that \(p_n\leq \frac{n^{1-2\gamma}}{\log n}\), substituting into~\eqref{eq:th11-bound1} yields
\begin{align*}
  \log N(\epsilon_n/2, \mathcal F_n, d_H) & \leq \log \left(\frac{n^{1-2\gamma}}{\log n}\right)+ \frac{2n^{1-2\gamma}}{\log n}\left[\log c_1 + \log\left(\frac{n^{3/2-2\gamma}}{(\log n)^{3/2}}\right)\right] \\
                                          & \quad + 2\left[\log c_1 + \log\left(\frac{n^{3/2-2\gamma}}{(\log n)^{3/2}}\right)\right].
\end{align*}
Simplifying the expression in the right-hand side, the dominant term for large \(n\) is
\[
  2\left(\frac{3}{2}-2\gamma\right)n^{1-2\gamma} = (3-4\gamma)n^{1-2\gamma},
\]
and the rest of the terms are of order strictly smaller. This means that there exists a constant \(c_2>0\) as small as we want such that every term apart from the dominant term is eventually smaller than \(c_2 n^{1-2\gamma}\). Thus, for sufficiently large \(n\) we can write
\[
  \log N(\epsilon_n/2, \mathcal F_n, d_H)  \leq  (c_2 + 3-4\gamma)n^{1-2\gamma}.
\]
Since \(n\epsilon_n^2 =D^2n^{1-2\gamma}\), the conclusion follows for \(D^2>3-4\gamma\).\\

\noindent\textit{Condition (ii): bounding the prior tail mass}

\noindent We follow a similar reasoning as in the proof of Theorem~\ref{th:consistency-schwartz-linear} above. Consider without loss of generality that \(\Pi(p)=c_1\exp\{-\delta p \log p\}\). Then, we have
\begin{align*}
  \Pi(\Theta \setminus \Theta_n) & \leq \sum_{p>p_n} c_1\exp\{-\delta p \log p\}                                                        \\
                                 & = c_1\sum_{m=1}^\infty \exp\{-\delta (p_n+m) \log(p_n+m)\}                                           \\
                                 & \leq c_1\sum_{m=1}^\infty  \exp\{-\delta(p_n+m)\log p_n\}                                            \\
                                 & \leq c_1\exp\{-\delta p_n \log p_n\} \sum_{m=1}^\infty \left(\exp\{-\delta \log p_n\}\right)^m       \\
                                 & = c_1\exp\{-\delta p_n \log p_n\} \cdot \frac{\exp\{-\delta \log p_n\}}{1-\exp\{-\delta \log p_n\}}.
\end{align*}
Now, since \(p_n\geq \frac{n^{1-2\gamma}}{2\log n}\), for large \(n\) we have
\begin{align*}
  \log p_n & \geq (1-2\gamma)\log n - \log (2\log n) \geq \frac{1-2\gamma}{2}\log n.
\end{align*}
In particular, \(p_n\to\infty\). Hence, there exists a constant \(c_2>0\) such that
\[
  \frac{\exp\{-\delta \log p_n\}}{1-\exp\{-\delta \log p_n\}}\leq c_2
\]
for large \(n\). Then, using the lower bounds on \(p_n\) and \(\log p_n\) it follows that
\[
  \Pi(\Theta \setminus \Theta_n) \leq c_1c_2 \exp\left\{\frac{-\delta (1-2\gamma)}{4}n^{1-2\gamma}\right\}.
\]
At this point we integrate \(c_1c_2\) into the exponential. Let \(c_3=\delta(1-2\gamma)/4>0\) and choose \(0<\xi<c_3\) such that \(c_1c_2\leq \exp\{\xi n^{1-2\gamma}\}\) for sufficiently large \(n\). Then, if \(c_4=c_3-\xi\), we have \(\Pi(\Theta \setminus \Theta_n) \leq \exp\{-c_4n^{1-2\gamma}\}\). Now, observe that \((C+4)n\epsilon_n^2=(C+4)D^2n^{1-2\gamma}\), where \(C\) is the constant appearing in condition (ii) of Theorem~\ref{th:contraction-rate}, so it suffices to choose \(C\leq c_4/D^2 -4\) and we are done. Given the lower bound on \(D\) imposed in condition (i) above, it can be easily verified that this upper bound on \(C\) is positive as long as \(\xi\) is chosen small enough and \(\delta > 16(3-4\gamma)/(1-2\gamma)\), which is true by assumption.\\

\noindent\textit{Condition (iii): sufficient prior mass on \(B_2(f_0, \epsilon_n)\)}

\noindent Let \(\theta\in\Theta_0\). We know from the proof of Theorem~\ref{th:consistency-schwartz-linear} (see~\eqref{eq:th9-bound4}) that there exists a constant \(C_1>0\) such that \(D_{\mathrm{KL}}(f_0 \,\|\, f_\theta)\leq C_1^2 \|\theta_0 - \theta\|^2\). On the other hand, direct calculation \citep[see e.g.][]{choi2007posterior} shows that, under Gaussianity,
\[
  V_2(f_0, f_\theta) = \frac{(\sigma_0^2 - \sigma^2)^2}{2\sigma^4} + \frac{\sigma_0^2}{\sigma^4}\int (\mu_0(x) - \mu_\theta(x))^2\, dQ_X(x).
\]
Following the same reasoning as we did multiple times before, we can see that \(V_2(f_0, f_\theta)\) is bounded above by \(C_2^2\|\theta_0 - \theta\|^2\) with \(C_2>0\). Therefore, if \(\|\theta_0 - \theta\| < \epsilon_n/\tilde C\) where \(\tilde C = \max\{C_1, C_2\}\), it follows that \(D_{\mathrm{KL}}(f_0 \,\|\, f_\theta)< \epsilon_n^2\) and \(V_2(f_0, f_\theta) < \epsilon_n^2\). Thus, if we denote by \(B(\theta_0, \epsilon)\) the open ball of radius \(\epsilon\) centered at \(\theta_0\) in \(\Theta_0\), we have just shown that \(\theta \in B(\theta_0,\epsilon_n/\tilde{C})\) implies \(f_\theta \in B_2(f_0, \epsilon_n)\), and hence
\begin{equation}\label{eq:th11-bound2}
  \Pi_{\mathcal F}(B_2(f_0, \epsilon_n)) \geq \Pi(B(\theta_0,\epsilon_n/\tilde{C})).
\end{equation}

Now, to get the final lower bound required for condition (iii) in Theorem~\ref{th:contraction-rate}, we will use the following auxiliary result.

\begin{lemma}\label{lemma:th11-lemma-lebesgue}
  If the prior \(\Pi\) has a density on \(\Theta_0\) with respect to Lebesgue measure that is bounded away from zero in a neighborhood of \(\theta_0\), then there is a constant \(c>0\) such that \(\Pi(B(\theta_0, r))\geq cr^{2p_0+2}\) for all sufficiently small \(r>0\), where \(p_0=p(\theta_0)\).
\end{lemma}
\begin{proof}
  Let \(g\) be the density of \(\Pi\) restricted to \(\Theta_0\), and let \(r_0>0\) be a radius such that \(g(\theta)\geq g_0>0\) for all \(\theta\in B(\theta_0, r_0)\). Then, for every \(0<r\leq r_0\) we have
  \[
    \Pi(B(\theta_0, r)) =\int_{B(\theta_0, r)} g\, d\lambda_0 \geq g_0 \lambda_0(B(\theta_0, r)) = g_0 c_d r^d,
  \]
  where \(d=2p_0+2\), \(\lambda_0\) is the Lebesgue measure on \(\Theta_0\) and \(c_d>0\) is a constant that depends only on the dimension of \(\Theta_0\).
  \end{proof}
  
  Applying Lemma~\ref{lemma:th11-lemma-lebesgue} with \(r=\epsilon_n/\tilde C=Dn^{-\gamma}/\tilde C\) and taking logarithms in~\eqref{eq:th11-bound2}, we have, for large \(n\),
  \[
  \log \Pi_{\mathcal F}(B_2(f_0, \epsilon_n)) \geq \log c + (2p_0+2)\log\left(\frac{Dn^{-\gamma}}{\tilde C}\right) \geq \tilde c -2\gamma p_0\log n \geq -4\gamma p_0\log n,
  \]
where \(\tilde c > 0\) is an irrelevant constant in the asymptotic regime. To finish the proof we want to find a constant \(C>0\) such that \(\log \Pi_{\mathcal F}(B_2(f_0, \epsilon_n)) \geq - Cn\epsilon_n^2=-CD^2n^{1-2\gamma}\), but \(\log n\) grows much slower than \(n^{1-2\gamma}\), so the inequality \(-4\gamma p_0\log n \geq -CD^2n^{1-2\gamma}\) holds for any choice of \(C>0\) and sufficiently large \(n\).

\subsubsection*{Proof of Corollary~\ref{th:corollary-contraction-rate}}

We simply use the fact that, if \(d\) is another distance on \(\mathcal F\) with \(d\leq \tilde C d_H\), then
\[
  \Pi_n(\theta: d(f_0, f_\theta)>\epsilon\mid \text{data}) \leq \Pi_n(\theta: d_H(f_0, f_\theta)>\epsilon/\tilde C\mid \text{data}).
\]

For the \(L^1\)-distance, by the Cauchy-Schwarz inequality we have
\begin{align*}
  d_1^2(f_0, f_\theta) & = \left(\int |f_0 - f_\theta|\, d\rho\right)^2                                                                          \\
                       & = \left(\int \left|\sqrt{f_0} - \sqrt{f_\theta}\right|\left(\sqrt{f_0} + \sqrt{f_\theta}\right) d\rho\right)^2          \\
                       & \leq \int\left(\sqrt{f_0} - \sqrt{f_\theta}\right)^2 d\rho \cdot \int \left(\sqrt{f_0} + \sqrt{f_\theta}\right)^2 d\rho \\
                       & \leq 8 d^2_H(f_0, f_\theta).
\end{align*}
In the last step we have used the alternative definition of the Hellinger distance, which can be expressed as \(2d_H^2(f_0, f_\theta) = \int\left(\sqrt{f_0} - \sqrt{f_\theta}\right)^2 d\rho \), the basic inequality \((a + b)^2 \leq 2(a^2 + b^2)\), and the fact that both \(f_0\) and \(f_\theta\) integrate to \(1\) over the whole space.

For the mean-variance discrepancy distance, recall the expression for the squared Hellinger distance between two normal densities, which can be formulated in this case as \(d^2_H(f_0, f_\theta)= \int g_{0,\theta}(x)\, dQ_X(x)\), with
\[
  g_{0,\theta}(x) = 1 - \sqrt{\frac{2\sigma_0\sigma}{\sigma_0^2 + \sigma^2}}\exp\left\{ -\frac{(\mu_0(x) - \mu_{\theta}(x))^2}{4(\sigma_0^2 + \sigma^2)} \right\}.
\]
Define \(\Delta_\mu = \mu_0-\mu_\theta\), \(\Delta_{\mu_M} = [\mu_0]_M - [\mu_\theta]_M\) and \(S=\sqrt{\frac{2\sigma_0\sigma}{\sigma_0^2+\sigma^2}}\). On the one hand, using the trivial bounds \(e^{-a}\leq 1\) and \(S\leq 1\), we have
\[
  g_{0,\theta}(x) \geq 1 - S = \frac{1-S^2}{1+S} \geq \frac{1-S^2}{2} \geq C_1(\sigma_0^2 - \sigma^2)^2,
\]
where the last inequality comes from the computations in~\eqref{eq:th9-computation-S}. On the other hand, observe that \(\Delta^2_\mu \geq \Delta^2_{\mu_M}\) and \(\Delta^2_{\mu_M}\leq 4M^2\). Use these facts together with the bound \(1-e^{-a}\geq \frac{a}{1+a}\) for \(a\geq 0\) to get
\[
  g_{0,\theta}(x)  \geq  1 - \exp\left\{ -\frac{\Delta^2_{\mu_M}(x)}{4(\sigma_0^2+\sigma^2)} \right\} \geq C_2\Delta^2_{\mu_M}(x).
\]
Bringing together both pieces of information, we can take \(\tilde C = \min\{C_1, C_2\}\) and write
\[
  g_{0,\theta}(x) \geq \tilde C\left[\Delta^2_{\mu_M}(x) + (\sigma_0^2 - \sigma^2)^2\right],
\]
and integrating on both sides with respect to \(Q_X\) yields
\[
  d_H^2(f_0, f_\theta)\geq \tilde C\left[\|\Delta_{\mu_M}\|^2_{2,Q_X} + (\sigma_0^2 - \sigma^2)^2\right].
\]
To conclude, note that \(\sqrt{a^2+b^2}\geq (1/\sqrt{2})(a+b)\) for \(a,b\geq 0\), so we can take square roots on both sides and get the desired bound.

\subsubsection*{Examples of processes that satisfy Condition~\ref{cond:condition-schwartz-X}}

Condition~\ref{cond:condition-schwartz-X} implies that the process \(X\) is ``mean-square Lipschitz continuous'', i.e., there exists \(L>0\) such that \(\mathbb E[(X(t) - X(s))^2] \leq L (t-s)^2\) for all \(t, s \in [0,1]\). Looking at the derivations we did, we can substitute the almost sure Lipschitzness of the trajectories by this new condition and arrive at the same conclusions. Moreover, note that this guarantees, via Kolmogorov's continuity theorem, that the sample paths of \(X\) (after an eventual modification) are continuous functions. Apart from somewhat trivial examples, such as a deterministic process or a linear random process, we can show that processes with a sufficiently smooth kernel satisfy our requirements.

\begin{lemma}\label{lemma:mean-square-lipschitz}
  If the covariance function \(K\) of the process \(X\) is twice-continuously differentiable and the mean function \(\mathbb E[X(t)]\) vanishes for all \(t\in[0,1]\), then \(X\) is ``mean-square Lipschitz continuous'', that is, there exists \(L>0\) such that \(\mathbb E[(X(t) - X(s))^2] \leq L (t-s)^2\) for all \(t, s \in [0,1]\).
\end{lemma}
\begin{proof}
We know that \(\mathbb E[(X(t) - X(s))^2] = K(t,t)-2K(t,s) + K(s,s)\), and a Taylor expansion around the point \((s,s)\) shows that
\[
  K(t,t)-2K(t,s) + K(s,s) = \partial_{ts}K(s,s)(t-s)^2 + o((t-s)^2).
  \]
Since \(\partial_{ts}K\) is continuous on the compact set \([0,1]^2\), there exists a constant \(L>0\) such that
\[
  K(t,t)-2K(t,s) + K(s,s) \leq L(t-s)^2, \quad \text{for } t\to s.
\]
Lastly, this bound extends to all \(t,s\in[0,1]\) by continuity of \(K\).
\end{proof}

Note that the above condition is by no means a necessary one. For example, given a zero-mean second-order stochastic process \(\{Z(t): t \in [0,1]\}\), we consider the integrated process
\[
  X(t)=\int_0^t Z(u)\, du, \quad t\in [0,1].
\]
If \(K_Z(t,s)=\mathbb E[Z(t)Z(s)]\) is uniformly bounded above by a constant \(L>0\) on \([0,1]\), we have
\[
  \mathbb E\left[(X(t) - X(s))^2\right] = \int_s^t \int_s^t \mathbb E[Z(u)Z(v)]\, du\,dv \leq L(t-s)^2.
\]
For example, we might take \(Z(t)\) to be a standard Brownian motion, an Ornstein-Uhlenbeck process, or, in general, any second-order process with a continuous covariance function. In any case, by virtue of Lemma~\ref{lemma:mean-square-lipschitz} we can always pre-process the trajectories so that they are smooth enough for the conditions to hold, considering for instance the convolution with a smooth kernel.

\newpage
\section{Experimentation}\label{app:experiments}

\subsection{Overview of data sets and comparison algorithms}\label{app:data-sets}

To generate the simulated data sets for the comparison experiments in Section~\ref{sec:results}, we used four types of Gaussian process regressors commonly employed in the literature, each with a different covariance function:
\begin{description}
  \item [BM.] A Brownian motion, with kernel \(K_1(t,s)=\min\{t,s\}\).
  \item [fBM.] A fractional Brownian motion, with kernel \(K_2(t,s)=1/2(s^{2H} + t^{2H} - |t-s|^{2H})\) and Hurst parameter \(H=0.8\).
  \item [O-U.] An Ornstein-Uhlenbeck process, with kernel \(K_3(t,s)=e^{-|t-s|}\).
  \item [Gaussian.] A Gaussian process with a squared exponential kernel (also known as Gaussian kernel), namely \(K_4(t,s)=e^{-(t-s)^2/2\xi^2}\), where \(\xi=0.2\).
\end{description}

For the comparison algorithms themselves, we considered several frequentist methods which were selected among popular ones in FDA and machine learning in general. As specified earlier, variable selection and dimensionality reduction methods are part of a pipeline followed by a standard multiple regression technique. In the linear regression case, we chose the following algorithms:

\begin{description}
  \item [PLS1.] Partial least squares regression \citep[e.g.][]{abdi2010partial}.
  \item[Lasso.] Linear least squares with \(l^1\) regularization.
  \item [FPLS and FPLS1.] Functional PLS regression through basis expansion, implemented as in \citet{aguilera2010using}.
  \item [FLin.] Standard \(L^2\) functional linear regression model with fixed basis expansion and regularization.
  \item [APLS.] Functional partial least squares regression proposed by \citet{delaigle2012methodology}.
  \item [PLS.] Partial least squares for dimension reduction.
  \item [PCA.] Principal component analysis for dimension reduction.
  \item[Manual.] Dummy variable selection method with a pre-specified number of components (equispaced on \([0, 1]\)).
  \item [FPCA.] Functional principal component analysis.
  \item [Ridge.] Linear least squares with \(l^2\) regularization. This is used as the multiple regression that follows variable selection or dimensionality reduction methods.

\end{description}

In the logistic regression case, all the variable selection and dimension reduction techniques from above were also considered, with the addition of the following classification methods:

\begin{description}
  \item [QDA.] Quadratic discriminant analysis.
  \item [MDC.] Maximum depth classifier \citep[e.g.][]{ghosh2005maximum}.
  \item [Log.] Standard multiple logistic regression with \(l^2\) regularization. This is used as a one-stage method and also as the multiple regression that follows variable selection or dimensionality reduction methods.
  \item [LDA.] Linear discriminant analysis.
  \item [FNC.] Functional nearest centroid classifier with the \(L^2\)-distance.
  \item [FLog.] Functional RKHS-based logistic regression algorithm proposed in \citet{berrendero2023functional}.
  \item [FLDA.] Implementation of the functional version of linear discriminant analysis proposed in \citet{preda2007pls}.
  \item [FKNN.] Functional K-nearest neighbors classifier with the \(L^2\)-distance.
  \item [RKVS.] RKHS-based variable selection and classification method proposed in \citet{berrendero2018use}.
  \item [APLS+NC.] Functional PLS used as a dimension reduction method, as proposed in \citet{delaigle2012achieving} in combination with the nearest centroid (NC) algorithm.
\end{description}

The main hyperparameters of all these algorithms were selected by 10-fold cross-validation, and for those that have a number of components to select, we set 10 as the maximum value so that comparison with our own methods are fair. In particular, regularization parameters are searched among 20 values in the logarithmic space \([10^{-4}, 10^4]\), the number of basis elements for cubic spline bases is in \(\{4,5,\dots,10\}\), the number of basis elements for Fourier bases is one of \(\{1,3,5,7,9\}\), and the number of neighbors in the KNN classifier is in \(\{3,5,7,9,11,13\}\). Most algorithms have been taken from the libraries \textit{scikit-learn} \citep{pedregosa2011scikit} and \textit{scikit-fda} \citep{ramos2024scikit}, the first oriented to machine learning in general and the second to FDA in particular. However, some methods were not found in these packages and had to be implemented from scratch. This is the case of the FLDA, FPLS and APLS methods, which we coded following the corresponding articles.

\subsection{Simulations with non-Gaussian regressors}\label{app:non-gp}

We performed an additional set of experiments in linear and logistic regression in which the regressors are not Gaussian processes (GPs), to see if our methods would hold up in this case. These experiments where run in the same conditions as those reported in Section~\ref{sec:results}.

\subsubsection*{Functional linear regression}

We use a geometric Brownian motion (GBM) as the regressor variable, defined as \(X(t)=\exp\{{\operatorname{BM}(t)}\}\), where \(\operatorname{BM}(t)\) is a standard Brownian motion. In this case we consider two data sets, one with a RKHS response and one with an \(L^2\) response, both with the same parameters as in the corresponding data sets in Section~\ref{sec:results}. The comparison results can be seen in Figure~\ref{fig:reg_non_gp}: in this case our methods still get better results under the RKHS model, while the results under the \(L^2\)-model are essentially the same, which is a positive outcome.

\begin{figure}[ht!]
  \centering
  \includegraphics[width=.6\textwidth]{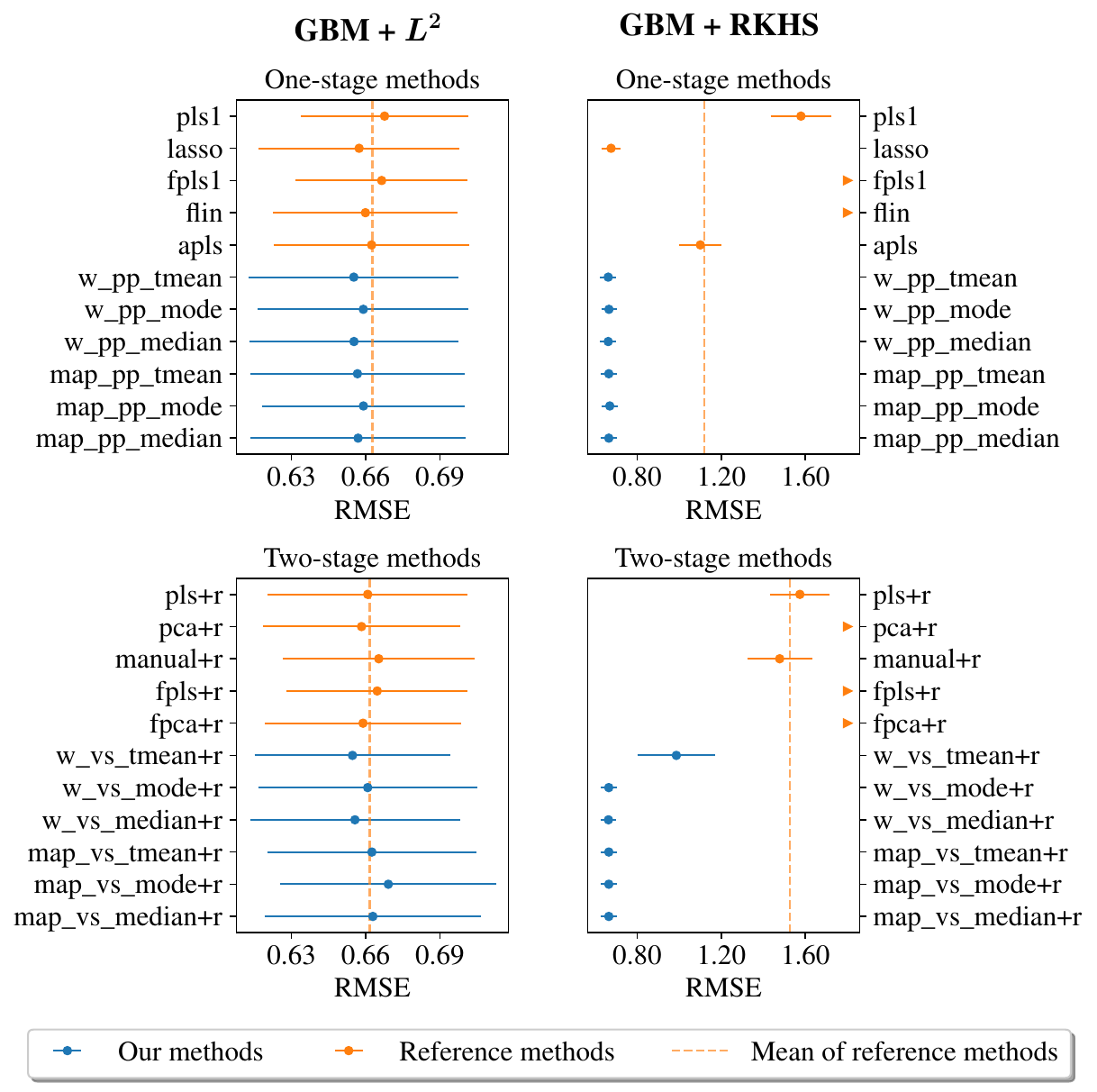}
  \caption{Mean and standard error of RMSE of predictors (lower is better) for 10 runs with GBM regressors. In the first column the response obeys a linear \(L^2\)-model, while in the second columns it follows a linear RKHS model.}\label{fig:reg_non_gp}
\end{figure}

\subsubsection*{Functional logistic regression}

We consider a ``mixture'' situation in which we combine regressors from two different GPs with equal probability and label them according to their origin. Firstly, we consider a homoscedastic case to distinguish between a standard Brownian motion and a Brownian motion with a mean function that is zero until \(t=0.5\), and then becomes \(m(t)=0.75t\). Secondly, we consider a heteroscedastic case to distinguish between a standard Brownian motion and a Brownian motion with variance 2, that is, with kernel \(K(t,s)=2\min\{t,s\}\).

Figure~\ref{fig:clf_non_gp} shows that our classifiers perform better than most comparison algorithms in both cases. The differences are most notable in the homoscedastic case, and in the heteroscedastic case the overall accuracy is low. Incidentally, this heteroscedastic case of two zero-mean Brownian motions has a special interest, since it can be shown that the Bayes error is zero in the limit of dense monitoring (i.e.\ with an arbitrarily fine measurement grid), a manifestation of the ``near-perfect'' classification phenomenon analyzed for example in \citet{torrecilla2020optimal}. Our results are in line with the empirical findings of this article, where the authors conclude that even though the asymptotic theoretical error is zero, most classification methods are suboptimal in practice (possibly due to the high collinearity of the data), with the notable exception of PCA+QDA.

\begin{figure}[ht!]
  \centering
  \includegraphics[width=.6\textwidth]{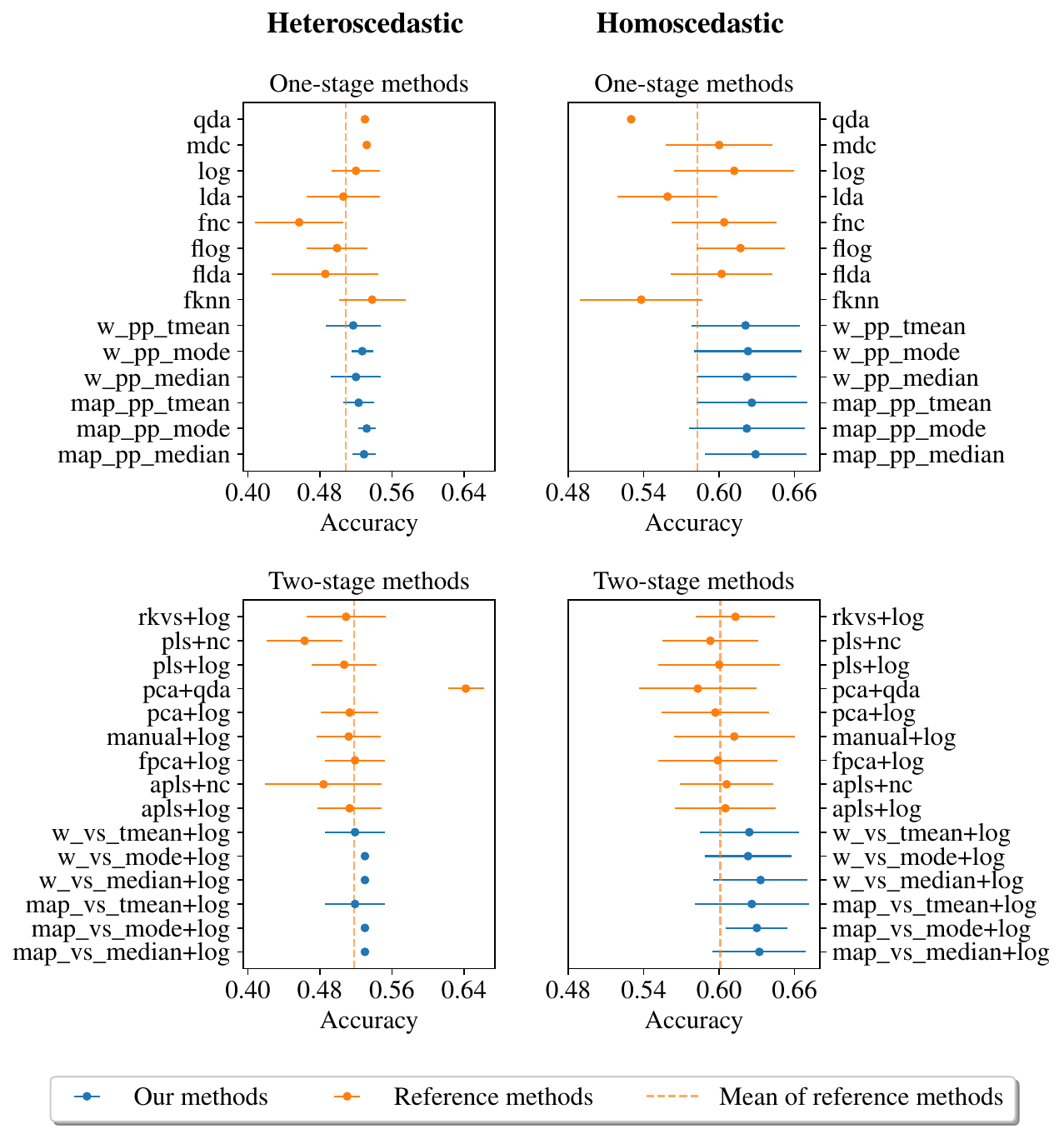}
  \caption{Mean and standard error of accuracy of classifiers (higher is better) for 10 runs with a mix of regressors coming from two different GPs and labeled according to their origin. In the first column we try to separate two Brownian motions with the same mean but different variance, while in the second column we discriminate between two Brownian motions with different mean functions but the same variance.}\label{fig:clf_non_gp}
\end{figure}

\subsection{Analysis and validation of a model}\label{app:validation}

We give an example through a series of visual representations of how one would analyze the outcome of our Bayesian methods. This is a preliminary step that comes before prediction; the idea is to validate the model and make sure that the resulting samples from the posterior are coherent and useful. For this illustration we consider a data set used in the experiments, for example the one with squared exponential GP regressors and an underlying linear RKHS response given by
\[
Y=5 - 5X(0.1) + 5X(0.6) + 10X(0.8) + \varepsilon, 
\] 
with \(\varepsilon \sim \mathcal N(0, 0.5)\) (see Figure~\ref{fig:dataset-linear}). We run the sampler for \(3000\) iterations and discard the first \(2000\), with a \(\text{Poisson}(3)\) prior truncated to \(\{1,\dots,5\}\) for \(p\).

\begin{figure}[ht!]
  \centering
  \includegraphics[width=.7\textwidth]{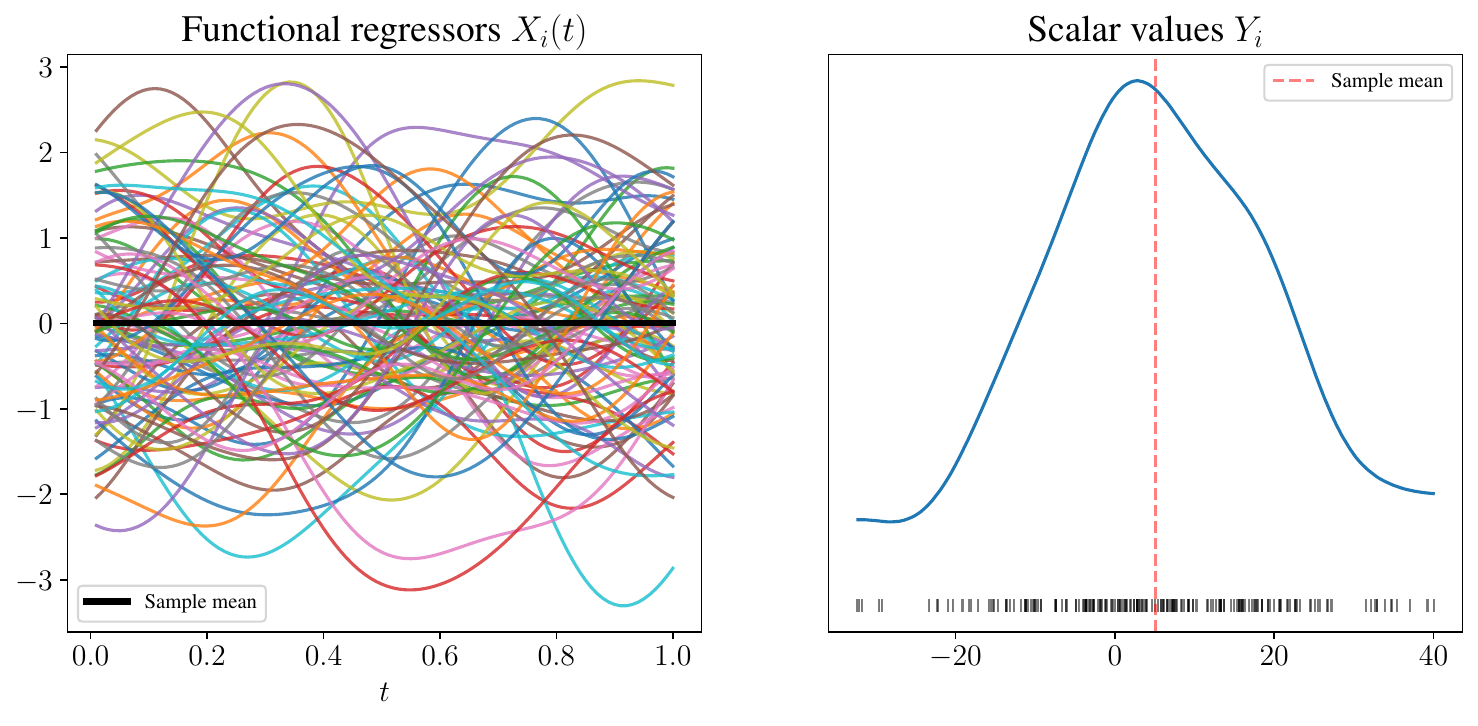}
  \caption{Data set with squared exponential GP regressors and linear RKHS response.}\label{fig:dataset-linear}
\end{figure}

The first thing we do is look at arbitrary samples in the last iteration of a few chains, to check that we get reasonable values. Then, we examine the acceptance rate of all the moves to see that they are not either very low or very high. Lastly, we compute the so-called Gelman-Rubin statistic \citep{gelman1992inference}, which is a quantitative measurement of the convergence of the chains (it should be near \(1\)).

Next we proceed with the visual checks. We can look at the flat posterior distribution of all parameters for all values of \(p\) and all the chains aggregated together (Figure~\ref{fig:flat-posterior-linear}), or visualize the posterior of the multidimensional parameters for each \(p\) in a sort of triangular configuration (Figure~\ref{fig:triangular-posterior-linear}). In addition, we can also look at the traces of individual parameters for all values of \(p\) (Figure~\ref{fig:trace-plot}).

\begin{figure}[ht!]
  \centering
  \includegraphics[width=.7\textwidth]{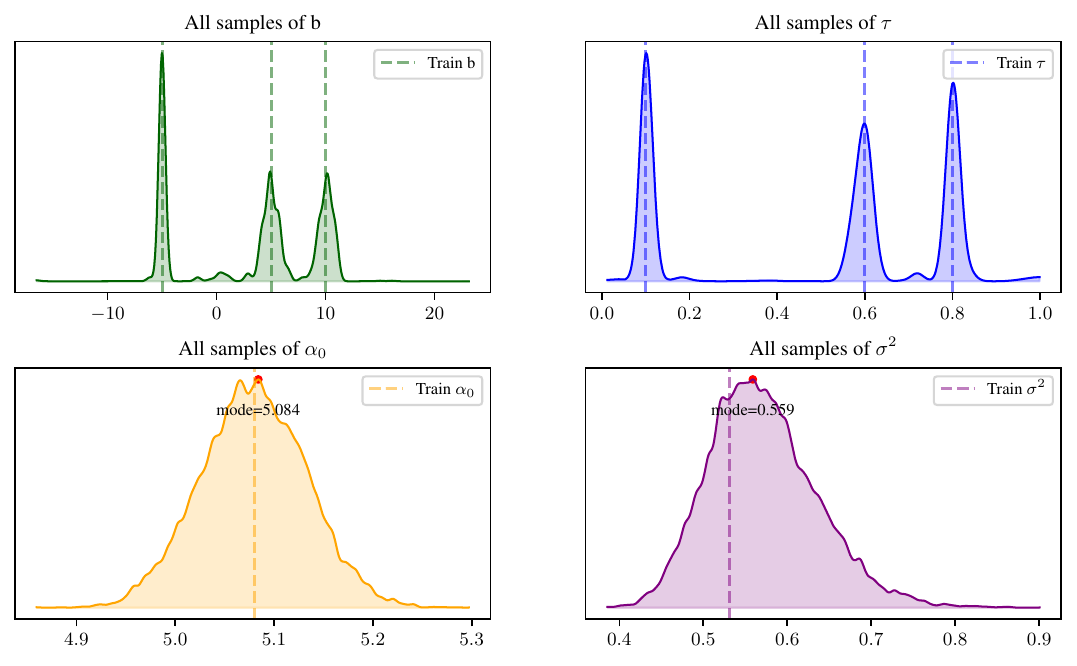}
  \caption{Aggregated posterior distribution of all the samples \(\theta^*_{p_m}\) for all \(m\). Note that the true values of the parameters are essentially recovered.}\label{fig:flat-posterior-linear}
\end{figure}

\begin{figure}[ht!]
  \centering
  \includegraphics[width=.7\textwidth]{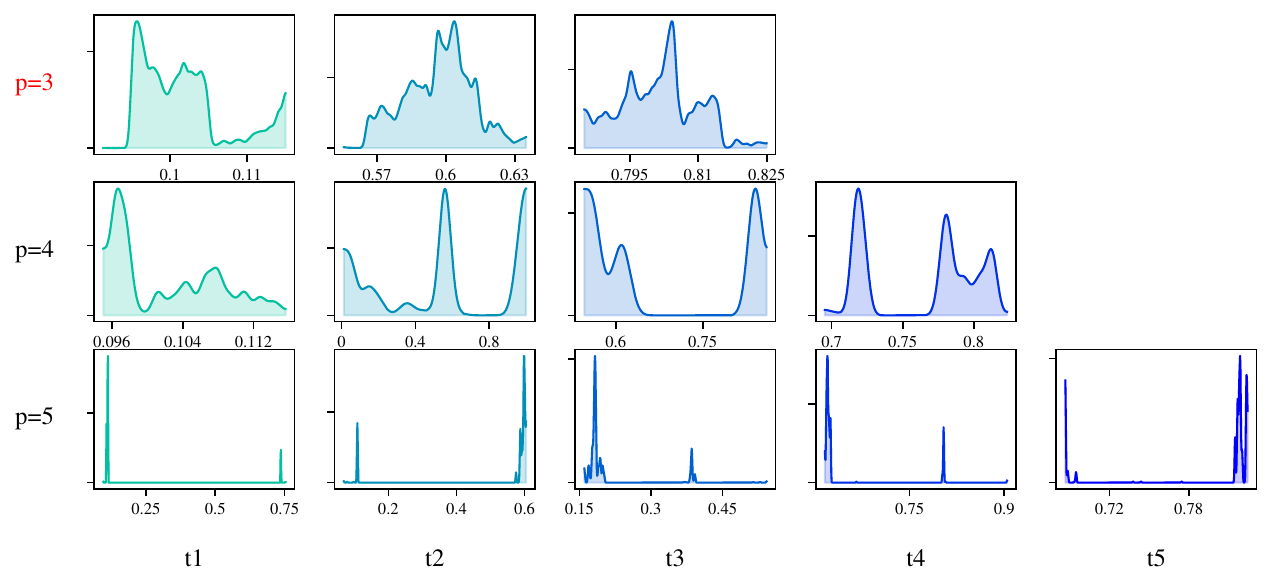}
  \caption{Posterior distribution of \(\tau_p^*\) for each \(p\). The most frequent value is \(p=3\), highlighted in red. In this run there were no samples with \(p=1\) or \(p=2\) after burn-in.}\label{fig:triangular-posterior-linear}
  \end{figure}
  
  \begin{figure}[ht!]
  \centering
  \begin{subfigure}[b]{.75\textwidth}
    \centering
    \includegraphics[width=\textwidth]{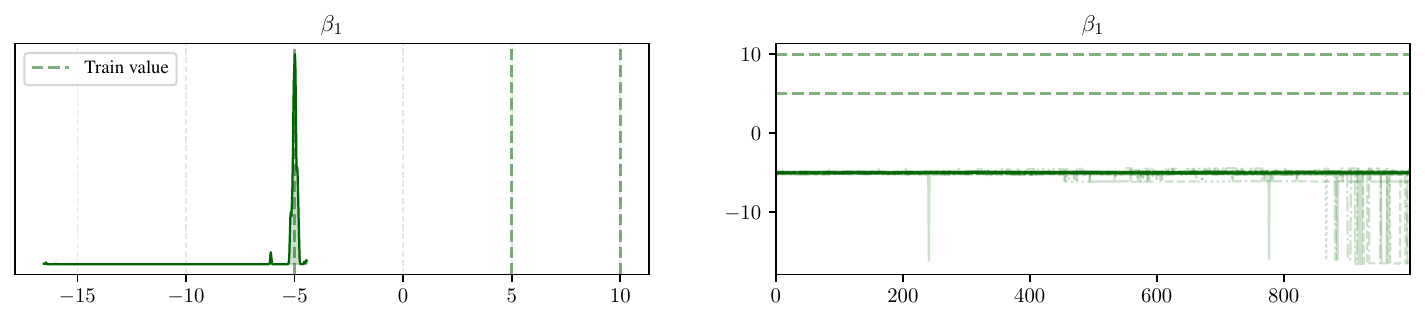}
  \end{subfigure}
  \begin{subfigure}[b]{.75\textwidth}
    \centering
    \includegraphics[width=\textwidth]{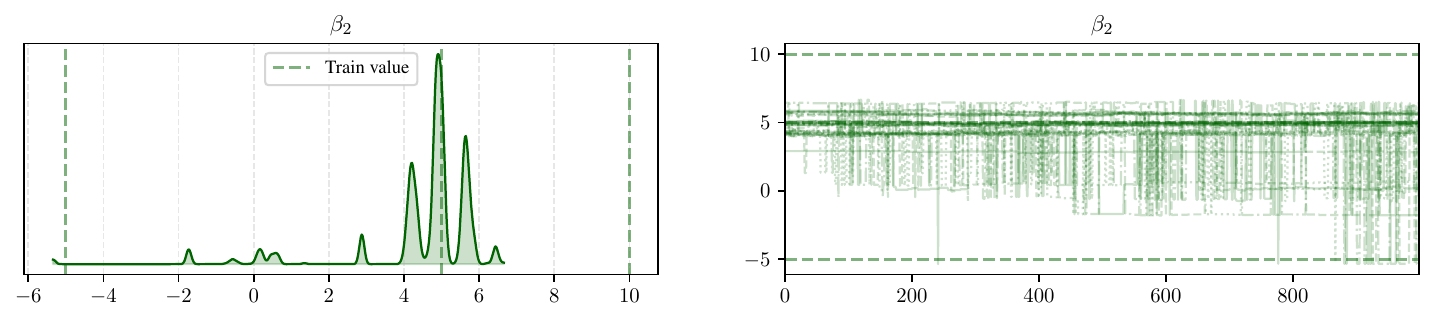}
  \end{subfigure}
  \begin{subfigure}[b]{.75\textwidth}
      \centering
      \includegraphics[width=\textwidth]{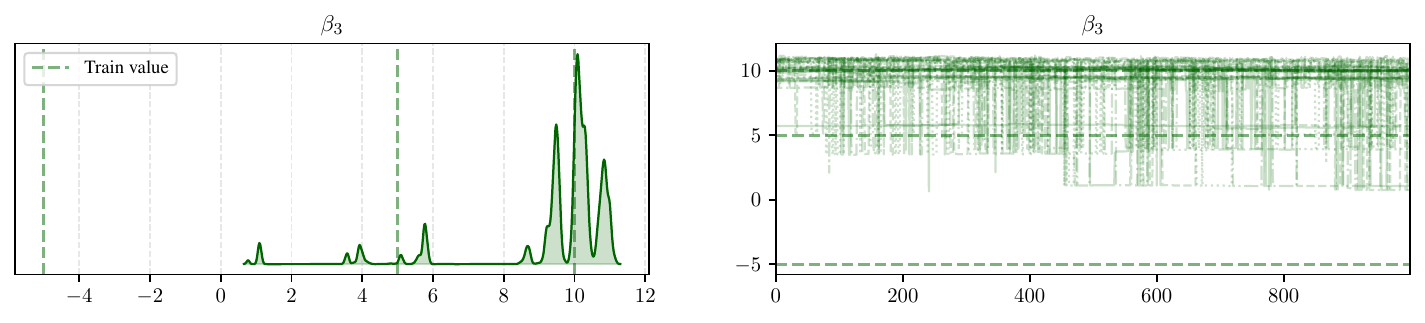}
  \end{subfigure}
  \begin{subfigure}[b]{.75\textwidth}
      \centering
      \includegraphics[width=\textwidth]{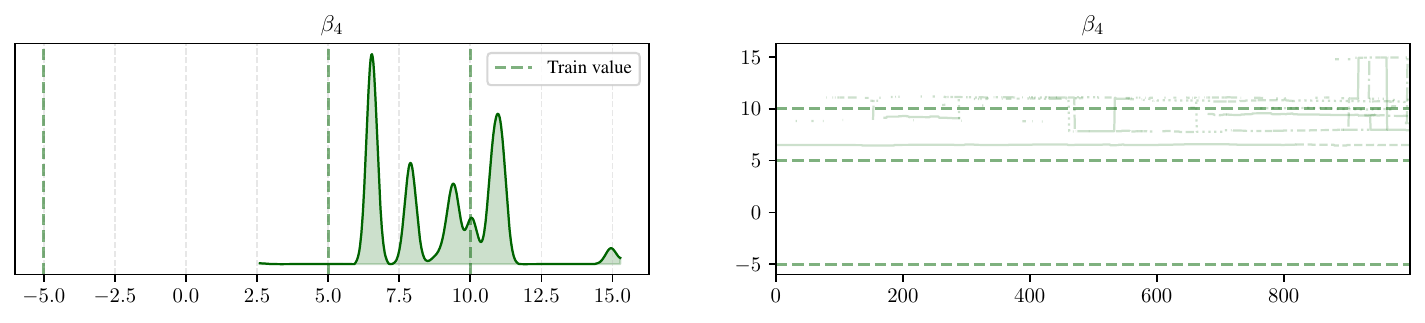}
  \end{subfigure}
  \begin{subfigure}[b]{.75\textwidth}
      \centering
      \includegraphics[width=\textwidth]{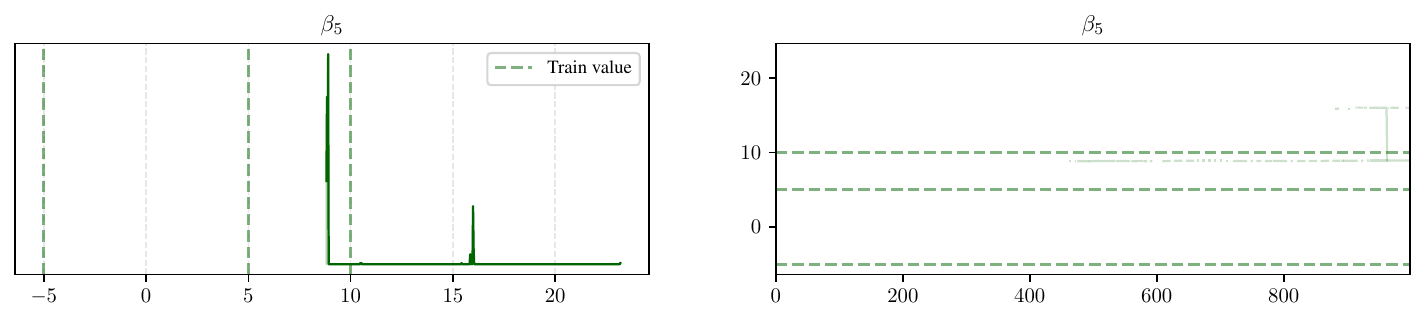}
  \end{subfigure}
  \caption{Posterior distribution (left) and trace (right) of all the \(\beta_j^*\), combined for all \(p\).}\label{fig:trace-plot}
\end{figure}

\FloatBarrier{}

Looking at the traces is useful to check that the chains are mixing well and that they are correctly exploring the parameter space. We can do the same thing with the values of \(p\) (Figure~\ref{fig:trace-p}). Moreover, we can visualize the tempered posterior distribution of \(p\), that is, the posterior distribution of \(p\) for each temperature (Figure~\ref{fig:tempered-posterior-p}).

\begin{figure}[ht!]
  \centering
  \includegraphics[width=.85\textwidth]{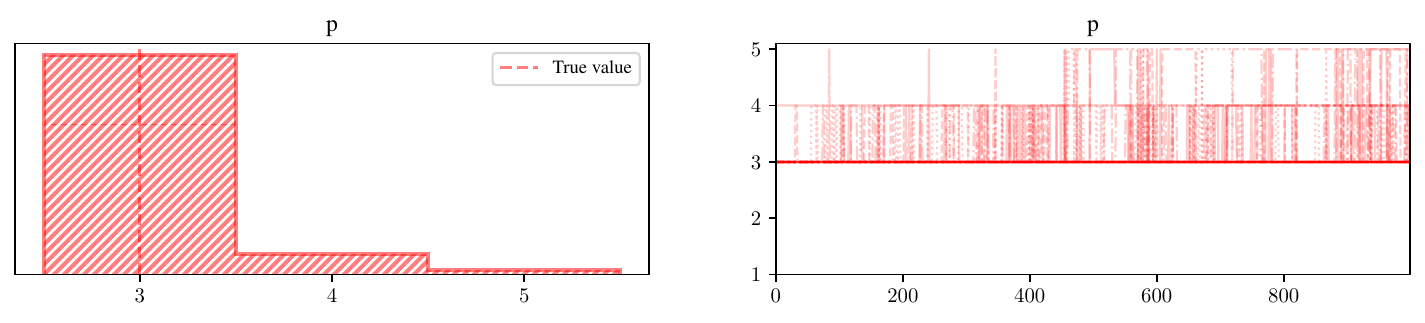}
  \caption{Posterior distribution (left) and trace (right) of \(p\). We see that the true number of components is recovered.}\label{fig:trace-p}
\end{figure}

\begin{figure}[ht!]
  \centering
  \includegraphics[width=.85\textwidth]{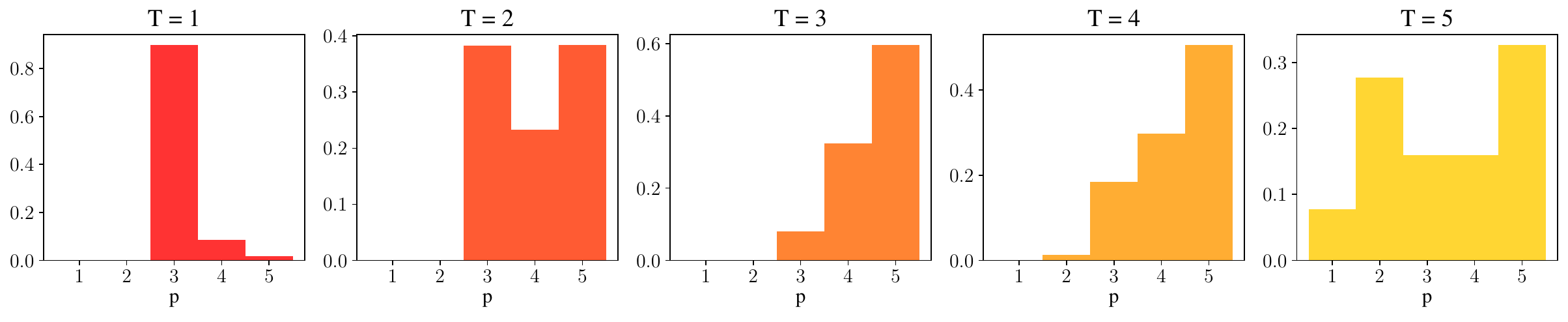}
  \caption{Tempered posterior distribution of \(p\). We are only interested in the cold chain (\(T=1\)), but by allowing different temperatures we increase the exploration of the parameter space, periodically transferring some of this information to the cold chain.}\label{fig:tempered-posterior-p}
\end{figure}

Lastly, we can perform a posterior predictive check (Figure~\ref{fig:pp-check}). This is arguably the most useful test for prediction purposes, since we represent the posterior predictive distribution that will be used for inference and prediction. We can do it on the training data or directly on previously unseen regressors. If the sampling has been successful, the posterior predictive should look like a tubular region around the observed data.

\enlargethispage{4\baselineskip}
\begin{figure}[ht!]
  \centering
  \includegraphics[width=.7\textwidth]{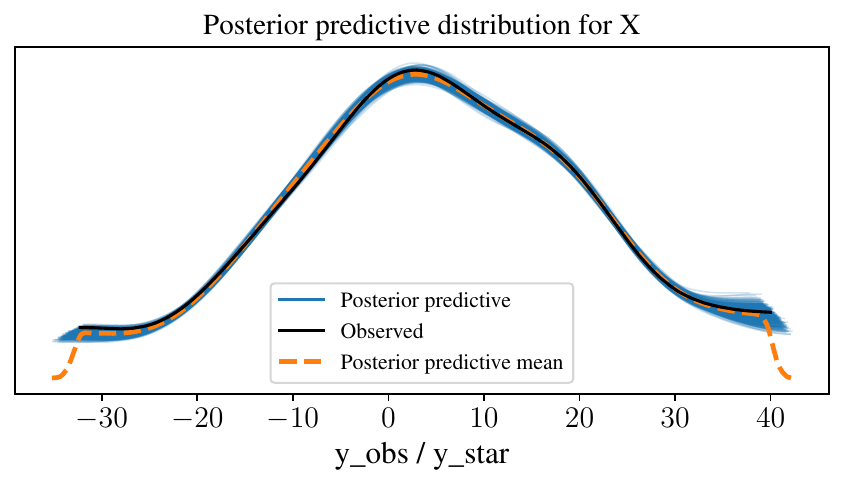}
  \caption{Posterior predictive distribution \(Y|X, \theta_{p_m}^*\) for each individual chain \(m\), along with the mean of all chains and the actual observed data \(Y\).}\label{fig:pp-check}
\end{figure}

\FloatBarrier{}

\subsection{Execution times}\label{app:execution-times}

In Figure~\ref{fig:reg-execution-times} and Figure~\ref{fig:clf-execution-times} we show the execution times of all the experiments in Appendix~\ref{app:non-gp} and Section~\ref{sec:results}. 

\begin{figure}[ht!]
  \centering
  \includegraphics[width=.75\textwidth]{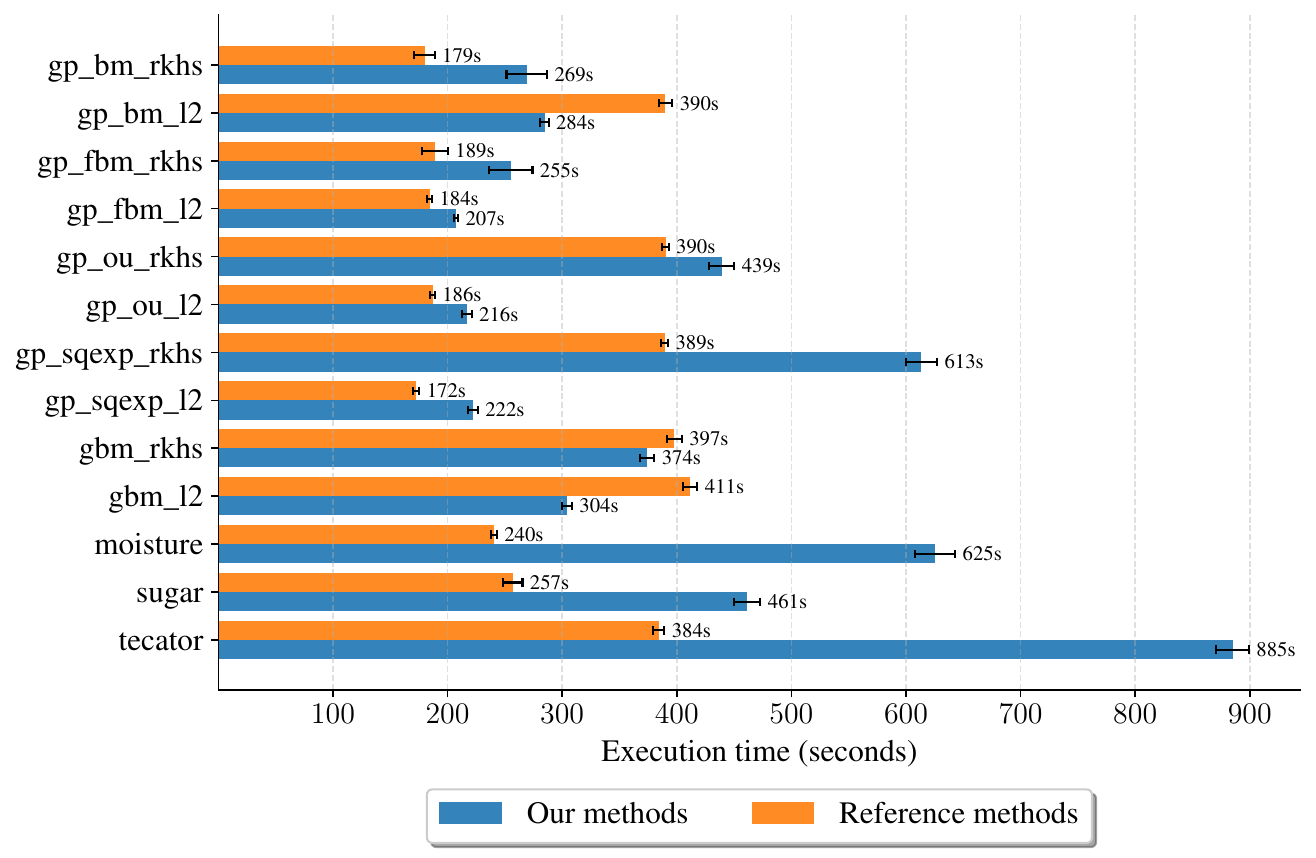}
  \caption{Mean and standard error of execution times for all splits in the experiments with the functional linear model.}\label{fig:reg-execution-times}
\end{figure}

\begin{figure}[ht!]
  \centering
  \includegraphics[width=.75\textwidth]{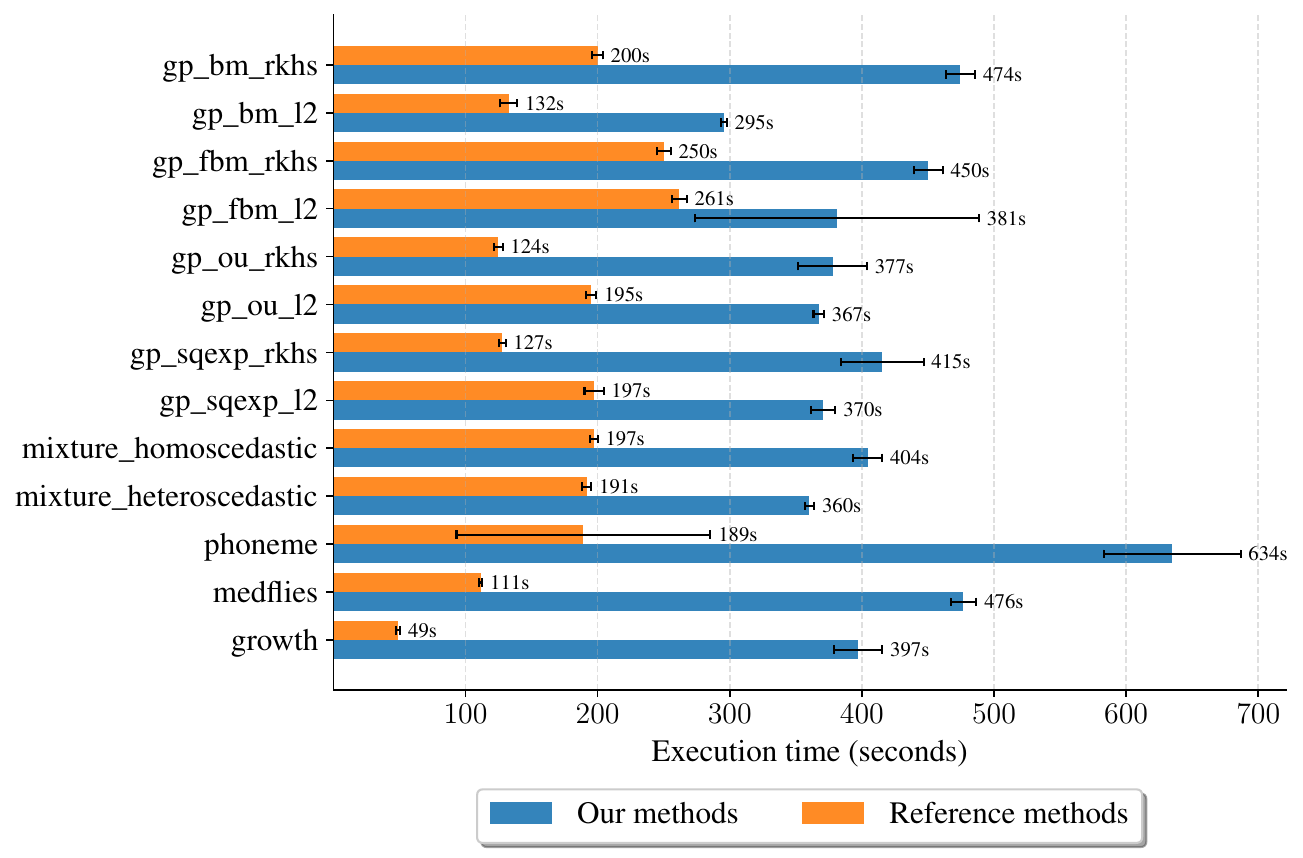}
  \caption{Mean and standard error of execution times for all splits in the experiments with the functional logistic model.}\label{fig:clf-execution-times}
\end{figure}

The execution times of the reference methods includes the duration of the cross-validation phase to select the best hyperparameters, a phase that our Bayesian methods lack by design. It is widely known that MCMC sampling tends to be slow, but as we can see the differences are manageable and our methods have reasonable run times for practical use, especially when taking into account the predictive improvement obtained in some cases. 

In addition, it is possible that we ran some MCMC chains for more steps than necessary, because we wanted to perform all experiments in the same generic configuration. In real-world scenarios one would pay closer attention to convergence metrics and try to stop the sampling earlier. Moreover, some splits may be artificially long because of the queue system in the cluster used to run the experiments in parallel.

\subsection{Tables of experimental results}\label{app:tables}

Here we present the tables corresponding to the empirical comparison studies in Appendix~\ref{app:non-gp} and Section~\ref{sec:results}, which show the numerical values that were depicted there graphically. In each case the best and second-best results are shown in \firstcolor{bold} and \secondcolor{italicized blue}, respectively.

\newpage
\subsubsection*{Functional linear regression}

\begin{table}[htbp!]
  \footnotesize
  \centering
  \rowcolors{2}{}{teal!8}
  \begin{tabular}{lcccc}
    \toprule
    \textbf{Prediction method} & \textbf{BM}                 & \textbf{fBM}                & \textbf{O-U}                & \textbf{Gaussian}           \\
    \midrule
    pls1 & 0.997 (0.051) & 0.730 (0.042) & 1.065 (0.055) & 0.673 (0.045) \\
    lasso & 0.680 (0.044) & 0.688 (0.031) & 0.684 (0.031) & 0.671 (0.042) \\
    fpls1 & 1.607 (0.078) & 0.751 (0.040) & 2.088 (0.099) & 0.670 (0.046) \\
    flin & 1.829 (0.073) & 0.869 (0.037) & 2.388 (0.078) & 0.957 (0.044) \\
    apls & 0.875 (0.053) & 0.704 (0.044) & 1.005 (0.050) & 0.672 (0.045) \\
    w\_pp\_tmean & \firstcolor{0.668 (0.043)} & \firstcolor{0.676 (0.032)} & 0.676 (0.042) & \firstcolor{0.662 (0.040)} \\
    w\_pp\_mode & 0.670 (0.041) & \firstcolor{0.676 (0.030)} & 0.675 (0.041) & \secondcolor{0.664 (0.040)} \\
    w\_pp\_median & \firstcolor{0.668 (0.043)} & \firstcolor{0.676 (0.032)} & 0.676 (0.042) & \firstcolor{0.662 (0.040)} \\
    map\_pp\_tmean & \secondcolor{0.669 (0.042)} & \secondcolor{0.677 (0.031)} & \secondcolor{0.671 (0.045)} & 0.666 (0.039) \\
    map\_pp\_mode & 0.672 (0.041) & 0.680 (0.029) & 0.672 (0.045) & 0.670 (0.042) \\
    map\_pp\_median & 0.670 (0.043) & \secondcolor{0.677 (0.031)} & \firstcolor{0.670 (0.045)} & 0.667 (0.040) \\
    \bottomrule
    \toprule
    pls+r & 0.996 (0.056) & 0.720 (0.036) & 1.065 (0.055) & 0.674 (0.044) \\
    pca+r & 1.521 (0.070) & 0.720 (0.034) & 2.249 (0.095) & 0.673 (0.045) \\
    manual+r & 1.342 (0.130) & 0.717 (0.040) & 1.719 (0.101) & 0.674 (0.042) \\
    fpls+r & 1.607 (0.078) & 0.752 (0.039) & 2.089 (0.098) & \secondcolor{0.669 (0.046)} \\
    fpca+r & 1.512 (0.071) & 0.721 (0.034) & 2.237 (0.096) & 0.673 (0.045) \\
    w\_vs\_tmean+r & \secondcolor{0.795 (0.069)} & 0.697 (0.029) & 0.987 (0.131) & 0.672 (0.037) \\
    w\_vs\_mode+r & \firstcolor{0.668 (0.040)} & \firstcolor{0.678 (0.034)} & 0.667 (0.041) & 0.680 (0.050) \\
    w\_vs\_median+r & \firstcolor{0.668 (0.039)} & 0.681 (0.031) & \secondcolor{0.666 (0.041)} & \firstcolor{0.666 (0.041)} \\
    map\_vs\_tmean+r & \firstcolor{0.668 (0.040)} & 0.747 (0.048) & 0.971 (0.398) & 0.737 (0.093) \\
    map\_vs\_mode+r & \firstcolor{0.668 (0.040)} & \secondcolor{0.679 (0.034)} & 0.668 (0.042) & 0.707 (0.037) \\
    map\_vs\_median+r & \firstcolor{0.668 (0.040)} & 0.695 (0.029) & \firstcolor{0.664 (0.041)} & 0.700 (0.063) \\
    \bottomrule
  \end{tabular}
  \caption{Mean RMSE of predictors (lower is better) for 10 runs with GP regressors, one on each column, that obey an underlying linear RKHS model. The corresponding standard errors are shown between brackets.}
\end{table}
\FloatBarrier{}
\newpage

\begin{table}[htbp!]
  \vspace{1em}
  \footnotesize
  \centering
  \rowcolors{2}{}{teal!8}
  \begin{tabular}{lcccc}
    \toprule
    \textbf{Prediction method} & \textbf{BM}                 & \textbf{fBM}                & \textbf{O-U}                & \textbf{Gaussian}           \\
    \midrule
    pls1 & 0.678 (0.039) & 0.674 (0.030) & 0.682 (0.043) & 0.669 (0.047) \\
    lasso & 0.667 (0.037) & 0.664 (0.035) & 0.659 (0.039) & 0.664 (0.041) \\
    fpls1 & 0.671 (0.041) & 0.671 (0.044) & 0.677 (0.037) & 0.671 (0.048) \\
    flin & 0.666 (0.037) & 0.664 (0.038) & 0.662 (0.038) & \firstcolor{0.661 (0.040)} \\
    apls & 0.673 (0.045) & 0.684 (0.043) & 0.681 (0.043) & 0.674 (0.041) \\
    w\_pp\_tmean & \firstcolor{0.652 (0.037)} & \firstcolor{0.659 (0.039)} & \firstcolor{0.651 (0.037)} & \firstcolor{0.661 (0.039)} \\
    w\_pp\_mode & 0.655 (0.037) & \secondcolor{0.661 (0.038)} & \secondcolor{0.652 (0.037)} & \firstcolor{0.661 (0.039)} \\
    w\_pp\_median & \firstcolor{0.652 (0.037)} & \firstcolor{0.659 (0.039)} & \firstcolor{0.651 (0.037)} & \firstcolor{0.661 (0.039)} \\
    map\_pp\_tmean & \secondcolor{0.654 (0.037)} & \secondcolor{0.661 (0.038)} & 0.655 (0.038) & \firstcolor{0.661 (0.039)} \\
    map\_pp\_mode & 0.655 (0.035) & 0.664 (0.039) & 0.656 (0.037) & \secondcolor{0.662 (0.037)} \\
    map\_pp\_median & \secondcolor{0.654 (0.037)} & \secondcolor{0.661 (0.038)} & 0.655 (0.038) & \firstcolor{0.661 (0.039)} \\
    \bottomrule
    \toprule
    pls+r & 0.667 (0.037) & 0.669 (0.033) & 0.669 (0.037) & \firstcolor{0.665 (0.044)} \\
    pca+r & 0.665 (0.041) & 0.664 (0.041) & 0.665 (0.041) & \secondcolor{0.666 (0.044)} \\
    manual+r & 0.665 (0.035) & \firstcolor{0.660 (0.037)} & 0.664 (0.040) & \secondcolor{0.666 (0.046)} \\
    fpls+r & 0.666 (0.040) & 0.664 (0.042) & 0.672 (0.036) & 0.668 (0.047) \\
    fpca+r & 0.665 (0.042) & 0.665 (0.042) & 0.663 (0.040) & \secondcolor{0.666 (0.045)} \\
    w\_vs\_tmean+r & \firstcolor{0.651 (0.038)} & \firstcolor{0.660 (0.039)} & \firstcolor{0.649 (0.036)} & 0.675 (0.039) \\
    w\_vs\_mode+r & 0.660 (0.037) & 0.662 (0.039) & 0.663 (0.032) & 0.673 (0.046) \\
    w\_vs\_median+r & \secondcolor{0.652 (0.036)} & \firstcolor{0.660 (0.038)} & \secondcolor{0.653 (0.039)} & 0.672 (0.039) \\
    map\_vs\_tmean+r & 0.654 (0.037) & \secondcolor{0.661 (0.037)} & 0.659 (0.039) & 0.674 (0.041) \\
    map\_vs\_mode+r & 0.661 (0.036) & 0.667 (0.040) & 0.670 (0.036) & 0.704 (0.058) \\
    map\_vs\_median+r & 0.655 (0.035) & 0.663 (0.035) & 0.662 (0.043) & 0.678 (0.042) \\
    \bottomrule
  \end{tabular}
  \caption{Mean RMSE of predictors (lower is better) for 10 runs with GP regressors, one on each column, that obey an underlying linear \(L^2\)-model. The corresponding standard errors are shown between brackets.}
\end{table}
\newpage
\FloatBarrier{}

\begin{table}[htbp!]
  \vspace{1em}
  \footnotesize
  \centering
  \rowcolors{2}{}{teal!8}
  \begin{tabular}{lcc}
    \toprule
    \textbf{Prediction method} & \textbf{GBM + \(\bf{L^2}\)} & \textbf{GBM + RKHS}         \\
    \midrule
    pls1 & 0.668 (0.034) & 1.579 (0.145) \\
    lasso & \secondcolor{0.657 (0.041)} & 0.676 (0.046) \\
    fpls1 & 0.666 (0.035) & 2.654 (0.182) \\
    flin & 0.660 (0.037) & 3.434 (0.384) \\
    apls & 0.662 (0.040) & 1.100 (0.100) \\
    w\_pp\_tmean & \firstcolor{0.655 (0.042)} & \firstcolor{0.662 (0.039)} \\
    w\_pp\_mode & 0.659 (0.043) & 0.666 (0.036) \\
    w\_pp\_median & \firstcolor{0.655 (0.042)} & \firstcolor{0.662 (0.039)} \\
    map\_pp\_tmean & \secondcolor{0.657 (0.043)} & \secondcolor{0.665 (0.038)} \\
    map\_pp\_mode & 0.659 (0.041) & 0.670 (0.037) \\
    map\_pp\_median & \secondcolor{0.657 (0.043)} & \secondcolor{0.665 (0.038)} \\
    \bottomrule
    \toprule
    pls+r & 0.661 (0.040) & 1.574 (0.141) \\
    pca+r & 0.658 (0.040) & 2.315 (0.195) \\
    manual+r & 0.665 (0.039) & 1.478 (0.154) \\
    fpls+r & 0.665 (0.037) & 2.669 (0.189) \\
    fpca+r & 0.659 (0.040) & 2.311 (0.194) \\
    w\_vs\_tmean+r & \firstcolor{0.655 (0.040)} & 0.986 (0.186) \\
    w\_vs\_mode+r & 0.661 (0.044) & \secondcolor{0.665 (0.038)} \\
    w\_vs\_median+r & \secondcolor{0.656 (0.042)} & \firstcolor{0.663 (0.037)} \\
    map\_vs\_tmean+r & 0.662 (0.042) & \secondcolor{0.665 (0.038)} \\
    map\_vs\_mode+r & 0.669 (0.044) & \secondcolor{0.665 (0.038)} \\
    map\_vs\_median+r & 0.663 (0.044) & \secondcolor{0.665 (0.038)} \\
    \bottomrule
  \end{tabular}
  \caption{Mean RMSE of predictors (lower is better) for 10 runs with GBM regressors. In the first column the response obeys a linear \(L^2\)-model, while in the second column it follows a linear RKHS model. The corresponding standard errors are shown between brackets.}
\end{table}
\newpage
\FloatBarrier{}

\begin{table}[htbp!]
  \vspace{1em}
  \footnotesize
  \centering
  \rowcolors{2}{}{teal!8}
  \begin{tabular}{lccc}
    \toprule
    \textbf{Prediction method} & \textbf{Moisture}           & \textbf{Sugar}              & \textbf{Tecator}            \\
    \midrule
    pls1 & 0.232 (0.024) & 2.037 (0.219) & \secondcolor{2.606 (0.283)} \\
    lasso & 0.242 (0.026) & 1.985 (0.226) & 2.842 (0.352) \\
    fpls1 & 0.248 (0.022) & \secondcolor{1.972 (0.201)} & \firstcolor{2.605 (0.262)} \\
    flin & 1.235 (0.138) & \firstcolor{1.966 (0.198)} & 7.486 (0.648) \\
    apls & 0.237 (0.028) & 2.020 (0.226) & 2.641 (0.165) \\
    w\_pp\_tmean & \secondcolor{0.223 (0.018)} & 1.988 (0.216) & 2.721 (0.260) \\
    w\_pp\_mode & \firstcolor{0.222 (0.020)} & 1.996 (0.219) & 2.718 (0.266) \\
    w\_pp\_median & \secondcolor{0.223 (0.018)} & 1.989 (0.217) & 2.721 (0.262) \\
    map\_pp\_tmean & 0.228 (0.021) & 1.997 (0.212) & 2.724 (0.267) \\
    map\_pp\_mode & 0.236 (0.023) & 2.010 (0.210) & 2.741 (0.271) \\
    map\_pp\_median & 0.228 (0.021) & 1.996 (0.212) & 2.720 (0.267) \\
    \bottomrule
    \toprule
    pls+r & \firstcolor{0.225 (0.021)} & \secondcolor{1.998 (0.208)} & \firstcolor{2.536 (0.236)} \\
    pca+r & \secondcolor{0.233 (0.024)} & 2.034 (0.219) & 2.712 (0.183) \\
    manual+r & 0.268 (0.021) & 2.041 (0.214) & \secondcolor{2.586 (0.270)} \\
    fpls+r & 0.241 (0.018) & \firstcolor{1.962 (0.202)} & 2.595 (0.249) \\
    fpca+r & 0.307 (0.055) & 2.054 (0.226) & 2.657 (0.189) \\
    w\_vs\_tmean+r & 0.308 (0.075) & 2.003 (0.221) & 2.822 (0.323) \\
    w\_vs\_mode+r & 0.236 (0.018) & 2.050 (0.221) & 2.713 (0.255) \\
    w\_vs\_median+r & 0.236 (0.025) & 2.000 (0.217) & 2.801 (0.261) \\
    map\_vs\_tmean+r & 0.455 (0.284) & 2.059 (0.233) & 2.900 (0.403) \\
    map\_vs\_mode+r & 0.261 (0.029) & 2.188 (0.321) & 2.833 (0.291) \\
    map\_vs\_median+r & 0.266 (0.030) & 2.082 (0.219) & 2.826 (0.269) \\
    \bottomrule
  \end{tabular}
  \caption{Mean RMSE of predictors (lower is better) for 10 runs with real data sets, one on each column. The corresponding standard errors are shown between brackets.}
\end{table}
\newpage
\FloatBarrier{}

\subsubsection*{Functional logistic regression}

\begin{table}[htbp!]
  \footnotesize
  \centering
  \rowcolors{2}{}{teal!8}
  \begin{tabular}{lcccc}
    \toprule
    \textbf{Classification method} & \textbf{BM}                 & \textbf{fBM}                & \textbf{O-U}                & \textbf{Gaussian}           \\
    \midrule
    qda & 0.510 (0.000) & 0.510 (0.000) & 0.510 (0.000) & 0.500 (0.000) \\
    mdc & 0.804 (0.034) & 0.822 (0.022) & 0.735 (0.037) & 0.839 (0.045) \\
    log & 0.849 (0.031) & \firstcolor{0.848 (0.015)} & 0.824 (0.022) & 0.868 (0.036) \\
    lda & 0.694 (0.032) & 0.621 (0.066) & 0.624 (0.042) & 0.823 (0.028) \\
    fnc & 0.814 (0.034) & \firstcolor{0.848 (0.014)} & 0.736 (0.035) & 0.864 (0.046) \\
    flog & 0.845 (0.036) & 0.837 (0.024) & 0.809 (0.028) & 0.871 (0.033) \\
    flda & 0.846 (0.031) & 0.830 (0.029) & 0.813 (0.029) & 0.854 (0.039) \\
    fknn & 0.851 (0.033) & 0.834 (0.027) & 0.799 (0.024) & 0.847 (0.041) \\
    w\_pp\_tmean & \firstcolor{0.856 (0.030)} & 0.846 (0.011) & 0.828 (0.022) & 0.873 (0.035) \\
    w\_pp\_mode & 0.853 (0.031) & \secondcolor{0.847 (0.012)} & 0.825 (0.025) & \firstcolor{0.878 (0.036)} \\
    w\_pp\_median & \firstcolor{0.856 (0.029)} & \secondcolor{0.847 (0.011)} & 0.827 (0.022) & 0.873 (0.035) \\
    map\_pp\_tmean & 0.854 (0.034) & 0.845 (0.013) & \firstcolor{0.830 (0.025)} & 0.874 (0.035) \\
    map\_pp\_mode & 0.852 (0.032) & 0.846 (0.014) & \secondcolor{0.829 (0.025)} & \secondcolor{0.877 (0.036)} \\
    map\_pp\_median & \secondcolor{0.855 (0.030)} & 0.844 (0.013) & \firstcolor{0.830 (0.025)} & 0.876 (0.036) \\
    \bottomrule
    \toprule
    rkvs+log & \firstcolor{0.848 (0.024)} & 0.838 (0.026) & 0.790 (0.032) & 0.872 (0.041) \\
    pls+nc & 0.816 (0.038) & 0.828 (0.022) & 0.793 (0.029) & 0.867 (0.039) \\
    pls+log & \secondcolor{0.847 (0.034)} & 0.844 (0.021) & 0.817 (0.022) & 0.864 (0.037) \\
    pca+qda & 0.839 (0.034) & 0.840 (0.018) & 0.818 (0.026) & 0.854 (0.033) \\
    pca+log & 0.842 (0.032) & \secondcolor{0.847 (0.016)} & 0.824 (0.019) & 0.868 (0.033) \\
    manual+log & 0.846 (0.032) & \firstcolor{0.850 (0.012)} & 0.821 (0.018) & 0.869 (0.032) \\
    fpca+log & \secondcolor{0.847 (0.030)} & \secondcolor{0.847 (0.014)} & \secondcolor{0.830 (0.024)} & 0.866 (0.036) \\
    apls+nc & 0.819 (0.036) & \secondcolor{0.847 (0.014)} & 0.816 (0.027) & 0.854 (0.036) \\
    apls+log & 0.829 (0.041) & 0.844 (0.012) & 0.816 (0.027) & 0.857 (0.039) \\
    w\_vs\_tmean+log & 0.831 (0.038) & 0.840 (0.017) & 0.823 (0.033) & 0.873 (0.038) \\
    w\_vs\_mode+log & 0.846 (0.021) & 0.828 (0.019) & 0.806 (0.020) & 0.875 (0.040) \\
    w\_vs\_median+log & 0.838 (0.029) & 0.844 (0.017) & \firstcolor{0.834 (0.030)} & 0.875 (0.039) \\
    map\_vs\_tmean+log & 0.816 (0.028) & 0.838 (0.023) & 0.801 (0.027) & \secondcolor{0.876 (0.040)} \\
    map\_vs\_mode+log & 0.839 (0.022) & 0.831 (0.022) & 0.807 (0.014) & \firstcolor{0.877 (0.043)} \\
    map\_vs\_median+log & 0.829 (0.024) & 0.839 (0.019) & 0.816 (0.034) & 0.871 (0.048) \\
    \bottomrule
  \end{tabular}
  \caption{Mean accuracy of classifiers (higher is better) for 10 runs with GP regressors, one on each column, that obey an underlying logistic RKHS model. The corresponding standard errors are shown between brackets.}
\end{table}
\newpage
\FloatBarrier{}

\begin{table}[htbp!]
  \vspace{1em}
  \footnotesize
  \centering
  \rowcolors{2}{}{teal!8}
  \begin{tabular}{lcccc}
    \toprule
    \textbf{Classification method} & \textbf{BM}                 & \textbf{fBM}                & \textbf{O-U}                & \textbf{Gaussian}           \\
    \midrule
    qda & \firstcolor{0.610 (0.000)} & 0.610 (0.000) & \firstcolor{0.620 (0.000)} & \secondcolor{0.610 (0.000)} \\
    mdc & 0.602 (0.033) & \firstcolor{0.619 (0.048)} & \secondcolor{0.615 (0.029)} & 0.603 (0.042) \\
    log & 0.594 (0.017) & 0.577 (0.039) & 0.591 (0.024) & 0.609 (0.037) \\
    lda & 0.507 (0.030) & 0.518 (0.029) & 0.541 (0.039) & 0.591 (0.033) \\
    fnc & \secondcolor{0.607 (0.038)} & \secondcolor{0.614 (0.045)} & 0.609 (0.029) & \firstcolor{0.625 (0.040)} \\
    flog & 0.580 (0.020) & 0.602 (0.037) & 0.600 (0.031) & 0.609 (0.028) \\
    flda & 0.601 (0.027) & 0.609 (0.049) & 0.593 (0.032) & 0.595 (0.048) \\
    fknn & 0.587 (0.056) & 0.576 (0.033) & 0.564 (0.042) & 0.578 (0.041) \\
    w\_pp\_tmean & 0.597 (0.027) & 0.600 (0.026) & 0.592 (0.023) & 0.608 (0.035) \\
    w\_pp\_mode & 0.599 (0.028) & 0.606 (0.027) & 0.595 (0.030) & 0.605 (0.034) \\
    w\_pp\_median & 0.597 (0.023) & 0.599 (0.030) & 0.591 (0.020) & 0.607 (0.037) \\
    map\_pp\_tmean & 0.595 (0.022) & 0.599 (0.027) & 0.594 (0.020) & 0.602 (0.037) \\
    map\_pp\_mode & 0.605 (0.027) & 0.604 (0.030) & 0.602 (0.030) & 0.606 (0.041) \\
    map\_pp\_median & 0.593 (0.026) & 0.599 (0.028) & 0.600 (0.023) & 0.602 (0.034) \\
    \bottomrule
    \toprule
    rkvs+log & 0.569 (0.039) & 0.586 (0.026) & 0.593 (0.035) & 0.611 (0.034) \\
    pls+nc & \firstcolor{0.610 (0.033)} & \firstcolor{0.623 (0.042)} & \firstcolor{0.607 (0.035)} & \firstcolor{0.629 (0.036)} \\
    pls+log & 0.590 (0.029) & 0.589 (0.036) & 0.593 (0.020) & \secondcolor{0.621 (0.043)} \\
    pca+qda & 0.577 (0.036) & \secondcolor{0.615 (0.032)} & 0.599 (0.043) & 0.618 (0.045) \\
    pca+log & 0.590 (0.027) & 0.593 (0.027) & 0.598 (0.037) & 0.616 (0.036) \\
    manual+log & 0.580 (0.026) & 0.585 (0.030) & 0.593 (0.019) & \secondcolor{0.621 (0.044)} \\
    fpca+log & \secondcolor{0.599 (0.021)} & 0.592 (0.034) & 0.600 (0.039) & 0.614 (0.042) \\
    apls+nc & 0.593 (0.024) & 0.569 (0.047) & 0.591 (0.039) & 0.615 (0.023) \\
    apls+log & 0.571 (0.030) & 0.585 (0.033) & 0.594 (0.025) & 0.614 (0.038) \\
    w\_vs\_tmean+log & 0.588 (0.034) & 0.594 (0.022) & 0.596 (0.017) & 0.605 (0.029) \\
    w\_vs\_mode+log & 0.597 (0.024) & 0.597 (0.028) & \secondcolor{0.601 (0.026)} & 0.617 (0.026) \\
    w\_vs\_median+log & 0.591 (0.030) & 0.589 (0.023) & 0.599 (0.035) & 0.600 (0.031) \\
    map\_vs\_tmean+log & 0.590 (0.031) & 0.594 (0.026) & 0.589 (0.030) & 0.609 (0.031) \\
    map\_vs\_mode+log & 0.597 (0.024) & 0.593 (0.022) & 0.599 (0.027) & 0.617 (0.021) \\
    map\_vs\_median+log & 0.587 (0.036) & 0.592 (0.021) & 0.598 (0.036) & 0.605 (0.034) \\
    \bottomrule
  \end{tabular}
  \caption{Mean accuracy of classifiers (higher is better) for 10 runs with GP regressors, one on each column, that obey an underlying logistic \(L^2\)-model. The corresponding standard errors are shown between brackets.}
\end{table}
\newpage
\FloatBarrier{}

\begin{table}[htbp!]
  \vspace{1em}
  \footnotesize
  \centering
  \rowcolors{2}{}{teal!8}
  \begin{tabular}{lcccc}
    \toprule
    \textbf{Classification method} & \textbf{Heteroscedastic}    & \textbf{Homoscedastic}      \\
    \midrule
    qda & 0.530 (0.000) & 0.530 (0.000) \\
    mdc & \secondcolor{0.532 (0.004)} & 0.600 (0.042) \\
    log & 0.520 (0.027) & 0.612 (0.048) \\
    lda & 0.506 (0.041) & 0.559 (0.040) \\
    fnc & 0.457 (0.049) & 0.604 (0.042) \\
    flog & 0.499 (0.034) & 0.617 (0.035) \\
    flda & 0.486 (0.059) & 0.602 (0.040) \\
    fknn & \firstcolor{0.538 (0.037)} & 0.538 (0.049) \\
    w\_pp\_tmean & 0.517 (0.030) & 0.621 (0.043) \\
    w\_pp\_mode & 0.527 (0.012) & 0.623 (0.043) \\
    w\_pp\_median & 0.520 (0.028) & 0.622 (0.040) \\
    map\_pp\_tmean & 0.523 (0.017) & \secondcolor{0.626 (0.044)} \\
    map\_pp\_mode & \secondcolor{0.532 (0.010)} & 0.622 (0.046) \\
    map\_pp\_median & 0.529 (0.013) & \firstcolor{0.629 (0.040)} \\
    \bottomrule
    \toprule
    rkvs+log & 0.509 (0.044) & 0.613 (0.031) \\
    pls+nc & 0.463 (0.042) & 0.593 (0.038) \\
    pls+log & 0.507 (0.036) & 0.600 (0.048) \\
    pca+qda & \firstcolor{0.642 (0.020)} & 0.583 (0.047) \\
    pca+log & 0.513 (0.032) & 0.597 (0.043) \\
    manual+log & 0.512 (0.036) & 0.612 (0.048) \\
    fpca+log & 0.519 (0.033) & 0.599 (0.047) \\
    apls+nc & 0.484 (0.065) & 0.606 (0.037) \\
    apls+log & 0.513 (0.036) & 0.605 (0.040) \\
    w\_vs\_tmean+log & 0.519 (0.033) & 0.624 (0.039) \\
    w\_vs\_mode+log & \secondcolor{0.530 (0.000)} & 0.623 (0.035) \\
    w\_vs\_median+log & \secondcolor{0.530 (0.000)} & \firstcolor{0.633 (0.037)} \\
    map\_vs\_tmean+log & 0.519 (0.033) & 0.626 (0.045) \\
    map\_vs\_mode+log & \secondcolor{0.530 (0.000)} & 0.630 (0.024) \\
    map\_vs\_median+log & \secondcolor{0.530 (0.000)} & \secondcolor{0.632 (0.037)} \\
    \bottomrule
  \end{tabular}
  \caption{Mean accuracy of classifiers (higher is better) for 10 runs with a mix of regressors coming from two different GPs and labeled according to their origin. In the first column we try to separate two heteroscedastic Brownian motions, while in the second column we discriminate between two homoscedastic Brownian motions. The corresponding standard errors are shown between brackets.}
\end{table}
\newpage
\FloatBarrier{}

\begin{table}[htbp!]
  \vspace{1em}
  \footnotesize
  \centering
  \rowcolors{2}{}{teal!8}
  \begin{tabular}{lcccc}
    \toprule
    \textbf{Classification method} & \textbf{Growth}             & \textbf{Medflies}           & \textbf{Phoneme}            \\
    \midrule
    qda & 0.581 (0.000) & 0.579 (0.028) & 0.578 (0.034) \\
    mdc & 0.694 (0.103) & 0.526 (0.024) & 0.704 (0.042) \\
    log & \firstcolor{0.961 (0.028)} & 0.575 (0.022) & \firstcolor{0.809 (0.052)} \\
    lda & 0.894 (0.054) & 0.576 (0.016) & 0.599 (0.049) \\
    fnc & 0.735 (0.117) & 0.550 (0.040) & 0.755 (0.065) \\
    flog & 0.926 (0.043) & 0.596 (0.026) & 0.785 (0.053) \\
    flda & 0.939 (0.051) & 0.550 (0.023) & 0.782 (0.046) \\
    fknn & 0.948 (0.036) & 0.539 (0.028) & \secondcolor{0.796 (0.035)} \\
    w\_pp\_tmean & 0.948 (0.041) & \firstcolor{0.611 (0.029)} & 0.790 (0.037) \\
    w\_pp\_mode & 0.935 (0.046) & 0.603 (0.033) & 0.790 (0.041) \\
    w\_pp\_median & \secondcolor{0.952 (0.036)} & \secondcolor{0.610 (0.029)} & 0.794 (0.041) \\
    map\_pp\_tmean & 0.945 (0.046) & 0.606 (0.035) & 0.791 (0.031) \\
    map\_pp\_mode & 0.935 (0.046) & 0.601 (0.040) & 0.779 (0.039) \\
    map\_pp\_median & \secondcolor{0.952 (0.036)} & 0.606 (0.038) & 0.791 (0.031) \\
    \bottomrule
    \toprule
    rkvs+log & 0.929 (0.050) & 0.589 (0.032) & \secondcolor{0.804 (0.046)} \\
    pls+nc & 0.858 (0.090) & 0.558 (0.032) & 0.776 (0.058) \\
    pls+log & 0.945 (0.032) & 0.574 (0.019) & \firstcolor{0.810 (0.043)} \\
    pca+qda & 0.955 (0.030) & 0.576 (0.025) & 0.754 (0.043) \\
    pca+log & \secondcolor{0.958 (0.035)} & 0.562 (0.028) & 0.793 (0.049) \\
    manual+log & 0.932 (0.055) & \firstcolor{0.615 (0.012)} & 0.730 (0.046) \\
    fpca+log & 0.955 (0.033) & 0.561 (0.024) & 0.769 (0.050) \\
    apls+nc & \firstcolor{0.961 (0.028)} & 0.551 (0.030) & 0.781 (0.047) \\
    apls+log & 0.952 (0.036) & 0.562 (0.015) & 0.776 (0.048) \\
    w\_vs\_tmean+log & \firstcolor{0.961 (0.028)} & \secondcolor{0.597 (0.036)} & 0.749 (0.071) \\
    w\_vs\_mode+log & 0.948 (0.046) & \secondcolor{0.597 (0.025)} & \secondcolor{0.804 (0.037)} \\
    w\_vs\_median+log & 0.952 (0.033) & \secondcolor{0.597 (0.021)} & 0.779 (0.049) \\
    map\_vs\_tmean+log & 0.945 (0.046) & 0.592 (0.047) & 0.746 (0.066) \\
    map\_vs\_mode+log & 0.948 (0.046) & 0.592 (0.036) & 0.779 (0.035) \\
    map\_vs\_median+log & 0.939 (0.047) & 0.592 (0.031) & 0.782 (0.050) \\
    \bottomrule
  \end{tabular}
  \caption{Mean accuracy of classifiers (higher is better) for 10 runs with real data sets, one on each column. The corresponding standard errors are shown between brackets.}
\end{table}
\newpage
\FloatBarrier{}

\section{Source code overview}\label{app:source-code}

The Python code developed for this work is available under a GPLv3 license at the GitHub repository \url{https://github.com/antcc/rk-bfr-jump}. The code is adequately documented and is structured in several directories as follows:

\begin{itemize}
  \item In the \texttt{rkbfr\_jump} folder we find the files responsible for the implementation of our Bayesian models, separated according to the functionality they provide. There is also a \texttt{utils} folder inside with some utility files for simulation, experimentation and visualization.
  \item The \texttt{reference\_methods} folder contains our implementation of the functional comparison algorithms that were not available through a standard Python library.
  \item The \texttt{results} folder contains plain text files with the execution times and the numerical results shown in Appendix~\ref{app:execution-times} and Appendix~\ref{app:tables}, as well as \texttt{.csv} files that facilitate working with them.
  \item At the root folder we have a Python script \texttt{experiments.py} for executing our experiments, which accepts several user-specified parameters (such as the number of iterations or the type of data set). There is also a \texttt{setup.py} file to install our method as a Python package.
\end{itemize}

When possible, the code was implemented in a generic way that would allow for easy extensions or derivations. It was also developed with efficiency in mind, so many functions and methods exploit the vectorization capabilities of the \textit{numpy} and \textit{scipy} libraries, and are sometimes parallelized using \textit{numba}. Moreover, since we followed closely the style of the \textit{scikit-learn} and \textit{scikit-fda} libraries, our methods are compatible and could be integrated (after some minor tweaking) with both of them. 

The code for the experiments was executed with a random seed set to the value 2024 for reproducibility. We provide a script file \texttt{launch.sh} that illustrates a typical execution. Lastly, there are \textit{Jupyter} notebooks that demonstrate the use of our methods in a more visual way. Inside these notebooks there is a step-by-step guide on how one might execute our algorithms, accompanied by many graphical representations, and offering the possibility of changing multiple parameters to experiment with the code. In addition, there is also a notebook that can be used to generate all the tables and figures of this document pertaining to the experimental results.


%% file: paper.bbl
\begin{thebibliography}{95}
\newcommand{\enquote}[1]{``#1''}
\expandafter\ifx\csname natexlab\endcsname\relax\def\natexlab#1{#1}\fi
\expandafter\ifx\csname url\endcsname\relax
  \def\url#1{{\tt #1}}\fi
\expandafter\ifx\csname urlprefix\endcsname\relax\def\urlprefix{url: }\fi
\expandafter\ifx\csname doiprefix\endcsname\relax\def\doiprefix{doi: }\fi
\ifx\endbibitem\undefined \let\endbibitem\relax\fi

\bibitem[{Abdi(2010)}]{abdi2010partial}
Abdi, H. (2010).
\newblock \enquote{Partial least squares regression and projection on latent structure regression (PLS Regression).}
\newblock {\em WIREs Computational Statistics\/}, 2(1): 97--106.
\newblock \doiprefix\url{https://doi.org/10.1002/wics.51}.
\endbibitem

\bibitem[{Abraham(2024)}]{abraham2024informative}
Abraham, C. (2024).
\newblock \enquote{An informative prior distribution on functions with application to functional regression.}
\newblock {\em Statistica Neerlandica\/}, 78(2): 357--373.
\newblock \doiprefix\url{https://doi.org/10.1111/stan.12322}.
\endbibitem

\bibitem[{Abraham and Grollemund(2020)}]{abraham2020posterior}
Abraham, C. and Grollemund, P.-M. (2020).
\newblock \enquote{Posterior concentration for a misspecified Bayesian regression model with functional covariates.}
\newblock {\em Journal of Statistical Planning and Inference\/}, 208: 58--65.
\newblock \doiprefix\url{https://doi.org/10.1016/j.jspi.2020.01.008}.
\endbibitem

\bibitem[{Aguilera and Aguilera-Morillo(2013)}]{aguilera2013comparative}
Aguilera, A.~M. and Aguilera-Morillo, M. (2013).
\newblock \enquote{Comparative study of different B-spline approaches for functional data.}
\newblock {\em Mathematical and Computer Modelling\/}, 58(7-8): 1568--1579.
\newblock \doiprefix\url{https://doi.org/10.1016/j.mcm.2013.04.007}.
\endbibitem

\bibitem[{Aguilera et~al.(2010)Aguilera, Escabias, Preda, and Saporta}]{aguilera2010using}
Aguilera, A.~M., Escabias, M., Preda, C., and Saporta, G. (2010).
\newblock \enquote{Using basis expansions for estimating functional PLS regression: applications with chemometric data.}
\newblock {\em Chemometrics and Intelligent Laboratory Systems\/}, 104(2): 289--305.
\newblock \doiprefix\url{https://doi.org/10.1016/j.chemolab.2010.09.007}.
\endbibitem

\bibitem[{Albert and Anderson(1984)}]{albert1984existence}
Albert, A. and Anderson, J.~A. (1984).
\newblock \enquote{On the existence of maximum likelihood estimates in logistic regression models.}
\newblock {\em Biometrika\/}, 71(1): 1--10.
\newblock \doiprefix\url{https://doi.org/10.1093/biomet/71.1.1}.
\endbibitem

\bibitem[{Aliprantis and Border(2006)}]{aliprantis2006infinite}
Aliprantis, C.~D. and Border, K.~C. (2006).
\newblock {\em Infinite Dimensional Analysis: A Hitchhiker's Guide\/}.
\newblock Springer.
\newblock \doiprefix\url{https://doi.org/10.1007/3-540-29587-9}.
\endbibitem

\bibitem[{Amewou-Atisso et~al.(2003)Amewou-Atisso, Ghosal, Ghosh, and Ramamoorthi}]{amewou2003posterior}
Amewou-Atisso, M., Ghosal, S., Ghosh, J.~K., and Ramamoorthi, R.~V. (2003).
\newblock \enquote{Posterior consistency for semi-parametric regression problems.}
\newblock {\em Bernoulli\/}, 9(2): 291--312.
\newblock \doiprefix\url{https://doi.org/10.3150/bj/1068128979}.
\endbibitem

\bibitem[{Berlinet and Thomas-Agnan(2004)}]{berlinet2004reproducing}
Berlinet, A. and Thomas-Agnan, C. (2004).
\newblock {\em Reproducing Kernel Hilbert Spaces in Probability and Statistics\/}.
\newblock Springer.
\newblock \doiprefix\url{https://doi.org/10.1007/978-1-4419-9096-9}.
\endbibitem

\bibitem[{Berrendero et~al.(2019)Berrendero, Bueno-Larraz, and Cuevas}]{berrendero2019rkhs}
Berrendero, J.~R., Bueno-Larraz, B., and Cuevas, A. (2019).
\newblock \enquote{An RKHS model for variable selection in functional linear regression.}
\newblock {\em Journal of Multivariate Analysis\/}, 170: 25--45.
\newblock \doiprefix\url{https://doi.org/10.1016/j.jmva.2018.04.008}.
\endbibitem

\bibitem[{Berrendero et~al.(2020)Berrendero, Bueno-Larraz, and Cuevas}]{berrendero2020mahalanobis}
--- (2020).
\newblock \enquote{On Mahalanobis Distance in Functional Settings.}
\newblock {\em Journal of Machine Learning Research\/}, 21(9): 1--33.
\newblock \urlprefix\url{http://jmlr.org/papers/v21/18-156.html}.
\endbibitem

\bibitem[{Berrendero et~al.(2023)Berrendero, Bueno-Larraz, and Cuevas}]{berrendero2023functional}
--- (2023).
\newblock \enquote{On functional logistic regression: some conceptual issues.}
\newblock {\em Test\/}, 32: 321--349.
\newblock \doiprefix\url{https://doi.org/10.1007/s11749-022-00836-9}.
\endbibitem

\bibitem[{Berrendero et~al.(2024)Berrendero, Cholaquidis, and Cuevas}]{berrendero2024functional}
Berrendero, J.~R., Cholaquidis, A., and Cuevas, A. (2024).
\newblock \enquote{On the functional regression model and its finite-dimensional approximations.}
\newblock {\em Statistical Papers\/}, 1--35.
\newblock \doiprefix\url{https://doi.org/10.1007/s00362-024-01567-9}.
\endbibitem

\bibitem[{Berrendero et~al.(2016)Berrendero, Cuevas, and Torrecilla}]{berrendero2016variable}
Berrendero, J.~R., Cuevas, A., and Torrecilla, J.~L. (2016).
\newblock \enquote{Variable selection in functional data classification: a maxima-hunting proposal.}
\newblock {\em Statistica Sinica\/}, 619--638.
\newblock \doiprefix\url{https://doi.org/10.5705/ss.202014.0014}.
\endbibitem

\bibitem[{Berrendero et~al.(2018)Berrendero, Cuevas, and Torrecilla}]{berrendero2018use}
--- (2018).
\newblock \enquote{On the use of reproducing kernel Hilbert spaces in functional classification.}
\newblock {\em Journal of the American Statistical Association\/}, 113(523): 1210--1218.
\newblock \doiprefix\url{https://doi.org/10.1080/01621459.2017.1320287}.
\endbibitem

\bibitem[{Borggaard and Thodberg(1992)}]{borggaard1992optimal}
Borggaard, C. and Thodberg, H.~H. (1992).
\newblock \enquote{Optimal minimal neural interpretation of spectra.}
\newblock {\em Analytical Chemistry\/}, 64(5): 545--551.
\newblock \doiprefix\url{https://doi.org/10.1021/ac00029a018}.
\endbibitem

\bibitem[{Bro(1999)}]{bro1999exploratory}
Bro, R. (1999).
\newblock \enquote{Exploratory study of sugar production using fluorescence spectroscopy and multi-way analysis.}
\newblock {\em Chemometrics and Intelligent Laboratory Systems\/}, 46(2): 133--147.
\newblock \doiprefix\url{https://doi.org/10.1016/s0169-7439(98)00181-6}.
\endbibitem

\bibitem[{Brooks et~al.(2003)Brooks, Giudici, and Roberts}]{brooks2003efficient}
Brooks, S.~P., Giudici, P., and Roberts, G.~O. (2003).
\newblock \enquote{Efficient construction of reversible jump Markov chain Monte Carlo proposal distributions.}
\newblock {\em Journal of the Royal Statistical Society Series B: Statistical Methodology\/}, 65(1): 3--39.
\newblock \doiprefix\url{https://doi.org/10.1111/1467-9868.03711}.
\endbibitem

\bibitem[{Bueno-Larraz and Klepsch(2019)}]{bueno2019variable}
Bueno-Larraz, B. and Klepsch, J. (2019).
\newblock \enquote{Variable Selection for the Prediction of \(C[0, 1]\)-Valued Autoregressive Processes using Reproducing Kernel Hilbert Spaces.}
\newblock {\em Technometrics\/}, 61(2): 139--153.
\newblock \doiprefix\url{https://doi.org/10.1080/00401706.2018.1505660}.
\endbibitem

\bibitem[{Cand{\`e}s and Sur(2020)}]{candes2020phase}
Cand{\`e}s, E.~J. and Sur, P. (2020).
\newblock \enquote{The phase transition for the existence of the maximum likelihood estimate in high-dimensional logistic regression.}
\newblock {\em The Annals of Statistics\/}, 48(1): 27 -- 42.
\newblock \doiprefix\url{https://doi.org/10.1214/18-AOS1789}.
\endbibitem

\bibitem[{Cardot and Sarda(2018)}]{cardot2011functional}
Cardot, H. and Sarda, P. (2018).
\newblock \enquote{Functional Linear Regression.}
\newblock In Ferraty, F. and Romain, Y. (eds.), {\em The Oxford Handbook of Functional Data Analysis\/}, 21--46. Oxford Handbooks.
\newblock \doiprefix\url{https://doi.org/10.1093/oxfordhb/9780199568444.013.2}.
\endbibitem

\bibitem[{Carey et~al.(1998)Carey, Liedo, M\"{u}ller, Wang, and Chiou}]{carey1998relationship}
Carey, J.~R., Liedo, P., M\"{u}ller, H.-G., Wang, J.-L., and Chiou, J.-M. (1998).
\newblock \enquote{Relationship of Age Patterns of Fecundity to Mortality, Longevity, and Lifetime Reproduction in a Large Cohort of Mediterranean Fruit Fly Females.}
\newblock {\em The Journals of Gerontology: Series A\/}, 53(4): B245--B251.
\newblock \doiprefix\url{https://doi.org/10.1093/gerona/53a.4.b245}.
\endbibitem

\bibitem[{Carlin and Chib(1995)}]{carlin1995bayesian}
Carlin, B.~P. and Chib, S. (1995).
\newblock \enquote{Bayesian model choice via Markov chain Monte Carlo methods.}
\newblock {\em Journal Of The Royal Statistical Society Series B: Statistical Methodology\/}, 57(3): 473--484.
\newblock \doiprefix\url{https://doi.org/10.1111/j.2517-6161.1995.tb02042.x}.
\endbibitem

\bibitem[{Celeux et~al.(2000)Celeux, Hurn, and Robert}]{celeux2000computational}
Celeux, G., Hurn, M., and Robert, C.~P. (2000).
\newblock \enquote{Computational and Inferential Difficulties with Mixture Posterior Distributions.}
\newblock {\em Journal of the American Statistical Association\/}, 95(451): 957--970.
\newblock \doiprefix\url{https://doi.org/10.1080/01621459.2000.10474285}.
\endbibitem

\bibitem[{Choi and Ramamoorthi(2008)}]{choi2008remarks}
Choi, T. and Ramamoorthi, R.~V. (2008).
\newblock \enquote{Remarks on consistency of posterior distributions.}
\newblock In Clarke, B. and Ghosal, S. (eds.), {\em Pushing the limits of contemporary statistics: contributions in honor of Jayanta K. Ghosh\/}, 170--186. Institute of Mathematical Statistics Collections.
\newblock \doiprefix\url{https://doi.org/10.1214/074921708000000138}.
\endbibitem

\bibitem[{Choi and Schervish(2007)}]{choi2007posterior}
Choi, T. and Schervish, M.~J. (2007).
\newblock \enquote{On posterior consistency in nonparametric regression problems.}
\newblock {\em Journal of Multivariate Analysis\/}, 98(10): 1969--1987.
\newblock \doiprefix\url{https://doi.org/10.1016/j.jmva.2007.01.004}.
\endbibitem

\bibitem[{Crainiceanu and Goldsmith(2010)}]{crainiceanu2010bayesian}
Crainiceanu, C.~M. and Goldsmith, A.~J. (2010).
\newblock \enquote{Bayesian Functional Data Analysis Using WinBUGS.}
\newblock {\em Journal of Statistical Software\/}, 32(11): 1--33.
\newblock \doiprefix\url{https://doi.org/10.18637/jss.v032.i11}.
\endbibitem

\bibitem[{Cuevas(2014)}]{cuevas2014partial}
Cuevas, A. (2014).
\newblock \enquote{A partial overview of the theory of statistics with functional data.}
\newblock {\em Journal of Statistical Planning and Inference\/}, 147: 1--23.
\newblock \doiprefix\url{https://doi.org/10.1016/j.jspi.2013.04.002}.
\endbibitem

\bibitem[{Cuevas et~al.(2004)Cuevas, Febrero, and Fraiman}]{cuevas2004anova}
Cuevas, A., Febrero, M., and Fraiman, R. (2004).
\newblock \enquote{An anova test for functional data.}
\newblock {\em Computational Statistics \& Data Analysis\/}, 47(1): 111--122.
\newblock \doiprefix\url{https://doi.org/10.1016/j.csda.2003.10.021}.
\endbibitem

\bibitem[{Davies et~al.(2023)Davies, Salomone, Sutton, and Drovandi}]{davies2023transport}
Davies, L., Salomone, R., Sutton, M., and Drovandi, C. (2023).
\newblock \enquote{Transport Reversible Jump Proposals.}
\newblock In Ruiz, F., Dy, J., and van~de Meent, J.-W. (eds.), {\em Proceedings of The 26th International Conference on Artificial Intelligence and Statistics\/}, volume 206 of {\em Proceedings of Machine Learning Research\/}, 6839--6852. PMLR.
\newblock \urlprefix\url{https://proceedings.mlr.press/v206/davies23a.html}.
\endbibitem

\bibitem[{Delaigle and Hall(2012{\natexlab{a}})}]{delaigle2012achieving}
Delaigle, A. and Hall, P. (2012{\natexlab{a}}).
\newblock \enquote{Achieving near Perfect Classification for Functional Data.}
\newblock {\em Journal of the Royal Statistical Society Series B: Statistical Methodology\/}, 74(2): 267--286.
\newblock \doiprefix\url{https://doi.org/10.1111/j.1467-9868.2011.01003.x}.
\endbibitem

\bibitem[{Delaigle and Hall(2012{\natexlab{b}})}]{delaigle2012methodology}
--- (2012{\natexlab{b}}).
\newblock \enquote{Methodology and theory for partial least squares applied to functional data.}
\newblock {\em The Annals of Statistics\/}, 40(1): 322--352.
\newblock \doiprefix\url{https://doi.org/10.1214/11-aos958}.
\endbibitem

\bibitem[{Diaconis and Freedman(1986)}]{diaconis1986consistency}
Diaconis, P. and Freedman, D. (1986).
\newblock \enquote{On the Consistency of Bayes Estimates.}
\newblock {\em The Annals of Statistics\/}, 14(1): 1--26.
\newblock \doiprefix\url{https://doi.org/10.1214/aos/1176349830}.
\endbibitem

\bibitem[{Doob(1949)}]{doob1949application}
Doob, J.~L. (1949).
\newblock \enquote{Application of the theory of martingales.}
\newblock {\em Le calcul des probabilités et ses applications. Colloques Internationaux\/}, 13: 23--27.
\newblock \lowercase{url}: \url{https://www.jehps.net/juin2009/Locker.pdf} [at the end].
\endbibitem

\bibitem[{Dudley(2002)}]{dudley2002real}
Dudley, R.~M. (2002).
\newblock {\em Real Analysis and Probability\/}.
\newblock Cambridge University Press.
\newblock \doiprefix\url{https://doi.org/10.1017/CBO9780511755347}.
\endbibitem

\bibitem[{Ferguson(1974)}]{ferguson1974prior}
Ferguson, T.~S. (1974).
\newblock \enquote{Prior Distributions on Spaces of Probability Measures.}
\newblock {\em The Annals of Statistics\/}, 2(4): 615 -- 629.
\newblock \doiprefix\url{https://doi.org/10.1214/aos/1176342752}.
\endbibitem

\bibitem[{Ferraty et~al.(2010)Ferraty, Hall, and Vieu}]{ferraty2010most}
Ferraty, F., Hall, P., and Vieu, P. (2010).
\newblock \enquote{Most-predictive design points for functional data predictors.}
\newblock {\em Biometrika\/}, 97(4): 807--824.
\newblock \doiprefix\url{https://doi.org/10.1093/biomet/asq058}.
\endbibitem

\bibitem[{Folland(1999)}]{folland1999real}
Folland, G.~B. (1999).
\newblock {\em Real Analysis: Modern Techniques and Their Applications\/}.
\newblock John Wiley \& Sons.
\endbibitem

\bibitem[{Foreman-Mackey et~al.(2013)Foreman-Mackey, Hogg, Lang, and Goodman}]{foreman2013emcee}
Foreman-Mackey, D., Hogg, D.~W., Lang, D., and Goodman, J. (2013).
\newblock \enquote{emcee: the MCMC hammer.}
\newblock {\em Publications of the Astronomical Society of the Pacific\/}, 125(925): 306--312.
\newblock \doiprefix\url{https://doi.org/10.1086/670067}.
\endbibitem

\bibitem[{Galeano et~al.(2015)Galeano, Joseph, and Lillo}]{galeano2015mahalanobis}
Galeano, P., Joseph, E., and Lillo, R.~E. (2015).
\newblock \enquote{The Mahalanobis Distance for Functional Data With Applications to Classification.}
\newblock {\em Technometrics\/}, 57(2): 281--291.
\newblock \doiprefix\url{https://doi.org/10.1080/00401706.2014.902774}.
\endbibitem

\bibitem[{Gelman et~al.(2008)Gelman, Jakulin, Pittau, and Su}]{gelman2008weakly}
Gelman, A., Jakulin, A., Pittau, M.~G., and Su, Y.-S. (2008).
\newblock \enquote{A weakly informative default prior distribution for logistic and other regression models.}
\newblock {\em The Annals of Applied Statistics\/}, 2(4): 1360--1383.
\newblock \doiprefix\url{https://doi.org/10.1214/08-AOAS191}.
\endbibitem

\bibitem[{Gelman and Rubin(1992)}]{gelman1992inference}
Gelman, A. and Rubin, D.~B. (1992).
\newblock \enquote{Inference from iterative simulation using multiple sequences.}
\newblock {\em Statistical Science\/}, 7(4): 457--472.
\newblock \doiprefix\url{https://doi.org/10.1214/ss/1177011136}.
\endbibitem

\bibitem[{Ghosal and van~der Vaart(2017)}]{ghosal2017fundamentals}
Ghosal, S. and van~der Vaart, A.~W. (2017).
\newblock {\em Fundamentals of Nonparametric Bayesian Inference\/}.
\newblock Cambridge University Press.
\newblock \doiprefix\url{https://doi.org/10.1017/9781139029834}.
\endbibitem

\bibitem[{Ghosh and Chaudhuri(2005)}]{ghosh2005maximum}
Ghosh, A.~K. and Chaudhuri, P. (2005).
\newblock \enquote{On Maximum Depth and Related Classifiers.}
\newblock {\em Scandinavian Journal of Statistics\/}, 32(2): 327--350.
\newblock \doiprefix\url{https://doi.org/10.1111/j.1467-9469.2005.00423.x}.
\endbibitem

\bibitem[{Ghosh et~al.(2018)Ghosh, Li, and Mitra}]{ghosh2018use}
Ghosh, J., Li, Y., and Mitra, R. (2018).
\newblock \enquote{On the use of Cauchy prior distributions for Bayesian logistic regression.}
\newblock {\em Bayesian Analysis\/}, 13(2): 359--383.
\newblock \doiprefix\url{https://doi.org/10.1214/17-BA1051}.
\endbibitem

\bibitem[{Goia and Vieu(2016)}]{goia2016introduction}
Goia, A. and Vieu, P. (2016).
\newblock \enquote{An introduction to recent advances in high/infinite dimensional statistics.}
\newblock {\em Journal of Multivariate Analysis\/}, 146: 1--6.
\newblock \doiprefix\url{https://doi.org/10.1016/j.jmva.2015.12.001}.
\endbibitem

\bibitem[{Goodman and Weare(2010)}]{goodman2010ensemble}
Goodman, J. and Weare, J. (2010).
\newblock \enquote{Ensemble samplers with affine invariance.}
\newblock {\em Communications in Applied Mathematics and Computational Science\/}, 5(1): 65--80.
\newblock \doiprefix\url{https://doi.org/10.2140/camcos.2010.5.65}.
\endbibitem

\bibitem[{Green(1995)}]{green1995reversible}
Green, P.~J. (1995).
\newblock \enquote{Reversible jump Markov chain Monte Carlo computation and Bayesian model determination.}
\newblock {\em Biometrika\/}, 82(4): 711--732.
\newblock \doiprefix\url{https://doi.org/10.2307/2337340}.
\endbibitem

\bibitem[{Grollemund et~al.(2019)Grollemund, Abraham, Baragatti, and Pudlo}]{grollemund2019bayesian}
Grollemund, P.-M., Abraham, C., Baragatti, M., and Pudlo, P. (2019).
\newblock \enquote{Bayesian Functional Linear Regression with Sparse Step Functions.}
\newblock {\em Bayesian Analysis\/}, 14(1): 111 -- 135.
\newblock \doiprefix\url{https://doi.org/10.1214/18-BA1095}.
\endbibitem

\bibitem[{Hastie et~al.(1995)Hastie, Buja, and Tibshirani}]{hastie1995penalized}
Hastie, T., Buja, A., and Tibshirani, R. (1995).
\newblock \enquote{Penalized discriminant analysis.}
\newblock {\em The Annals of Statistics\/}, 23(1): 73--102.
\newblock \doiprefix\url{https://doi.org/10.1214/aos/1176324456}.
\endbibitem

\bibitem[{Hoeting et~al.(1999)Hoeting, Madigan, Raftery, and Volinsky}]{hoeting1999bayesian}
Hoeting, J.~A., Madigan, D., Raftery, A.~E., and Volinsky, C.~T. (1999).
\newblock \enquote{Bayesian model averaging: a tutorial (with comments by M. Clyde, David Draper and E. I. George, and a rejoinder by the authors).}
\newblock {\em Statistical Science\/}, 14(4): 382--417.
\newblock \doiprefix\url{https://doi.org/10.1214/ss/1009212519}.
\endbibitem

\bibitem[{Horv{\'a}th and Kokoszka(2012)}]{horvath2012inference}
Horv{\'a}th, L. and Kokoszka, P. (2012).
\newblock {\em Inference for Functional Data with Applications\/}.
\newblock Springer.
\newblock \doiprefix\url{https://doi.org/10.1007/978-1-4614-3655-3}.
\endbibitem

\bibitem[{Hsing and Eubank(2015)}]{hsing2015theoretical}
Hsing, T. and Eubank, R. (2015).
\newblock {\em Theoretical Foundations of Functional Data Analysis, with an Introduction to Linear Operators\/}.
\newblock John Wiley \& Sons.
\newblock \doiprefix\url{https://doi.org/10.1002/9781118762547}.
\endbibitem

\bibitem[{Hukushima and Nemoto(1996)}]{hukushima1996exchange}
Hukushima, K. and Nemoto, K. (1996).
\newblock \enquote{Exchange Monte Carlo method and application to spin glass simulations.}
\newblock {\em Journal of the Physical Society of Japan\/}, 65(6): 1604--1608.
\newblock \doiprefix\url{https://doi.org/10.1143/JPSJ.65.1604}.
\endbibitem

\bibitem[{Jasra et~al.(2005)Jasra, Holmes, and Stephens}]{jasra2005markov}
Jasra, A., Holmes, C.~C., and Stephens, D.~A. (2005).
\newblock \enquote{Markov Chain Monte Carlo Methods and the Label Switching Problem in Bayesian Mixture Modeling.}
\newblock {\em Statistical Science\/}, 20(1): 50--67.
\newblock \doiprefix\url{https://doi.org/10.1214/088342305000000016}.
\endbibitem

\bibitem[{Jeffreys(1946)}]{jeffreys1946invariant}
Jeffreys, H. (1946).
\newblock \enquote{An invariant form for the prior probability in estimation problems.}
\newblock {\em Proceedings of the Royal Society of London Series A: Mathematical and Physical Sciences\/}, 186(1007): 453--461.
\newblock \doiprefix\url{https://doi.org/10.1098/rspa.1946.0056}.
\endbibitem

\bibitem[{Kalivas(1997)}]{kalivas1997two}
Kalivas, J.~H. (1997).
\newblock \enquote{Two data sets of near infrared spectra.}
\newblock {\em Chemometrics and Intelligent Laboratory Systems\/}, 37(2): 255--259.
\newblock \doiprefix\url{https://doi.org/10.1016/s0169-7439(97)00038-5}.
\endbibitem

\bibitem[{Karnesis et~al.(2023)Karnesis, Katz, Korsakova, Gair, and Stergioulas}]{karnesis2023eryn}
Karnesis, N., Katz, M.~L., Korsakova, N., Gair, J.~R., and Stergioulas, N. (2023).
\newblock \enquote{Eryn: a multipurpose sampler for Bayesian inference.}
\newblock {\em Monthly Notices of the Royal Astronomical Society\/}, 526(4): 4814--4830.
\newblock \doiprefix\url{https://doi.org/10.1093/mnras/stad2939}.
\endbibitem

\bibitem[{Kneip et~al.(2016)Kneip, Po{\ss}, and Sarda}]{kneip2016functional}
Kneip, A., Po{\ss}, D., and Sarda, P. (2016).
\newblock \enquote{Functional linear regression with points of impact.}
\newblock {\em The Annals of Statistics\/}, 44(1): 1--30.
\newblock \doiprefix\url{https://doi.org/10.1214/15-AOS1323}.
\endbibitem

\bibitem[{Korsakova et~al.(2024)Korsakova, Babak, Katz, Karnesis, Khukhlaev, and Gair}]{korsakova2024neural}
Korsakova, N., Babak, S., Katz, M.~L., Karnesis, N., Khukhlaev, S., and Gair, J.~R. (2024).
\newblock \enquote{Neural density estimation for Galactic binaries in the LISA data analysis.}
\newblock {\em Physical Review D\/}, 110: 104069.
\newblock \doiprefix\url{https://doi.org/10.1103/PhysRevD.110.104069}.
\endbibitem

\bibitem[{Kupresanin et~al.(2010)Kupresanin, Shin, King, and Eubank}]{kupresanin2010rkhs}
Kupresanin, A., Shin, H., King, D., and Eubank, R. (2010).
\newblock \enquote{An RKHS framework for functional data analysis.}
\newblock {\em Journal of Statistical Planning and Inference\/}, 140(12): 3627--3637.
\newblock \doiprefix\url{https://doi.org/10.1016/j.jspi.2010.04.030}.
\endbibitem

\bibitem[{Lian et~al.(2016)Lian, Choi, Meng, and Jo}]{lian2016posterior}
Lian, H., Choi, T., Meng, J., and Jo, S. (2016).
\newblock \enquote{Posterior convergence for Bayesian functional linear regression.}
\newblock {\em Journal of Multivariate Analysis\/}, 150: 27--41.
\newblock \doiprefix\url{https://doi.org/10.1016/j.jmva.2016.04.008}.
\endbibitem

\bibitem[{Lindquist and McKeague(2009)}]{lindquist2009logistic}
Lindquist, M.~A. and McKeague, I.~W. (2009).
\newblock \enquote{Logistic regression with Brownian-like predictors.}
\newblock {\em Journal of the American Statistical Association\/}, 104(488): 1575--1585.
\newblock \doiprefix\url{https://doi.org/10.1198/jasa.2009.tm08496}.
\endbibitem

\bibitem[{Lo\`{e}ve(1948)}]{loeve1948fonctions}
Lo\`{e}ve, M. (1948).
\newblock \enquote{Fonctions al\'{e}atoires du second ordre.}
\newblock In L\'{e}vy, P. (ed.), {\em Processus stochastiques et mouvement Brownien\/}, 299--352. Gauthier-Villars.
\endbibitem

\bibitem[{L\'{o}pez-Pintado and Romo(2009)}]{lopez2009concept}
L\'{o}pez-Pintado, S. and Romo, J. (2009).
\newblock \enquote{On the Concept of Depth for Functional Data.}
\newblock {\em Journal of the American Statistical Association\/}, 104(486): 718--734.
\newblock \doiprefix\url{https://doi.org/10.1198/jasa.2009.0108}.
\endbibitem

\bibitem[{Luki\'{c} and Beder(2001)}]{lukic2001stochastic}
Luki\'{c}, M.~N. and Beder, J.~H. (2001).
\newblock \enquote{Stochastic processes with sample paths in reproducing kernel Hilbert spaces.}
\newblock {\em Transactions of the American Mathematical Society\/}, 353(10): 3945--3969.
\newblock \doiprefix\url{https://doi.org/10.1090/s0002-9947-01-02852-5}.
\endbibitem

\bibitem[{Miller(2023)}]{miller2023consistency}
Miller, J.~W. (2023).
\newblock \enquote{Consistency of mixture models with a prior on the number of components.}
\newblock {\em Dependence Modeling\/}, 11(1): 20220150.
\newblock \doiprefix\url{https://doi.org/10.1515/demo-2022-0150}.
\endbibitem

\bibitem[{Miller and Harrison(2018)}]{miller2018mixture}
Miller, J.~W. and Harrison, M.~T. (2018).
\newblock \enquote{Mixture Models with a Prior on the Number of Components.}
\newblock {\em Journal of the American Statistical Association\/}, 113(521): 340--356.
\newblock \doiprefix\url{https://doi.org/10.1080/01621459.2016.1255636}.
\endbibitem

\bibitem[{M\"{u}ller and Stadtm\"{u}ller(2005)}]{muller2005generalized}
M\"{u}ller, H.-G. and Stadtm\"{u}ller, U. (2005).
\newblock \enquote{Generalized functional linear models.}
\newblock {\em The Annals of Statistics\/}, 33(2): 774--805.
\newblock \doiprefix\url{https://doi.org/10.1214/009053604000001156}.
\endbibitem

\bibitem[{Nobile(1994)}]{nobile1994bayesian}
Nobile, A. (1994).
\newblock \enquote{Bayesian analysis of finite mixture distributions.}
\newblock Ph.D. thesis, Carnegie Mellon University.
\endbibitem

\bibitem[{Nobile and Fearnside(2007)}]{nobile2007bayesian}
Nobile, A. and Fearnside, A.~T. (2007).
\newblock \enquote{Bayesian finite mixtures with an unknown number of components: The allocation sampler.}
\newblock {\em Statistics and Computing\/}, 17: 147--162.
\newblock \doiprefix\url{https://doi.org/10.1007/s11222-006-9014-7}.
\endbibitem

\bibitem[{Pardo(2018)}]{pardo2018statistical}
Pardo, L. (2018).
\newblock {\em Statistical Inference Based on Divergence Measures\/}.
\newblock Chapman and Hall/CRC.
\newblock \doiprefix\url{https://doi.org/10.1201/9781420034813}.
\endbibitem

\bibitem[{Parzen(1961)}]{parzen1961approach}
Parzen, E. (1961).
\newblock \enquote{An Approach to Time Series Analysis.}
\newblock {\em The Annals of Mathematical Statistics\/}, 32(4): 951--989.
\newblock \doiprefix\url{https://doi.org/10.1214/aoms/1177704840}.
\endbibitem

\bibitem[{Pedregosa et~al.(2011)Pedregosa, Varoquaux, Gramfort, Michel, Thirion, Grisel, Blondel, Prettenhofer, Weiss, Dubourg, Vanderplas, Passos, Cournapeau, Brucher, Perrot, and Duchesnay}]{pedregosa2011scikit}
Pedregosa, F., Varoquaux, G., Gramfort, A., Michel, V., Thirion, B., Grisel, O., Blondel, M., Prettenhofer, P., Weiss, R., Dubourg, V., Vanderplas, J., Passos, A., Cournapeau, D., Brucher, M., Perrot, M., and Duchesnay, {\'E}. (2011).
\newblock \enquote{Scikit-learn: Machine Learning in Python.}
\newblock {\em Journal of Machine Learning Research\/}, 12(85): 2825--2830.
\newblock \urlprefix\url{http://jmlr.org/papers/v12/pedregosa11a.html}.
\endbibitem

\bibitem[{Po\ss{} et~al.(2020)Po\ss{}, Liebl, Kneip, Eisenbarth, Wager, and Barrett}]{poss2020superconsistent}
Po\ss{}, D., Liebl, D., Kneip, A., Eisenbarth, H., Wager, T.~D., and Barrett, L.~F. (2020).
\newblock \enquote{Superconsistent Estimation of Points of Impact in Non-Parametric Regression with Functional Predictors.}
\newblock {\em Journal of the Royal Statistical Society Series B: Statistical Methodology\/}, 82(4): 1115--1140.
\newblock \doiprefix\url{https://doi.org/10.1111/rssb.12386}.
\endbibitem

\bibitem[{Preda et~al.(2007)Preda, Saporta, and L{\'e}v{\'e}der}]{preda2007pls}
Preda, C., Saporta, G., and L{\'e}v{\'e}der, C. (2007).
\newblock \enquote{PLS classification of functional data.}
\newblock {\em Computational Statistics\/}, 22(2): 223--235.
\newblock \doiprefix\url{https://doi.org/10.1007/s00180-007-0041-4}.
\endbibitem

\bibitem[{Ramos-Carreño et~al.(2024)Ramos-Carreño, Torrecilla, Carbajo-Berrocal, Marcos, and Suárez}]{ramos2024scikit}
Ramos-Carreño, C., Torrecilla, J.~L., Carbajo-Berrocal, M., Marcos, P., and Suárez, A. (2024).
\newblock \enquote{scikit-fda: A Python Package for Functional Data Analysis.}
\newblock {\em Journal of Statistical Software\/}, 109(2): 1--37.
\newblock \doiprefix\url{https://doi.org/10.18637/jss.v109.i02}.
\endbibitem

\bibitem[{Ramsay and Silverman(2005)}]{ramsay2005functional}
Ramsay, J.~O. and Silverman, B.~W. (2005).
\newblock {\em Functional Data Analysis\/}.
\newblock Springer.
\newblock \doiprefix\url{https://doi.org/10.1007/b98888}.
\endbibitem

\bibitem[{Reiss et~al.(2017)Reiss, Goldsmith, Shang, and Ogden}]{reiss2017methods}
Reiss, P.~T., Goldsmith, J., Shang, H.~L., and Ogden, R.~T. (2017).
\newblock \enquote{Methods for Scalar-on-Function Regression.}
\newblock {\em International Statistical Review\/}, 85(2): 228--249.
\newblock \doiprefix\url{https://doi.org/10.1111/insr.12163}.
\endbibitem

\bibitem[{Richardson and Green(1997)}]{richardson1997bayesian}
Richardson, S. and Green, P.~J. (1997).
\newblock \enquote{On Bayesian analysis of mixtures with an unknown number of components (with discussion).}
\newblock {\em Journal of the Royal Statistical Society Series B: Statistical Methodology\/}, 59(4): 731--792.
\newblock \doiprefix\url{https://doi.org/10.1111/1467-9868.00095}.
\endbibitem

\bibitem[{Robert(2014)}]{robert2014jeffreys}
Robert, C.~P. (2014).
\newblock \enquote{On the Jeffreys-Lindley paradox.}
\newblock {\em Philosophy of Science\/}, 81(2): 216--232.
\newblock \doiprefix\url{https://doi.org/10.1086/675729}.
\endbibitem

\bibitem[{Rodr{\'\i}guez and Walker(2014)}]{rodriguez2014label}
Rodr{\'\i}guez, C.~E. and Walker, S.~G. (2014).
\newblock \enquote{Label Switching in Bayesian Mixture Models: Deterministic Relabeling Strategies.}
\newblock {\em Journal of Computational and Graphical Statistics\/}, 23(1): 25--45.
\newblock \doiprefix\url{https://doi.org/10.1080/10618600.2012.735624}.
\endbibitem

\bibitem[{Roodaki et~al.(2014)Roodaki, Bect, and Fleury}]{roodaki2014relabeling}
Roodaki, A., Bect, J., and Fleury, G. (2014).
\newblock \enquote{Relabeling and Summarizing Posterior Distributions in Signal Decomposition Problems When the Number of Components is Unknown.}
\newblock {\em IEEE Transactions On Signal Processing\/}, 62(16): 4091--4104.
\newblock \doiprefix\url{https://doi.org/10.1109/TSP.2014.2333569}.
\endbibitem

\bibitem[{Rosenthal(2011)}]{rosenthal2011optimal}
Rosenthal, J.~S. (2011).
\newblock \enquote{Optimal proposal distributions and adaptive MCMC.}
\newblock In Brooks, S., Gelman, A., Jones, G., and Meng, X.-L. (eds.), {\em Handbook of Markov Chain Monte Carlo\/}, 93--111. Chapman \& Hall/CRC.
\newblock \doiprefix\url{https://doi.org/10.1201/b10905}.
\endbibitem

\bibitem[{Schwartz(1965)}]{schwartz1965bayes}
Schwartz, L. (1965).
\newblock \enquote{On Bayes procedures.}
\newblock {\em Zeitschrift f\"{u}r Wahrscheinlichkeitstheorie und Verwandte Gebiete\/}, 4(1): 10--26.
\newblock \doiprefix\url{https://doi.org/10.1007/bf00535479}.
\endbibitem

\bibitem[{Shi and Choi(2011)}]{shi2011gaussian}
Shi, J.~Q. and Choi, T. (2011).
\newblock {\em Gaussian Process Regression Analysis for Functional Data\/}.
\newblock Chapman and Hall/CRC.
\newblock \doiprefix\url{https://doi.org/10.1201/b11038}.
\endbibitem

\bibitem[{Shin(2008)}]{shin2008extension}
Shin, H. (2008).
\newblock \enquote{An extension of Fisher's discriminant analysis for stochastic processes.}
\newblock {\em Journal of Multivariate Analysis\/}, 99(6): 1191--1216.
\newblock \doiprefix\url{https://doi.org/10.1016/j.jmva.2007.08.001}.
\endbibitem

\bibitem[{Simola et~al.(2021)Simola, Cisewski-Kehe, and Wolpert}]{simola2021approximate}
Simola, U., Cisewski-Kehe, J., and Wolpert, R.~L. (2021).
\newblock \enquote{Approximate Bayesian computation for finite mixture models.}
\newblock {\em Journal of Statistical Computation and Simulation\/}, 91(6): 1155--1174.
\newblock \doiprefix\url{https://doi.org/10.1080/00949655.2020.1843169}.
\endbibitem

\bibitem[{Sperrin et~al.(2010)Sperrin, Jaki, and Wit}]{sperrin2010probabilistic}
Sperrin, M., Jaki, T., and Wit, E. (2010).
\newblock \enquote{Probabilistic relabelling strategies for the label switching problem in Bayesian mixture models.}
\newblock {\em Statistics and Computing\/}, 20: 357--366.
\newblock \doiprefix\url{https://doi.org/10.1007/s11222-009-9129-8}.
\endbibitem

\bibitem[{Stephens(2000)}]{stephens2000dealing}
Stephens, M. (2000).
\newblock \enquote{Dealing With Label Switching in Mixture Models.}
\newblock {\em Journal of the Royal Statistical Society Series B: Statistical Methodology\/}, 62(4): 795--809.
\newblock \doiprefix\url{https://doi.org/10.1111/1467-9868.00265}.
\endbibitem

\bibitem[{Torrecilla et~al.(2020)Torrecilla, Ramos-Carre\~{n}o, S\'{a}nchez-Monta\~{n}\'{e}s, and Su\'{a}rez}]{torrecilla2020optimal}
Torrecilla, J.~L., Ramos-Carre\~{n}o, C., S\'{a}nchez-Monta\~{n}\'{e}s, M., and Su\'{a}rez, A. (2020).
\newblock \enquote{Optimal classification of Gaussian processes in homo- and heteroscedastic settings.}
\newblock {\em Statistics and Computing\/}, 30(4): 1091--1111.
\newblock \doiprefix\url{https://doi.org/10.1007/s11222-020-09937-7}.
\endbibitem

\bibitem[{Tuddenham and Snyder(1954)}]{tuddenham1954physical}
Tuddenham, R.~D. and Snyder, M.~M. (1954).
\newblock \enquote{Physical growth of California boys and girls from birth to eighteen years.}
\newblock {\em University of California Publications in Child Development\/}, 1(2): 183--364.
\newblock \urlprefix\url{https://pubmed.ncbi.nlm.nih.gov/13217130/}.
\endbibitem

\bibitem[{Ullah and Finch(2013)}]{ullah2013applications}
Ullah, S. and Finch, C.~F. (2013).
\newblock \enquote{Applications of functional data analysis: A systematic review.}
\newblock {\em BMC Medical Research Methodology\/}, 13:43: 1--12.
\newblock \doiprefix\url{https://doi.org/10.1186/1471-2288-13-43}.
\endbibitem

\bibitem[{van~der Vaart and Wellner(2023)}]{van1996weak}
van~der Vaart, A.~W. and Wellner, J.~A. (2023).
\newblock {\em Weak Convergence and Empirical Processes\/}.
\newblock Springer.
\newblock \doiprefix\url{https://doi.org/10.1007/978-3-031-29040-4}.
\endbibitem

\bibitem[{Yuan and Cai(2010)}]{yuan2010reproducing}
Yuan, M. and Cai, T.~T. (2010).
\newblock \enquote{A reproducing kernel Hilbert space approach to functional linear regression.}
\newblock {\em The Annals of Statistics\/}, 38(6): 3412--3444.
\newblock \doiprefix\url{https://doi.org/10.1214/09-AOS772}.
\endbibitem

\end{thebibliography}
